\newif\ifsocg       % for the 500-line socg version
\socgfalse

\newif\iflipics    % for the lipcs style file
\ifsocg
    \lipicstrue   
\else
    \lipicsfalse
\fi

\iflipics   % lipics style for socg
\documentclass[english]{socg-lipics-v2021}

\EventEditors{Oswin Aichholzer and Haitao Wang} \EventNoEds{2} \EventLongTitle{41st International Symposium on Computational Geometry (SoCG 2025)} \EventShortTitle{SoCG 2025} \EventAcronym{SoCG} \EventYear{2025} \EventDate{June 23--27, 2025} \EventLocation{Kanazawa, Japan} \EventLogo{socg-logo.pdf} \SeriesVolume{332}
\ArticleNo{XX}

\usepackage{thmtools}
\usepackage{amsmath,amssymb}
\usepackage{mathtools}
\usepackage{amsthm}
\usepackage{enumerate}
\usepackage{graphicx}
\usepackage{xspace}
\usepackage{hyperref}
\usepackage{textcomp}
\usepackage{subcaption}
\usepackage{lineno}
\usepackage{soul}

\usepackage{cite}

\usepackage[dvipsnames]{xcolor}

\newcommand{\mdb}[1]{}
\newcommand{\mdbred}[1]{}
\newcommand{\mst}[1]{}
\newcommand{\ajs}[1]{}
\newcommand{\ben}[1]{}

\renewcommand{\leq}{\leqslant}
\renewcommand{\geq}{\geqslant}
\theoremstyle{definition}
\newtheorem{defn}[theorem]{Definition}

\newenvironment{myquote}%
  {\list{}{\leftmargin=4mm\rightmargin=4mm}\item[]}%
  {\endlist}

\newcommand{\Reals}{\mathbb{R}}

\DeclareMathOperator{\sign}{sign}
\newcommand{\bd}{\partial}
\newcommand{\out}{\mathrm{out}}

\newcommand{\graph}{\ensuremath{\mathcal{G}}}

\newcommand{\A}{\ensuremath{\mathcal{A}}}
\newcommand{\V}{\ensuremath{\mathcal{V}}}
\newcommand{\arr}[1]{\A(L_{#1})}
\newcommand{\B}{\ensuremath{\mathcal{B}}}

\newcommand{\E}{\ensuremath{\mathcal{E}}}
\newcommand{\F}{\ensuremath{\mathcal{F}}}

\newcommand{\G}{\ensuremath{\mathcal{G}}}

\newcommand{\R}{\ensuremath{R}}

\newcommand{\calP}{\ensuremath{\mathcal{P}}}

\newcommand{\eps}{\varepsilon}
\newcommand{\etal}{\emph{et al.}\xspace}

\newcommand{\fre}{\F} % free space of a single square in \pol
\DeclareMathOperator{\interior}{int}

\newcommand{\bolds}{\boldsymbol{s}}
\newcommand{\boldt}{\boldsymbol{t}}
\newcommand{\BF}{\boldsymbol{F}}
\newcommand{\plan}{\boldsymbol{\pi}}
\newcommand{\stplan}{$(\bolds,\boldt)$-plan\xspace}

\newcommand{\rect}{R}
\newcommand{\Q}{Q}

\newcommand{\topR}{\mathrm{\textsf{top}}(\R)}
\newcommand{\botR}{\mathrm{\textsf{bot}}(\R)}
\newcommand{\topRsq}{\mathrm{\textsf{top}}(\Rsq)}
\newcommand{\botRsq}{\mathrm{\textsf{bot}}(\Rsq)}
\newcommand{\botQ}{\mathrm{\textsf{bot}}(\Q)}
\newcommand{\topQ}{\mathrm{\textsf{top}}(\Q)}
\newcommand{\topr}{\mathop{\mathrm{\textsf{top}}(\rho)}}
\newcommand{\botr}{\mathop{\mathrm{\textsf{bot}}(\rho)}}
\newcommand{\topin}[1]{\mathop{\mathrm{\textsf{top}}(#1)}}
\newcommand{\botin}[1]{\mathop{\mathrm{\textsf{bot}}(#1)}}
\newcommand{\topx}[1]{\mathop{\mathrm{\textsf{top}}(#1)}}

\newcommand{\height}{\mathop{\mathrm{\textsf{ht}}}}

\newcommand{\push}{\textsc{Push}\xspace}

\usepackage{color}
\definecolor{darkorchid}{rgb}{0.6,0.196,0.8}

\newcommand{\camera}[1]{}

\newcommand{\RR}{\mathbb{R}}

\newcommand{\bt}{\boldsymbol{t}}
\newcommand{\bu}{\boldsymbol{u}}
\newcommand{\bv}{\boldsymbol{v}}

\newcommand{\bs}{\boldsymbol{s}}
\newcommand{\bp}{\boldsymbol{p}}
\newcommand{\bq}{\boldsymbol{q}}

\newcommand{\pth}{\boldsymbol{\pi}}
\newcommand{\cost}[1]{\| #1 \|}
\newcommand{\timec}[1]{\mbox{\textcent}( #1 )}
\newcommand{\Rsq}{\ensuremath{\R_{\square}}}
\newcommand{\Rxy}{\ensuremath{\R_{\leftrightarrow}}}
\newcommand{\Rud}{\ensuremath{\R_{\updownarrow}}}
\newcommand{\Ixp}{\ensuremath{I_x^{\oplus}}}
\newcommand{\Iyp}{\ensuremath{I_y^{\oplus}}}
\newcommand{\I}{\ensuremath{\mathcal{I}}}

\newcommand{\pc}{\mbox{\textcent}}

\newcommand*{\reals}{{\mathbb R}}

\newcommand*{\mycent}{{\normalfont \textrm{\textcent}}}

\newcommand*{\envir}{\EuScript{W}}
\newcommand*{\freesp}{\EuScript{F}}
\newcommand*{\fdfreesp}{\boldsymbol{F}}
\newcommand*{\fd}[1]{\boldsymbol{#1}}

\newcommand*{\ldist}[3]{d_{#3}(#1,#2)}
\newcommand*{\fdpi}{\boldsymbol{\pi}}

\newcommand*{\robA}{A}
\newcommand*{\robB}{B}

\newcommand*{\gedges}{\mathcal{E}}
\newcommand*{\gverts}{\mathcal{V}}

\newcommand*{\norm}[1]{\mathopen|| #1 \mathclose||}
\newcommand*{\assign}{\coloneqq}

\newcommand{\figref}[1]{Figure~\ref{fig:#1}}
\newcommand{\secref}[1]{Section~\ref{sec:#1}}
\newcommand{\thmref}[1]{Theorem~\ref{thm:#1}}
\newcommand{\lemref}[1]{Lemma~\ref{lem:#1}}
\newcommand{\obsref}[1]{Observation~\ref{obs:#1}}

\newcommand{\propref}[1]{Proposition~\ref{prop:#1}}
\newcommand{\subsecref}[1]{Section~\ref{subsec:#1}}

\newcommand{\minsum}{{\sc Min-Sum}\xspace}
\newcommand{\minmakespan}{{\sc Min-Makespan}\xspace}
\newcommand{\partition}{{\sc Partition}\xspace}
\newcommand{\pspace}{{\sf PSPACE}\xspace}
\newcommand{\p}{{\sf P}\xspace}

\newcommand{\nphard}{{\sf NP}-hard\xspace}
\newcommand{\horF}{\mbox{\sc hor}(\F)}

\newcommand{\yco}[1]{{#1}_y}
\newcommand{\xco}[1]{{#1}_x}

\title{Optimal Motion Planning for Two Square Robots in a Rectilinear Environment}

\author{Pankaj K.~Agarwal}{Department of Computer Science, Duke University, Durham, NC, USA}{pankaj@cs.duke.edu}{}{}

\author{Mark de Berg}{Department of Mathematics and Computer Science, TU Eindhoven, the Netherlands}{M.T.d.Berg@tue.nl}{https://orcid.org/0000-0001-5770-3784}{MdB is supported by the  Dutch Research Council (NWO) through Gravitation-grant NETWORKS-024.002.003.}

\author{Benjamin Holmgren}{Department of Computer Science, Duke University, Durham, NC, USA}{ben.holmgren@duke.edu}{}{}

\author{Alex Steiger}{Department of Computer Science, Duke University, Durham, NC, USA}{asteiger@cs.duke.edu}{}{}

\author{Martijn Struijs}{Department of Mathematics and Computer Science, TU Eindhoven, the Netherlands}{M.A.C.Struijs@tue.nl}{https://orcid.org/0000-0002-0116-7238}{}

\authorrunning{P.K.~Agarwal, M.~de Berg, B.~Holmgren, A.~Steiger, and M.~Struijs} 

\titlerunning{Optimal Motion Planning for Two Square Robots}

\Copyright{ Pankaj K.~Agarwal, Mark de Berg, Benjamin Holmgren, Alex Steiger, and Martijn Struijs}

\relatedversion{} % Related version of a paper, e.g. a link to the full version in arxiv

\else                % For non-lipics format
  \documentclass[11pt]{article}
  \title{Optimal Motion Planning for Two Square Robots in a Rectilinear Environment}
  \author{
  Pankaj K. Agarwal \thanks{
    Department of Computer Science, Duke University, Durham, NC 27708, USA;
    {\sf pankaj@cs.duke.edu,
    https://orcid.org/0000-0002-9439-181X}}
    \and
    Mark de Berg \thanks{
      Department of Mathematics and Computer Science, TU Eindhoven, the Netherlands;
      {\sf M.T.d.Berg@tue.nl, https://orcid.org/0000-0001-5770-3784}}
    \and
    Benjamin Holmgren \thanks{
      Department of Computer Science, Duke University, Durham, NC, USA;
      {\sf ben.holmgren@duke.edu, https://orcid.org/0009-0002-7986-7987}}
    \and
    Alex Steiger \thanks{
      Department of Computer Science, Duke University, Durham, NC, USA;
      {\sf asteiger@cs.duke.edu, https://orcid.org/0000-0003-1546-6244}}
    \and
    Martijn Struijs \thanks{
      Department of Mathematics and Computer Science, TU Eindhoven, the Netherlands;
      {\sf M.A.C.Struijs@tue.nl, https://orcid.org/0000-0002-0116-7238}}
    }
\usepackage{geom-plain}
\usepackage{changes}
\fi         % END OF \iflipcs SWITCH
\begin{document}
\maketitle
\blfootnote{Our manuscript has no associated data.}

%---------------------------------------------------------------------------
\begin{abstract}
    Let $\envir \subset \Reals^2$ be a rectilinear polygonal environment (that is, a rectilinear 
    polygon potentially with holes) with a total of $n$ vertices, and let $A,B$ be two robots, 
    each modeled as an axis-aligned unit square, that can move rectilinearly inside $\envir$. 
    The goal is to compute a \emph{collision-free motion plan} $\plan$, that is, 
    a motion plan that continuously moves $\robA$ from $s_A$ to $t_A$ and $B$ from $s_B$ to $t_B$ 
    so that $\robA$ and $\robB$ remain inside $\envir$ and do not collide with each other during the motion. 
    We study two variants of this problem which are focused additionally on the \emph{optimality} of $\plan$,
    and obtain the following results.
    \begin{itemize}
    \item \minsum: Here the goal is to compute a motion plan that minimizes
          the sum of the lengths of the paths of the robots. We present 
          an $O(n^4\log{n})$-time algorithm for computing an optimal solution to the min-sum problem.
          This is the first polynomial-time algorithm to compute an optimal, collision-free
          motion of two robots amid obstacles in a planar polygonal environment.
    \item \minmakespan: Here the robots can move with at most unit speed,
          and the goal is to compute a motion plan that minimizes
          % the makespan, which is 
          the maximum time taken by a robot to reach its target location. 
          We prove that the min-makespan variant is NP-hard.
    \end{itemize}
    \end{abstract}
%----------------------------------------------------------------------------
\section{Introduction}\label{sec:intro}
%----------------------------------------------------------------------------
Autonomous multi-robot systems are being increasingly used for a wide range of 
tasks such as logistics in industry, precision
agriculture, exploration of confined and cluttered 
environments, 
search and rescue operations, and visual inspection of areas of interest.
These applications have led to extensive work on designing efficient algorithms for computing 
high-quality--motion plans in a structured or an unstructured environment for a system of robots; 
see e.g.~\cite{ASG*-23,YJC-13} for recent surveys.
Several criteria have been proposed to measure the quality of a motion plan, including
the total length of the robot paths, 
the make-span (i.e., the maximum time taken by a robot to complete its motion), or 
by the utility of the plan (i.e., how well it performs the underlying task).
Already for two simple robots, such as unit squares or disks, translating in a planar polygonal environment,
little is known about computing an optimal motion plan.
Although polynomial-time algorithms are known for 
computing a collision-free motion plan of two simple robots~\cite{DBLP:journals/dcg/AronovBSSV99, SS91}, no polynomial-time algorithm is known even for computing a plan such that the sum (or the maximum) of the path lengths of 
the two robots is minimized, nor is the problem known to be NP-hard. In this paper we study the 
problem of computing an optimal motion plan for two simple robots, each modeled as a unit square.

%----------------------------------------------------------------------------
%----------------------------------------------------------------------------
\paragraph{Problem statement.}
%----------------------------------------------------------------------------
Let $\Box := \{x \in \Reals^2 \mid \norm{x}_\infty \leq 1/2\}$ denote the 
axis-aligned square of unit side length---referred to as 
a \emph{unit square} for short---centered at the origin. 
Let $\robA$ and $\robB$ be two robots, each modeled as a
unit square, that can translate inside the same closed rectilinear 
polygonal environment $\envir\subset\reals^2$.
In other words, the shared \emph{workspace}~$\envir$ is a rectilinear polygon,
possibly with holes.\footnote{For simplicity we assume that the
outer polygon
and its holes are non-degenerate, that is, they do not have dangling edges. 
With some care, our algorithm can also handle degenerate cases.} Let $n$
be the number of vertices of $\envir$.
A placement of $\robA$ (and, similarly, of $\robB$) is represented by 
the position of its center in the workspace~$\envir$. For such a placement 
to be free of collision with the boundary $\bd\envir$ of $\envir$,
the representing point should be at $\ell_\infty$-distance at least $1/2$ 
from $\bd\envir$. (Note that the robot is allowed to touch
an obstacle, since we define $\F$ to be a closed set.)
We let $\freesp \subset \envir$ denote
the \emph{free space} of a single robot, which is the subset of $\envir$ 
consisting of all collision-free placements.  
A (joint) \emph{configuration} $\fd{p}$ of $\robA$ and $\robB$ 
is represented as a pair $\fd{p}=(p_A,p_B) \in \envir \times \envir$, 
where $p_A$ and $p_B$ are the placements of $\robA$ and $\robB$, respectively.
The \emph{configuration space}, called \emph{C-space} for short, 
is the set of all configurations, 
and is thus represented as $\envir\times\envir\subset\Reals^4$. 
A configuration $\fd{p} = (p_A,p_B) \in \Reals^4$ is called \emph{free} 
if $p_A, p_B\in\freesp$ %, that is, $p_A+\Box, p_B +\Box \subseteq \envir$, 
% MdB: I removed the "that is, [...]" because I don't think it is necessary
% and it makes it less clear if "and  $\norm{p_A-p_B}_\infty \geq 2$"
% is part of the "that is" (which it is not)
and  $\norm{p_A-p_B}_\infty \geq 1$.
%Such a free configuration is called a \emph{kissing configuration} if $\norm{p_A - p_B}_\infty = 2$, i.e., the robots touch each other (but their interiors remain disjoint). 
Let $\fdfreesp \assign \fdfreesp(\envir)$ 
denote the \emph{four-dimensional free space}, comprising the set of all free configurations.
Clearly, $\fdfreesp \subseteq \freesp \times \freesp$.

Let $\fd{s}=(s_A,s_B)$ be a given \emph{source configuration} and let 
$\fd{t}=(t_A,t_B)$ be a given \emph{target configuration}.
An \emph{$(\fd{s}, \fd{t})$-plan} is a continuous function $\plan:[0,T] \to \envir\times\envir$,
for some $T \in \RR_{\geq 0}$, with $\plan(0) = \fd{s}$ and $\plan(T) = \fd{t}$.
The image of $\plan$ is a (continuous) curve in the C-space, referred to 
as an \emph{$(\bs, \bt)$-path}. 
% \mst{What is a path? If a \emph{path} is an image, a subset of $\envir\times\envir$ or $\envir$, then the projected paths can be self-intersecting paths. This can be fine but it is seems counter-intuitive enough to me that we should mention this. Maybe even explicitly define what we consider a path in this context?} \mst{To clarify, if the projected paths are images of a self-intersecting trajectory or curve, then we cannot obtain the trajectory from the image. In particular, when given an image, I do not know its length, as I do not know how often a single edge is traversed. This information is of course in $\plan$, but it is not clear to me how the projections access this if we only work with the images.}
% \mdb{After discussing, we decided it was fine as written. I just replaced "path" by "curve"}
With a slight abuse of notation, we use $\plan$ to denote its image 
as well.\footnote{Since we do not impose any kinodynamic constraint on 
the motion of robots (e.g. maximum acceleration or maximum curvature), given an $(\bs,\bt)$-path
in C-space, it is straightforward to compute an $(\bs,\bt)$-plan corresponding to this path.}
If $\pth \subset \fdfreesp$ we say that $\pth$ is \emph{feasible},
and if there exists a feasible $(\bs,\bt)$-plan we say that the 
pair $(\bs,\bt)$ is \emph{reachable}.
For a plan $\fdpi:[0,T] \to \envir \times \envir$, let $\pi_A:[0,T] \to \envir$ 
and $\pi_B:[0,T] \to \envir$ be the projections of $\fdpi$ onto the two-dimensional plane 
spanned by the first two coordinates and the last two coordinates, respectively. 
The functions $\pi_A$ and $\pi_B$
specify the motions of $\robA$ and $\robB$ that $\fdpi$ induces, that is,
$\plan(\lambda) = \left(\pi_A(\lambda), \pi_B(\lambda)\right)$ for all $\lambda \in [0,T]$. 
Again, with a slight abuse of notation, we also use $\pi_A$
and $\pi_B$ to denote the paths followed by $A$ and $B$, respectively.
We define two versions of optimal motion planning.
\vspace*{2mm}
\begin{itemize}
\item \minsum.
    For a path $\gamma$ in $\envir$, let $\cost{\gamma}$ denote its
    $\ell_1$-length.
    We define $\cost{\plan}$, the \emph{cost} of an $(\bs,\bt)$-plan $\plan$,
    by $\cost{\plan} := \cost{\pi_A} + \cost{\pi_B}$, that is, 
    $\cost{\plan}$ is the sum of the $\ell_1$-lengths of the paths of the two robots. 
    The problem is to decide whether a given pair $(\bs,\bt) \in \BF^2$ is reachable, and, if so, compute a minimum-cost feasible
    $(\bs,\bt)$-plan, referred to as a min-sum $(\bs,\bt)$-plan.
    (As explained later, we can in fact restrict our attention to rectilinear paths
    because $\envir$ is rectilinear and we measure the length of a path in the $\ell_1$-metric.)
    % If $\bs$ and $\bt$ are not reachable, then $\plan^*(\bs,\bt)$ does not exist.
\item \minmakespan.
    We define the \emph{makespan} of an $(\bs,\bt)$-plan $\plan:[0,T] \to \envir \times \envir$,
    denoted by $\timec{\plan}$, to be $T$.
    % \ben{Should we also say that $\timec{\plan} = \max(\timec{\pi_A}, \timec{\pi_B})$, to match, eg., \cite{demaine2019coordinated}?}
    % \mdb{Let's leave it for now.}
    The problem is to decide for a given pair $(\bs,\bt) \in \BF$ if $(\bs,\bt)$ is 
    reachable and, if so, compute a feasible $(\bs,\bt)$-plan $\plan^*(\bs,\bt)$
    that minimizes the makespan
    under the condition that the maximum speed of each robot is at most~1. 
    % where the minimization is taken over all feasible $(\bs,\bt)$-plans
    % such that the maximum speed of each robot is at most~1. 
    (Note that we do not require the $(\bs,\bt)$-plan to be $C^1$ continuous, 
    so the speed of $A$ or $B$ may instantaneously change from $0$ to $1$, 
    and the robots may take sharp turns.)
\end{itemize}
\vspace*{2mm}
%\mdb{Above I introduced problem names \minsum and \minmakespan. Do we want to use these  names in the rest of the paper, for instance when stating the theorems? And/or do we want to define the problems using an Input/Question format?}
Both variants of the optimal motion planning problem are interesting in their own right.
The first minimizes the total work done by the robots, while the second minimizes 
the total time taken until both robots have reached their destination.
  
%----------------------------------------------------------------------------
\paragraph{Related work.}
%----------------------------------------------------------------------------
It is beyond the scope of this paper to review the  known results on 
motion-planning algorithms; for a review of key relevant results,
we refer the reader to recent books and surveys on 
the topic~\cite{hks-r-18,hss-amp-18,Lav06,m-spn-18,ASG*-23,Sal19}.
We mention here only a small sample of results---ones that 
are most closely related to the problem at hand.

When there is a single translating square robot, or more generally when there is a single
convex polygonal translating robot with a constant number of vertices, the problem is 
equivalent---through C-space formulation---to moving a point robot amid polygonal 
obstacles with $O(n)$ vertices, and it can be solved in $O(n\log n)$ 
time~\cite{chen2015computing,DBLP:journals/siamcomp/HershbergerS99,DBLP:conf/stoc/Wang21}.
The problem of computing 
the shortest path for a point robot amid polyhedral obstacles in $\reals^3$  
is NP-hard~\cite{CanRei87}, and fast $(1+\eps)$-approximation algorithms are known \cite{clarkson1987,sharath2009}. 

Computing a feasible (not necessarily optimal) plan for a team of translating 
unit square robots in a polygonal environment 
is \pspace-hard~\cite{DBLP:journals/ijrr/SoloveyH16}; 
see~\cite{DBLP:conf/fun/BrockenHKLS21,DBLP:conf/fun/BrunnerCDHHSZ21,DBLP:journals/tcs/HearnD05,hopcroft1984complexity,DBLP:journals/ipl/SpirakisY84,DBLP:conf/aaai/YuL13, abrahamsen2025} 
for related intractability results.
Notwithstanding a rich literature~\cite{DBLP:journals/trob/DayanSPH23,
% DBLP:journals/ijrr/KaramanF11,
% KavSveLatOve96,
Sal19,
DBLP:journals/arobots/ShomeSDHB20,
stern2019multiagent,
DBLP:journals/arobots/TurpinMMK14} on multi-robot motion planning 
in both continuous and discrete setting---robots move on a graph in the 
latter setting---little 
is known about algorithms computing plans with provable quality guarantees.
% except in some very special cases.
%Even in the absence of obstacles, computing the min-sum motion plan for two unit squares/disks is non-trivial \cite{Esteban2022,DBLP:conf/cccg/KirkpatrickL16}.
 Kirkpatrick and Liu~\cite{DBLP:conf/cccg/KirkpatrickL16} presented an efficient algorithm 
 for computing the min-sum plan for unit disks in this setting and showed that an optimal plan always consists of at most six segments of 
 straight lines and circular arcs. Recently, similar results was derived for the cases of two 
 unit squares~\cite{Esteban2022,mastersthesisRuizHerrero} and two centrally-symmetric convex polygons~\cite{KL25}.
Recently, there have been a few results for more than two robots in a setting without obstacles:
Deligkas~\etal~\cite{DBLP:conf/icalp/DeligkasEGK024} and Eiben~\etal~\cite{DBLP:conf/compgeom/EibenGK23}
consider the problem on graphs minimizing the makespan in 
discrete time intervals, and Kanj~\etal~\cite{DBLP:conf/compgeom/KanjP24}
minimize the number of ``serial" or ``parallel" moves of the robots by proving various structural properties of an optimal plan. 

Approximation algorithms for minimizing the total path length are known
for unit-disk robots, assuming a certain minimum separation between 
the source and target positions as well as from the 
obstacles~\cite{DBLP:journals/comgeo/AgarwalGHT23,SolomonHalperin2018,DBLP:conf/rss/SoloveyYZH15}.
The separation assumption makes the problem considerably easier: 
a feasible plan always exists, and one can first compute an optimal path for each robot
independently (ignoring the other robots) and then locally modify the paths so that 
the robots do not collide with each other during their motion. 
For computing a plan that minimizes the makespan for a set of unit disks (or squares) 
in the plane without obstacles, an $O(1)$-approximation algorithm was proposed 
as well~\cite{demaine2019coordinated}, again assuming some separation.
Recently, Agarwal~\etal~\cite{AHSS-24} presented the first polynomial-time 
approximation algorithm for the min-sum motion-planning problem for two squares
without any assumptions on the work environment or on the source/target configurations.
Their algorithm computes an $(1+\eps)$-approximate solution in 
$(1/\eps)^{O(1)}\cdot n^2\log n$ time.
They leave it as an open problem whether an optimal plan in this case 
can be computed in polynomial time.

We conclude our discussion by mentioning a few other lines of work on multi-robot 
motion planning. The central and prevalent family of practical motion-planning 
techniques in robotics is based on sampling 
of the underlying C-space; see~\cite{Sal19} for a recent review. 
This paradigm has been used for multi-robot motion planning as well,
but it does not lead to a polynomial-time approximation 
algorithm for computing an optimal plan. 
See~\cite{DBLP:conf/icra/DayanSPH21,DBLP:journals/trob/DayanSPH23,DBLP:journals/ijrr/KaramanF11,DBLP:journals/arobots/ShomeSDHB20,DBLP:conf/icra/SoloveyJSFP20} 
for a few results on the analysis of this approach.
There is also work on the \emph{unlabeled} version of the problem, where
each robot can end up at any of the (collective) target positions, as long as 
all of the target positions are occupied by robots
at the end of the motion. For a team of unlabeled unit disks, an approximate solution 
for the minimum total path length is given in~\cite{DBLP:conf/rss/SoloveyYZH15}, 
assuming a certain separation between the source and target positions of the robots, 
as well as from the obstacles.  
See also~\cite{DBLP:journals/tase/AdlerBHS15,BanyassadyEtAl.SoCG.2022}. 
Finally, another major line of work on optimizing multi-robot motion plans 
addresses a discrete version of the problem, where robots are moving on graphs. 
In this setting the robots are often referred to as \emph{agents},
and the problem is called \emph{multi-agent path finding}~(MAPF). 
There is a rich literature on MAPF, 
and we refer the reader to the recent survey~\cite{stern2019multiagent}.

%----------------------------------------------------------------------------
\paragraph{Our contributions.}
%----------------------------------------------------------------------------
Our main result is that the min-sum motion planning problem for 
unit square robots translating in a 2D polygonal environment is in~\p
while the min-makespan problem in the same setting is \nphard, as
stated in the two following theorems.
%
%-------------------------------------------------------------------------
\begin{theorem}
\label{thm:two-robots-alg}
Let $\envir$ be a closed rectilinear polygonal environment with $n$ vertices, let 
$\robA,\robB$ be two axis-parallel unit-square robots translating inside $\envir$,
and let $\fd{s},\fd{t}$ be source and target configurations of $\robA,\robB$. 
We can compute an optimal min-sum motion plan $\fdpi$ from $\fd{s}$ to $\fd{t}$ under 
the $\ell_1$-metric, or determine that no feasible motion exists, in $O(n^4\log n)$ time. 
\end{theorem}
%-------------------------------------------------------------------------
\begin{restatable}{theorem}{hardness}
\label{thm:two-robots-hardness}
Let $\envir$ be a closed rectilinear polygonal environment with $n$ vertices, let 
$\robA,\robB$ be two axis-parallel unit-square robots translating inside $\envir$,
let $\fd{s},\fd{t}$ be source and target configurations of $\robA,\robB$,
and let $T_{\max}$ be a given maximum time.
It is \nphard to determine whether there is a 
feasible $(\bs,\bt)$-plan $\plan^*$ such that the maximum speed of each robot is 
at most~1 and the makespan of $\plan^*$ is at most $T_{\max}$.
\end{restatable}
%-------------------------------------------------------------------------

We prove \thmref{two-robots-hardness} by a reduction from the 
\textsc{Partition} problem, which is to decide whether a given set $Y$ of integers
admits a partition into subsets 
$Y_1,Y_2$ such that $\sum_{p \in Y_1} p = \sum_{q \in Y_2} q$.
The idea of the reduction is to build a workspace~$\envir$ that consists
of gadgets $\envir_1,\ldots,\envir_m$, each corresponding to an element of~$Y$ and to choose a parameter $T_{\max} \geq 0$, 
so that both robots must pass through every gadget and there is
a plan with makespan at most $T_{\max}$ if and only if there is a valid partition of~$Y$. The reduction is described in detail in \secref{hardness}.

\medskip

Our main technical contribution is a polynomial-time exact algorithm for \minsum,
as stated in \thmref{two-robots-alg}.
The crucial step towards designing this algorithm is to prove the existence of 
% an optimal plan that satisfies several desirable structural properties. 
% In particular, we prove the existence of 
an optimal \emph{canonical-grid plan},
in which the path followed by each robot lies on an $O(n) \times O(n)$ 
non-uniform predefined grid; see \thmref{canonical}.

For a single robot, the existence of such a plan is well known and easy to prove~\cite{RLW89}:
since an optimal path in the $\ell_1$-metric is rectilinear,
one can push the segments of an optimal path in such a way that
each segment is either incident to the source or target position,
or a part of the segment contains an edge of the free space~$\freesp$.
Importantly, pushing the path can be done without increasing its length.
Thus there is an optimal path on the grid defined by the horizontal
and vertical lines through the source and target positions and the lines
containing the edges of the free space~$\freesp$.  

For a pair of robots, the existence of an optimal plan on a suitably 
defined grid seems intuitive as well. Indeed, consider an optimal plan
$\plan^*(\bs,\bt)$ and imagine trying to push a segment of~$\pi_A$ in the same
manner as in the single-robot case. It may happen that we fail to push
the segment onto the grid because robot~$\robB$ is blocking it. 
The hope is that we can continue to push the segment of $\pi_A$
and push the appropriate segment of $\pi_B$ in the same direction---essentially
$\robA$ is pushing $\robB$ out of the way---until $\robB$
hits $\bd\envir$. Thus, we have pushed a segment of $\pi_B$
onto a line containing an edge of~$\freesp$ and the segment of $\pi_A$ 
onto a line at distance~1 from that line. Making this idea work
is highly nontrivial though. For instance, we must choose the
direction into which we push a specific segment of $\pi_A$ in such a way
that $\cost{\pi_A}$ does not increase, but this
may force us to push $\pi_B$ in such a way that $\cost{\pi_B}$ increases.
Hence, we have to choose the pushing directions carefully.
This is further complicated by the fact that $\robB$ may be on
different sides of a segment of $\pi_A$ at different moments in time.
Another major complication is that pushing (a segment of) one robot 
onto a line at distance $i$ from an edge of $\freesp$, may
cause the other robot to end up on a line at distance~$i+1$ from this edge. 
If we are not careful, this can cause a cascading effect
and the final grid will not have the desired size.

To overcome these issues, we first prove several structural properties of an optimal plan.
Using these properies, we combine local pushing
operations with more global re-routing operations such
that the transformed paths end up on lines at distance 0, 1, or 2
from the edges of~$\freesp$. These operations may increase the length
of one of the paths, but we will argue that when this happens
the length of the other path decreases by the same amount.
Thus we end up with a plan on a grid of size $O(n) \times O(n)$ 
and that is still optimal.

Using the existence of an optimal canonical-grid plan, 
we show that we can construct a $4$-dimensional weighted grid graph $\graph=(\gverts,\gedges)$, 
with $\fd{s},\fd{t}\in \gverts$, of size $O(n^4)$, in $\fdfreesp$
such that a shortest $(\bs,\bt$)-path in $\graph$ corresponds to a min-sum $(\bs,\bt)$-plan.
%-------------------------

%-------------------------------------------------------------------------
\section{A Polynomial-Time Algorithm for \minsum}
%-------------------------------------------------------------------------
This section describes a polynomial-time algorithm for \minsum. We first define a
\emph{canonical-grid plan} in \secref{canonical}---an \stplan where both robots move on a 
predefined $O(n) \times O(n)$ nonuniform grid---and claim that for any reachable
pair $\bs,\bt \in \BF$, there exists an optimal $(\bs,\bt)$-plan that is a
canonical-grid plan. In \secref{convert} we give a high-level overview
of the proof of this claim, which proceeds by transforming an optimal plan 
to a canonical-grid plan whose cost is no more than that of the original plan.
A detailed proof of the claim is deferred to \secref{grid-snapping}.
Finally, we describe our algorithm in \secref{alg}.
% As mentioned in the introduction, there is an $O(n\log{n})$ algorithm to determine whether
% two configurations $\bs, \bt \in \envir \times \envir$ in a given polygonal environment
% $\envir$ are reachable, so in the remainder of this paper we assume that $\bs$ and $\bt$ are reachable. \mst{I don't think we need to use another algorithm. We prove that if there is a path, then there is an optimal one with the desired structure. So, we learn that the points are not reachable after the algorithm exhausts its search.}

%-------------------------------------------------------------------------
\subsection{Canonical-grid plans}\label{sec:canonical}
%-------------------------------------------------------------------------
\ifsocg
\subparagraph{Rectilinear, decoupled plans.}
\else
\paragraph{Rectilinear, decoupled plans.}
\fi
%-------------------------------------------------------------------------
Since $\envir$ is a rectilinear environment and $A,B$ are squares, it is easily
seen that $\BF$ is a polyhedral region. Hence, if there is a feasible plan,
then there is a piecewise linear 
optimal plan, that is, an optimal plan whose image
is a polygonal chain in $\envir \times \envir$~\cite{AHSS-24}.
We thus focus on piecewise-linear plans.
We refer to the vertices of such a plan as \emph{breakpoints}. 
For a path $\plan = (\pi_A, \pi_B)$,
the breakpoints of $\pi_A$ and $\pi_B$ are the projections of the breakpoints of~$\pi$. 
Note that the two segments incident on a breakpoint of $\pi_A$ (or $\pi_B$) may be collinear.

An $(\bs,\bt)$-plan $\plan$ is called \emph{decoupled} if only one robot 
moves at any given time; the other robot is then \emph{parked} at some 
point in $\fre$. A decoupled plan can be represented as a sequence of 
\emph{moves}, with each move specifying the parking location of one robot 
and the motion of the other robot. By definition of breakpoints,
each parking spot of $A$ (resp. $B$) is a breakpoint of $\pi_A$ (resp. $\pi_B$). 
Agarwal et al.~\cite{AHSS-24} have shown that for any pair $\bs,\bt$ of 
reachable configurations, there is an optimal $(\bs,\bt)$-plan that is decoupled. 
Since only one robot moves at a time, the parked robot can be considered 
as an obstacle during the move. Thus, the moving robot moves from one position
(the start of the move) to another position (the end of the move) in
a rectilinear environment with obstacles (the original obstacles
plus the parked robot). This can always be done in an optimal
manner---that is, in a manner that minimizes the length
of the motion---with a rectilinear path.
Hence, we can ensure that each robot follows a rectilinear path in $\freesp$ in each move.
We refer to such a plan as a \emph{rectilinear} plan. The breakpoints of $\pi_A$
and $\pi_B$ in a rectilinear plan $\plan = (\pi_A, \pi_B)$
are the parking spots of $A$ and $B$ and the points at which they switch between horizontal and vertical segments.
We thus conclude the following.
\begin{proposition}\label{prop:rectilinear-decoupled}
      For any reachable pair of configurations $\bs, \bt \in \BF$, there is an 
      optimal $(\bs,\bt)$-plan that is rectilinear and decoupled.
\end{proposition}

%-------------------------------------------------------------------------
\ifsocg
\subparagraph{The canonical grid.}
\else
\paragraph{The canonical grid.}
\fi
%-------------------------------------------------------------------------
For an integer $i \geq 0$, we define an \emph{$i$-line} to be a horizontal 
or vertical line~$\ell$ that lies at distance~$i$ from an edge of $\partial \freesp$ 
parallel to $\ell$, or that lies at distance~$i$ from one of the points~$s_A, s_B, t_A, t_B$.
Note that $0$-lines support an edge of $\partial\freesp$ or pass through
one of the points~$s_A, s_B, t_A, t_B$. Let $L_i$ denote the set of all $i$-lines,
and set $L_{\leq i} := \bigcup_{j=0}^{i} L_j$. We will be mostly interested in $L_{\leq 2}$. 
The \emph{arrangement}~\cite{arrangements} of $L_{\leq 2}$, denoted by $\A(L_{\leq 2})$, is the 
subdivision of $\RR^2$ induced by $L_{\leq 2}$.
Since each edge of $\freesp$ is contained in a $0$-line, each face of $\A(L_{\leq 2})$ is 
a rectangle that is either contained in $\freesp$ or disjoint from $\freesp$.
The vertices and edges of $\arr{\leq 2}$ form a planar grid graph. 
Let $G = (V,E)$ denote the subgraph of this graph that lies in~$\freesp$,
that is, $V$ and $E$ are the sets of vertices and edges, respectively,
of $\arr{\leq 2}$ that lie in $\freesp$. Note that $|V|= O(n^2)$ and $|E| = O(n^2)$. 
We refer to $G$ as a \emph{grid}, to a vertex of $V$ as a \emph{grid point},
to an edge of $E$ as a \emph{grid edge},
% \mst{or grid segment? Usage is rare, either way.}
% \mdb{The plan is to use "segment" for when we talk about a plan/path, and "edge" for the free spce. Let's also use edge for the grid.}
and to a line of $L_{\leq 2}$ as a \emph{grid line}. 
We call a rectilinear, decoupled $(\bs,\bt)$-plan $\plan = (\pi_A, \pi_B)$
a \emph{canonical-grid plan} if each breakpoint of $\pi_A$ and $\pi_B$
is a grid point. Thus each segment of the path is the union of a number
of consecutive collinear grid edges.
%\mdb{I do not understand the rest of the sentence. 
% Do we want to say: \ldots and each segment is the union of a number
% of consecutive grid edges. Doesn't this follow from the fact that
%all breakpoints are grid points?} \mst{The rest of the sentence is needed when we have 0-length segments. However, it may be better to leave that subtlety out of the definition here and only note this when we introduce 0-length segments.}
%and each horizontal (resp. vertical)
%segment of $s = pq$ spans a sequence of horizontal (resp. vertical) edges of $E$ connecting $p$ to $q$.
By definition, each parking spot of $A$ or $B$ in a canonical-grid plan is a grid point.
Our main technical result for the \minsum problem is the following theorem.
%-------------------------------------------------------------------------
\begin{theorem}\label{thm:canonical}
      Let $\envir$ be a closed rectilinear polygonal environment with $n$ vertices, let 
      $\robA,\robB$ be two axis-parallel unit-square robots translating inside $\envir$, 
      and let $\fd{s},\fd{t}$ be source and target configurations of $\robA,\robB$. 
      There is an optimal $(\bs,\bt)$-plan that is a canonical-grid plan.
\end{theorem}
%----------------------------------------------------------------------------	

%-------------------------------------------------------------------------
\subsection{Converting an optimal plan to a canonical-grid plan}\label{sec:convert}
%-------------------------------------------------------------------------
Next we give a high-level overview of how we convert an optimal $(\bs,\bt)$-plan
$\plan = (\pi_A, \pi_B)$ to an optimal $(\bs,\bt)$-plan that is also a canonical-grid plan.
By \propref{rectilinear-decoupled}, we can assume $\plan$ is a rectilinear, decoupled plan.

\ifsocg
\subparagraph{Alternating plans.}
\else
\paragraph{Alternating plans.}
\fi
We can assume that the path $\pi_A$ alternates between vertical and horizontal segments
or, in other words, that each breakpoint of $\pi_A$ is incident to a horizontal and 
a vertical segment and is thus a vertex of $\pi_A$. This \emph{alternation property} 
for $\pi_A$ can be ensured by adding zero-length segments\footnote{The time interval associated with a zero-length segment---the
time during which the segment is traversed---is $[\lambda,\lambda']$, where $\lambda$
is the arrival time at the corresponding parking spot and $\lambda'$ is the departure time.} 
at breakpoints whose
two incident segments have the same orientation. 
The alternation property, which we can also assume for $\pi_B$,
implies that all breakpoints of $\pth_A$ and $\pth_B$ are vertices of the respective paths and thus
so are the parking spots.

We now describe how to convert $\plan$ into a canonical-grid plan without increasing its cost.
We accomplish this conversion in two phases: the first phase modifies the plan $\plan$
such that all horizontal segments of $\pi_A$ and $\pi_B$ lie on grid lines, and the
second phase does the same for the vertical segments. Here we describe the first phase;
the second phase proceeds analogously, with the roles of the $x$ and $y$-coordinate exchanged.

\ifsocg
\subparagraph{The first phase.}
\else
\paragraph{The first phase.}
\fi
We define a \emph{bad horizontal segment} to be a horizontal segment on $\pi_A$ 
or $\pi_B$ that does not lie on a grid line. % (i.e., a line of $L_{\leq 2}$).
As long as $\plan$ contains a bad horizontal segment $e$, 
we modify $\plan$ into a new plan $\plan^*=(\pi^*_A,\pi^*_B)$
such that $e$ lies on a grid line and the following conditions are satisfied:
%----------------------------------------------------------------------------
\begin{itemize}
    \item (P1) \textsc{Feasibility:} The new plan $\plan^*$ is feasible.
    \item (P2) \textsc{Optimality:} The cost of the plan does not increase,
          that is, $\cost{\plan^*} \leq \cost{\plan}$.
    \item (P3) \textsc{Progress:} The plan $\plan^*$ has at least one fewer bad horizontal segment
    than $\plan$.
    \item (P4) \textsc{Vertical alignment:} Any vertical segment of $\pi^*_A$ 
                is collinear with a vertical  segment of~$\pi_A$ or $\pi_B$,
                or it is contained in a vertical grid line; and
                the same holds for any vertical segment of~$\pi_B^*$.
    \item (P5) \textsc{Alternation:} $\plan^*$ is an alternating, rectilinear, decoupled plan.
\end{itemize}
%----------------------------------------------------------------------------
The \textsc{Progress} property guarantees that 
by applying the procedure finitely many times, starting with the optimal plan~$\plan$,
we obtain a plan without bad horizontal segments. The \textsc{Feasibility} and \textsc{Optimality}
properties imply that $\plan^*$ is an optimal $(\bs,\bt)$-plan.
Furthermore, the \textsc{Alternation} property implies that $\plan^*$
is a rectilinear, decoupled, alternating plan. Hence, the solution resulting
from one iteration satisfies the precondition for the next iteration,
so we can indeed apply the procedure iteratively.

After eliminating all bad horizontal segments in the first phase, 
we eliminate the bad vertical segments in the second phase. The \textsc{Vertical Alignment}
property, which in the second phase applies to horizontal segments, guarantees that 
the horizontal segments are not moved off the grid in the second phase. 
So, all segments of $\plan^*$ lie on grid lines, and since $\plan^*$ is alternating, its breakpoints therefore lie on grid points.
Hence, at the end of the second phase, we obtain an optimal, canonical-grid plan.
\\ \\
\emph{Note:} Our transformation of the original plan~$\plan$ into a $\plan^*$ satisfying (P1)--(P5)
is done via several intermediate steps. When describing these steps, 
we typically denote, with a slight abuse of notation, the current plan by 
$\plan$ and the next plan by $\plan^*$.

%----------------------------------------------------------------------------
\ifsocg
\subparagraph{Corridors.}
\else
\paragraph{Corridors.}
\fi
%----------------------------------------------------------------------------
To be able to push bad horizontal segments onto horizontal grid lines
without creating a cascading effect, we classify the horizontal grid lines into two
classes: \emph{primary} grid lines, which are the horizontal $0$- and $1$-lines,
and \emph{secondary} grid lines, which are the horizontal $2$-lines.
Let $\horF$ be the subdivision of~$\freesp$ induced by the primary grid lines; see \figref{corridors}.
%----------------------------------------------------------------------------
\begin{figure}
      \centering
      \includegraphics{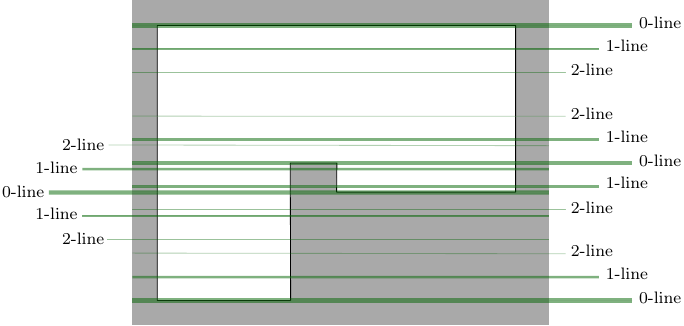}
      \caption{The $0$-lines and $1$-lines define $\horF$, which consist of twelve corridors. 
      The dark gray region is an obstacle. To avoid cluttering the figure, the start and goal 
      positions of the robots and the grid lines they define are omitted.}
      \label{fig:corridors}
      \end{figure}
%----------------------------------------------------------------------------
Each face of $\horF$ is a rectangle because every vertex of $\freesp$
is contained in a horizontal $0$-line. We refer to the faces of $\horF$
as \emph{corridors}. Note that the interior of a corridor may be intersected
by secondary grid lines.
The vertical edges of a corridor are contained in vertical edges of $\freesp$.

By definition, (the relative interior of) every bad horizontal segment 
lies in the interior of a corridor of $\horF$.
At each step, we pick a bad segment $e$, say of $\pi_A$. We perform surgery on
$\pi_A$ and $\pi_B$ so that $e$ is \emph{aligned} with one of the horizontal edges 
of the corridor~$R$ containing~$e$,
and so that the resulting plan $\plan^*$ satisfies (P1)--(P5). 
% Hence, it may take several steps to get rid of all bad
% segments in a short corridor~$R$. Inside a tall corridor we can be much more aggressive,

% By performing the surgery described in \todo{link correct section 3 parts}
% and using \todo{link lemmas} at each step, we convert $\plan$ into an optimal $(\bs,\bt)$-plan that does
% not have any bad horizontal segments. Next, we conduct the same procedure for the vertical segments.
% We obtain an optimal, canonical-grid $(\bs,\bt)$-plan, which proves \thmref{canonical}.

%----------------------------------------------------------------------------	
\subsection{Computing an optimal canonical-grid plan}\label{sec:alg}
%----------------------------------------------------------------------------	
With \thmref{canonical} at hand, we are now ready to describe our algorithm.

%----------------------------------------------------------------------------	
\ifsocg
\subparagraph{The configuration graph.}
\else
\paragraph{The configuration graph.}
\fi
%----------------------------------------------------------------------------	
Let $G = (V,E)$ be the canonical grid defined in \secref{canonical}.
We construct an edge-weighted \emph{configuration graph} $\G = (\V, \E)$
with weight function $w:\E \to \RR_{\geq 0}$ whose nodes correspond to 
free configurations $\bu \in \BF$. More precisely,
\[
\V := \{ \bu = (u_A, u_B) \in V \times V : \cost{u_A - u_B} \geq 1 \}.
\]
By construction, $\bs,\bt \in \V$.
The set $\E$ of edges of the configuration graph is defined as
\[
\E := \Big\{ (\bu,\bv) \in \V \times \V : \left(u_A = v_A \mbox{ and } (u_B, v_B) \in E\right)
      \mbox{ or } \left(u_B = v_B \mbox{ and } (u_A, v_A) \in E\right) 
      \Big\}.
\]
An edge $\left((p, u_B), (p, v_B)\right)\in\E$ corresponds a move in which
robot~$\robA$ is parked at $p$ while robot~$\robB$ moves from $u_B$ to an adjacent grid point~$v_B$.
Similarly, an edge $\left((u_A, q), (v_A, q)\right) \in \E$ corresponds to a move in which
$A$ moves from a grid point $u_A$ to an adjacent grid point $v_A$ while $B$ is parked at $q$.
\ifsocg 
(\lemref{free-segment} in the appendix states that such moves are feasible.)
\else 
(\lemref{free-segment} below states that such moves are feasible.)
\fi
The weight of an edge $(\bu,\bv) \in \E$ is defined to be $w(\bu,\bv) := \|\bu-\bv\|_1$,
which is the $\ell_1$-distance between $\bu$ and $\bv$ in $\RR^4$.

%----------------------------------------------------------------------------	
\ifsocg
    \subparagraph{The algorithm.}
\else
    \paragraph{The algorithm.}
\fi
%----------------------------------------------------------------------------	
We first compute $\freesp$ in $O(n\log{n})$ time in a standard manner~\cite[Chapter~13]{debergbook}.
Next, we compute $L_{\leq 2}$ in additional $O(n)$ time and then compute the grid graph
$G$ in $O(n^2\log{n})$ time using a sweep-line algorithm. After computing $G$, we compute
the configuration graph $\G$ in $O(n^4)$ time. Finally, we compute a shortest path from
$\bs \in \V$ to $\bt \in \V$  in $\G$ in $O(n^4\log{n})$ time using Dijkstra's algorithm.
If we find that no path from $\bs$ to $\bt$ exists, we report the instance to be infeasible,
otherwise we return the $(\bs,\bt)$-plan corresponding to the shortest path in $\G$.
The correctness of the algorithm follows from \lemref{correctness}.
The total run time is $O(n^4\log{n})$. This completes the proof of \thmref{two-robots-alg}.

\ifsocg
\else
%----------------------------------------------------------------------------	
\subsection{Correctness of the Algorithm}

The correctness of the algorithm
follows from the following key properties of $\G$.
%----------------------------------------------------------------------------	
\begin{lemma}\label{lem:free-segment}
      For every edge $(\bu,\bv) \in \E$, the segment $\bu\bv$ lies in $\BF$.
\end{lemma}
%----------------------------------------------------------------------------	
\begin{proof}
      Without loss of generality, assume that $(\bu,\bv) = \left((p,u_B), (p, v_B)\right)$ 
      and that $u_B v_B$ is a vertical edge of $G$. Note that $(p, u_B) \in \V$ and $(p, v_B) \in \V$. 
      Hence, the only way for the move corresponding to~$\bu\bv$ to be infeasible,
      is if $B$ would collide with $A$ (which is parked at the grid point $p$) 
      after leaving $u_B$ and before arriving at~$v_B$.
      Since $p$ conflicts with a point in the interior of $u_Bv_B$ and not with
      the points $u_B$ or $v_B$ themselves, the horizontal grid line through~$p$ 
      must intersect the interior of the vertical segment~$u_B v_B$.
      But this would contradict that $u_B$ and $v_B$ are adjacent grid points. 
      Thus $\bu\bv \subset \BF$.
\end{proof}
%----------------------------------------------------------------------------	
\begin{lemma}\label{lem:correctness}
      Assuming $\bs, \bt \in \BF$ are reachable, a shortest path in $\G$ from $\bs$ to $\bt$
      corresponds to an optimal, canonical-grid $(\bs,\bt)$-plan.
\end{lemma}
%----------------------------------------------------------------------------	
\begin{proof}
      Denote the cost of a path~$\plan$ in the configuration graph~$\G$ 
      by $w(\plan) := \sum_{e \in \plan} w(e)$.
      Let $\plan$ be a shortest path in $\G$ from $\bs$ to $\bt$.
      By \lemref{free-segment} we have
      $\plan \subset \BF$, and so $\plan$ is a feasible $(\bs,\bt)$-plan. Furthermore,
      $w(\plan) = \cost{\plan}$ by definition. To bound the cost of $\plan$, let $\plan^* = (\pi_A^*, \pi_B^*)$
      be an optimal $(\bs,\bt)$-plan. By \thmref{canonical}, we can assume that $\plan^*$ is a
      decoupled, canonical-grid plan. Therefore $\plan^*$ can be decoupled into a sequence of moves
      in which one robot is parked at a grid point and the other moves along 
      a segment of a grid line. Suppose $A$ is parked at $p$ and $B$ moves along a grid segment $e = u_B v_B$. Then
      $(p,u_B),(p,v_B) \in \V$ and the move from $(p, u_B)$ to $(p, v_B)$ can be represented as a sequence
      $(p, q_0), (p, q_1), ..., (p,q_k)$ with $q_0 = u_B$ and $q_k = v_B$
      such that each $q_i q_{i+1}$%, $1 \leq i \leq k$,
      is an edge of $G$. Hence, the move from $(p, u_B)$ to $(p, v_B)$ 
      is represented by a path in $\G$ from $(p, u_B)$
      to $(p,v_B)$. Repeating the argument for all moves of $\plan^*$, 
      % we map $\plan^*$ to a path $\tilde{\plan}^*$
      % from $\bs$ to $\bt$ in $\G$ with $w(\tilde{\plan}^*) = \cost{\plan^*}$.
      % Therefore $\cost{\plan} = w(\plan) \leq w(\tilde{\plan}^*) = \cost{\plan^*}$,
      % which means that $\plan$ is an optimal $(\bs,\bt)$-plan.
      we see that the entire plan $\plan^*$ is represented by a path in $\G$.
      We conclude that $\plan$, which is the shortest path in $\G$, cannot be 
      longer than~$\plan^*$. Thus, $\plan$ is an optimal $(\bs,\bt)$-plan.
\end{proof}
\fi

%----------------------------------------------------------------------------
\section{Existence of an optimal canonical-grid plan}
\label{sec:grid-snapping}
%----------------------------------------------------------------------------
Let $\plan$ be an optimal plan that contains a bad horizontal 
segment $e$, say of $\pth_A$.
Let $\R = I_x \times I_y$, where $I_x = [x_\R^-, x_\R^+]$ and $I_y = [y_\R^-, y_\R^+]$,
be the corridor that contains~$e$.
% \mdb{We use $\backslash$R, as well as just the letter R, for the corridor. Similarly, we use $\backslash$robA and A, and $\backslash$robB and B.}
This section describes how to modify $\plan$ so that $e$ is aligned
with a grid line and the resulting plan satisfies (P1)--(P5).
We begin by introducing some notation. For an axis-aligned rectangle $\rho = \delta_x \times \delta_y$,
let $\topr$ and $\botr$ be the top and bottom edge of $\rho$, respectively,
and let $\height(\rho) :=  |\delta_y|$ denote its height.
For $0 \leq \mu_1 \leq \mu_2$, let 
$\plan[\mu_1, \mu_2]$ % = \{\plan(\lambda)|\mu_1 \leq \lambda \leq \mu_2\}$ 
% \mst{I don't think we should define subplans as point-sets. I would remove the formula and either describe it more carefully in words, or assume that the reader will understand by analogy with subpaths.}
denote the subplan of $\plan$ (or its image) during the closed interval $[\mu_1, \mu_2]$,
and let $\plan(\mu_1, \mu_2)$ denote the subplan during the open interval $(\mu_1, \mu_2)$.
We define $\pth_A[\mu_1, \mu_2]$ and $\pth_A(\mu_1, \mu_2)$ 
(resp. $\pth_B[\mu_1, \mu_2]$ and $\pth_B(\mu_1, \mu_2)$) in the same
manner.
For a horizontal segment $g$, let $y(g)$ be its $y$-coordinate.

Let $[\lambda',\lambda'']$ be the time interval associated with~$e$, i.e.,
$e = \pth_A[\lambda',\lambda'']$. 
We define two critical time values related to the bad segment~$e$. 
 %----------------------------------------------------------------------------
 \begin{itemize}  
		 \label{def:lambdas}
	\item $\lambda_1 = \max\{\lambda < \lambda' \mid \pth_A(\lambda) \in \topR \cup\botR\}$
	\item $\lambda_2 = \min\{\lambda > \lambda'' \mid \pth_A(\lambda) \in \topR\cup\botR\}$ 
 \end{itemize}
%----------------------------------------------------------------------------
 Thus, % $\firste$ is the first time value at which $A$ is on the segment $e$,
 $\pth_A(\lambda_1)$ (resp. $\pth_A(\lambda_2)$) is the last (resp. first) point on
 $\pth_A$ before (resp. after) $e$ at which $A$ lies on a horizontal edge of $\R$.
 Note that $e \subseteq \pth_A(\lambda_1, \lambda_2) \subseteq R\setminus (\topR\cup\botR)$ 
 and that $\pth_A(\lambda_1, \lambda_2)$ may touch a vertical edge of $\R$.
 The goal of our surgery is to push $\pth_A[\lambda_1, \lambda_2]$ to
 $\topR$ or $\botR$, to align it with a grid line.
 This may cause collisions with $\pth_B[\lambda_1,\lambda_2]$, so
 we have to proceed carefully. Sometimes we push part of
 $\pth_A[\lambda_1, \lambda_2]$ to $\topR$ and part to $\botR$, 
 and in one case we even modify the paths outside
 the interval~$[\lambda_1, \lambda_2]$. In all cases, 
 we resolve the collisions by pushing $B$ out of the way,
 without increasing the overall length of the paths. 
  \medskip
  
We begin in \subsecref{influence} by defining a neighborhood of $\R$, called the \emph{influence region}
of~$\R$ and denoted by $\I(\R)$, and prove a few key properties of $\I(\R)$. Next, in \subsecref{unsafe}
we define \emph{unsafe} and \emph{swap} time intervals---subintervals of $[\lambda_1, \lambda_2]$
during which the surgery on $\plan$ requires more care---and we prove several structural
properties of an optimal plan during these intervals. With these properties at our disposal,
we describe the surgery of $\plan$ in Section~\ref{sec:surgery},
and prove in \secref{surgery-correct} that the resulting plan
satisfies (P1)--(P5). 
We first state a simple but important property of $\freesp$ 
that we will be using repeatedly.
%----------------------------------------------------------------------------
\begin{observation}\label{obs:free-segments}
	For any axis-aligned line $\ell$, the length of each
	connected component of $\ell\setminus \freesp$ is more than $1$.
	Consequently, for any axis-aligned segment $pq$ of length at most $1$ such that $p,q \in \freesp$,
    we have $pq \subseteq \freesp$.
\end{observation}
%----------------------------------------------------------------------------

%----------------------------------------------------------------------------
\subsection{Influence regions}\label{subsec:influence}
%----------------------------------------------------------------------------
%----------------------------------------------------------------------------
\begin{figure}[t]
	\centering
    \includegraphics[width=0.495\textwidth]{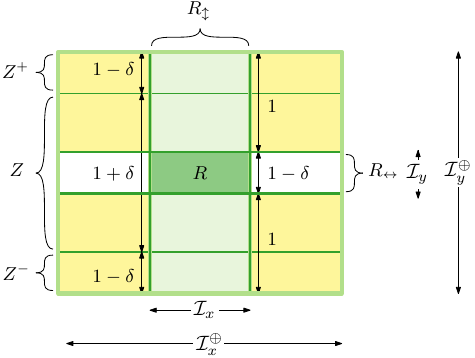}
        \hspace{1cm}
        \includegraphics[width=0.42\textwidth]{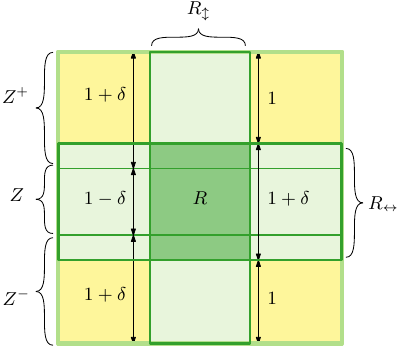}
	\caption{An overview of the various regions inside $\Rsq$ (not to scale). 
	The regions $\Rxy^-, \Rxy, \Rxy^+$ partition $\Rsq$, and so do the regions 
	$\Rud^-, \Rud, \Rud^+$. Corner squares are yellow. 
    (Left) The height of $\rect$ is less than $1$, 
    so the regions $Z^-, Z^+$ are disjoint from~$\R$. 
    (Right) The height of $\rect$ is at least $1$, so the regions $Z^-, Z^+$ overlap~$\R$.
	}
	\label{fig:lines-of-rect}
\end{figure}
%----------------------------------------------------------------------------
\paragraph{Buffer region and its partitions.}
Let $\square$ be the unit square as defined in \secref{intro}.
Let $\Rsq := \R \oplus 2\square$ be the \emph{buffer region} around $\R$.
Thus, $\Rsq = \Ixp \times \Iyp$
where $\Ixp := [x_{\R}^{-}-1, x^{+}_{\R} + 1]$
and $\Iyp := [y_{\R}^{-}-1, y^{+}_{\R} + 1]$; see \figref{lines-of-rect}.
For any pair $p,q \in \RR^2$ such that $p \in \R$ and
$\left(p + \square\right) \cap \left(q + \square \right) \not= \emptyset$, 
we have $q \in \Rsq$.

Next, we define three different partitions of $\Rsq$, each of which plays a role
in our analysis. Let $\Rxy := \Ixp \times I_y$, that is, $\Rxy$ is obtained by extending
$\R$ by unit length to the left and to the right. Similarly, let $\Rud := I_x \times \Iyp$
be the extension of $\R$ by unit length in the upwards and downwards direction. 
Let $\Rxy^- := \Ixp \times [y_\R^{-} - 1, y_\R^-)$ and
$\Rxy^+ := \Ixp \times (y_\R^+, y_\R^{+} + 1]$.
Thus, $\Rxy^-$ and $\Rxy^+$ are unit-height rectangles that lie
immediately above and below $\Rxy$, respectively. 
Note that $\Rxy^-$ does not include its top edge and 
$\Rxy^-$ does not include its bottom edge, so that
$\Rxy^-, \Rxy, \Rxy^+$ form a partition of~$\Rsq$.
We also define $\Rud^- := [x_\R^{-} - 1, x_\R^-) \times \Iyp$
and $\Rud^+ := (x_\R^{+}, x_\R^{+} + 1] \times \Iyp$ as the
unit-width rectangles to the left and right of~$\Rud$, and note that
$\Rud^-, \Rud, \Rud^+$ form a partition of~$\Rsq$ as well.
Observe that $\Rsq \setminus(\Rud \cup \Rxy)$
consists of four semi-open unit squares, which we refer to as the \emph{corner squares} 
of $\Rsq$. Finally, we define 
$Z = \Ixp \times [y_\R^{+} - 1, y_\R^{-} + 1]$,
$Z^- = \Ixp \times [y_\R^{-} - 1, y_\R^{+} - 1)$ and $Z^+ = \Ixp \times (y_\R^{-} + 1, y_\R^{+}+1]$.
%Note that $\height(Z) = 1 + \delta$ and $\height(Z^-) = \height(Z^+) = 1- \delta$. 
The closures of $Z^-, Z^+$ are translates of $\R$ by distance $1$ in the $y$-direction.
Since no primary grid line of $L_{\leq 1}$  intersects $\interior(R)$,
the regions $Z^+$ and $Z^-$ do not contain any vertices of $\freesp$. By \obsref{free-segments}, the connected components
of $Z^+ \cap \freesp$ and $Z^- \cap \freesp$ are rectangles and $\Rud \cap \freesp$ is an $x$-monotone rectilinear
polygon.

\paragraph{Influence region and blocked pairs.}
We define the \emph{influence region} of $R$ to be $\I(R) := \freesp \cap \Rsq$.
If $\height(\rect) \geq 1$ then $\I(\rect) = \Rud = I_x \times \Iyp$, 
because no vertex of $\F$ lies within distance 1 from $\rect$ when $\height(\rect) \geq 1$. 
However, if $\height(\rect) < 1$, then
$\I(R)$ may consist of several connected components--- a \emph{giant} component $\gamma(R)$ that contains
$R$, and several \emph{tiny} components, that each lie inside a corner square of~$\Rsq$.
%\figref{tiny-giant}(ii) shows an example of how tiny components can arise.
By \obsref{free-segments}, each tiny component of $\I(R)$ is $xy$-monotone.
The giant component may intersect a corner square. In particular, let $\rho$ be a corner square
that lies in, say, $\Rud^-$. By \obsref{free-segments}, there is at most one connected component $\phi$ of $\rho \cap \freesp$
adjacent to the right edge of $\rho$ (which is contained in the left edge of~$\Rud$). If such a component $\phi$ exists,
then $\phi = \rho \cap \gamma(R)$ and (like the tiny components)
$\phi$ is $xy$-monotone. Other components of $\rho \cap \freesp$ are tiny components of $\I(R)$. Let $\rho'$
be the other corner square in $\Rud^-$. If both $\rho \cap \freesp$ and $\rho' \cap \freesp$ have non-empty connected components
adjacent to the left edge of $R$, denoted by $\phi$ and $\phi'$ respectively, then we refer to pairs of
points in $\phi \times \phi'$ as \emph{blocked pairs}.
%----------------------------------------------------------------------------
\begin{figure}
	\centering
	\includegraphics{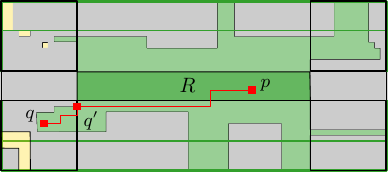}
	\caption{An example influence region within $\Rsq$ (not to scale), with the giant
 component $\gamma(\R)$ drawn in green 
 and the tiny components in the corner squares in yellow. 
	In red, an illustration of a $pq$-path from \protect\lemref{xy-shortest path}.
	%Example of how tiny components can arise in a corner square. The obstacles are shown in dark gray and their Minkowski sum with a unit square is shown in light gray.  and one of the robots is shown in red. 
	}
	\label{fig:tiny-giant}
\end{figure}
%----------------------------------------------------------------------------
 
A crucial property of $\I(R)$, which we prove in the following lemmas, is that any shortest path % (under the $\ell_1$ metric)
for a pair of points $p,q$ lying in the same connected component of $\I(R)$ is $x$-monotone
if and only if $(p,q)$ is not a blocked pair. Note that any $xy$-monotone path is
trivially a shortest path between its endpoints.
%-----------------------------------------------------------------------------
\begin{lemma}\label{lem:xy-shortest path}
	Let $p \in \R$ and $q \in \gamma(\R)$. Then any shortest path from $p$ to $q$ in $\freesp$ is
	$xy$-monotone.
\end{lemma}
%-----------------------------------------------------------------------------
\begin{proof}
	Recall that the upper and lower boundaries of $\Rud\cap \freesp$ are $x$-monotone chains. Therefore there is an $L$-shaped
	$pq$-path if $q \in \gamma(\R) \cap \Rud$. Next, assume $q$ lies in a corner square~$\rho$, say,
	the bottom-left corner square of $\Rsq$. Since $q \in \gamma(\R)$, we know that $q$ lies in the component
	$\phi$ of $\freesp \cap \rho$ adjacent to the left edge
	of $\R$ and this component is $xy$-monotone; see above.
	Hence, there is an $xy$-monotone path from $q$ to the top-right corner $q'$ of $\phi$, which can be extended to an $xy$-monotone
	path to $p$, as $q'$ lies in $\gamma(\R) \cap \Rud$. See \figref{tiny-giant}(ii).
 % (and hence there is an $xy$-monotone $pq'$-path as before).
\end{proof}
%-----------------------------------------------------------------------------
\begin{lemma}\label{lem:x-mono}
	Let $p,q\in\fre$ and let $h\subset\fre$ be a horizontal segment.
    Let $p'$ be the point on $h$ whose Euclidean distance to $p$ is minimum,
    and define $q'$ similarly for~$q$.
    Suppose the following conditions hold.
    \begin{enumerate}[(i)]
    \item $|p_y - y(h)|\leq 1$ and $|q_y - y(h)|\leq 1$.
    \item Any shortest path from $p$ to $p'$ is $xy$-monotone, and any  
          shortest path from $q$ to $q'$ is $xy$-monotone.
    \item The points $p$ and $q$ do not lie to the same side of the
          vertical slab $[x_1(h),x_2(h)]\times[-\infty,+\infty]$,
          where $x_1(h)$ and $x_2(h)$ denote the minimum and maximum $x$-coordinate of~$h$.
          In other words, $\max(p_x,q_x)\geq x_1(h)$ and $\min(p_x,q_x)\leq x_2(h)$. 
    \end{enumerate}  
    Then any shortest path from $p$ to $q$ in $\freesp$ is $x$-monotone.
\end{lemma}
%----------------------------------------------------------------------------
\begin{proof}
  Let $\pth(p,p')$ and $\pth(q',q)$ be shortest paths from $p$ to $p'$ 
  and from $q'$ to $q$, respectively. Observe that if $p_x\in[x_1(h),x_2(h)]$ then
  $\pth(p,p')$ is a vertical segment by condition~(i) and \obsref{free-segments};
  a similar statement holds for $\pth(q,q')$. Hence,
  conditions~(ii) and~(iii) of the lemma together imply that $\pth':= \pth(p,p') \circ p'q'\circ \pth(q',q)$,
  which is a $pq$-path in~$\fre$,  is $x$-monotone. To prove the lemma, we will argue that
  any $pq$-path that is not $x$-monotone must be longer than~$\pi'$.
  \medskip

  Let $\ell(h)$ be the horizontal line containing~$h$.
  If $p$ and $q$ lie on opposite sides of~$\ell(h)$ or one of them lies on $\ell(h)$,
  then $\pth'$ is $xy$-monotone, and so any shortest $pq$-path must be $xy$-monotone as well.
        %\ben{I think not necessarily.. $\pi(p,p')$ and $\pi(q,q')$ are both definitionally $xy$-monotone, but $\pi'$ doesn't have to be as defined, right? Both can lie inside the vertical slab, so what if their $x$-coordinates are equal, $p$ and $q$ lie on opposing sides of $\ell(h)$, and are separated by an obstacle? It's possible to have obstacles between them per the lemma statement, since $p$ and $q$ could each be $y$-coordinate 1 away from $\ell(h)$, and have $|p_y - q_y| \leq 2$. This shouldn't be possible with how we apply the lemma though, since being on opposite sides of $h$ would interact with $R$, which doesn't contain obstacles.}
        %\mdb{But $p',q'\in h$ and $h\subset \fre$. So in the case you describe, we have a point in between $p$ and $q$ (namely $p'$, which equals $q'$) that is in $\fre$ and at within distance 1 from $p$ and $q$. Hence, there cannot be an obstacle between $p$ and $q$.}
        %\ben{Ah, yep, thanks. This coupled with (iii), and the original statement gives monotonicity.}
  It remains to prove the lemma for the case where $p$ and $q$ lie 
  to the same side of $\ell(h)$, say above.
  
  For a path $\pth$, define $\cost{\pth}_x$ and $\cost{\pth}_y$
  to be the total length of the horizontal and vertical
  segments, respectively, of $\pth$. 
  Note that  $\cost{\pth} = \cost{\pth}_x + \cost{\pth}_y$. 
  Since we assumed that $p$ and $q$ both lie above~$\ell(h)$, we have
  \[
  \cost{\pth'}_y = 2\cdot \max \big\{ \yco{p}- y(h), \yco{q}- y(h) \big\} - |\yco{p} - \yco{q}| \leq 2 - |\yco{p} - \yco{q}|.
  \]
  
Now consider a shortest $pq$-path $\pth$ and suppose for a contradiction that $\pth$ 
is not $x$-monotone. With a slight abuse of notation, we also regard $\pth$ 
as a function $\pth:[0,T]\rightarrow \freesp$ that is a parameterization of the path $\pth$. 
Let $\lambda, \mu \in [0,T]$ be such that 
$\xco{\pth(\lambda)} = \xco{\pth(\mu)}$ and such that there is a time 
$\lambda' \in [\lambda, \mu]$ with $\xco{\pth(\lambda')} \not= \xco{\pth(\lambda)}$.
 If the vertical segment $\pth(\lambda)\pth(\mu)$ is contained in $\freesp$ then 
 we could shortcut~$\pth$, contradicting that $\pth$ is a 
 shortest path. Hence, $\pth(\lambda)\pth(\mu)\not\subset\freesp$ and so
 $|\pth(\lambda)_y - \pth(\mu)_y| > 1$ by \obsref{free-segments}.
 We claim that $\cost{\pth}_y > 2 - |p_y - q_y|$. 
 This will imply that $\cost{\pth'}_y < \cost{\pth}_y$,
  which gives the desired contradiction since the $x$-monotonicity of $\pth'$
  implies that $\cost{\pth'}_x \leq \cost{\pth}_x$. 
 
  To prove the claim, let $\overline{\pth}$ be an L-shaped path that connects $p$ to $q$ 
  (which need not lie completely in the free space~$\fre$). Clearly, $\| \overline{\pth}\|_y = |p_q-q_y|$.
  Then $\pth \circ \overline{\pth}$
  is a (not necessarily simple) cycle that passes through 
  $\pth(\lambda)$ and $\pth(\mu)$. 
  Therefore 
  \[
  \|\pth \circ \overline{\pth}\|_y \geq 2 |\pth(\lambda)_y - \pth(\mu)_y| > 2.
  \]
  On the other hand,
  \[
  \|\pth \circ \overline{\pth}\|_y = \cost{\pth_y} + \cost{\overline{\pth}_y} = \cost{\pth}_y + |p_y - q_y|,
  \]
  and so the claim follows.
\end{proof}
%-----------------------------------------------------------------------------
\begin{lemma}\label{lem:x-geodesic}
	Let $p,q\in\gamma(\R)$ be such that $p,q$ both lie in $\Rxy^+$ or both lie in $\Rxy^-$.
    Then any shortest path from $p$ to $q$ in $\freesp$ is $x$-monotone.
\end{lemma}
%----------------------------------------------------------------------------
\begin{proof}
 Without loss of generality, let $p,q \in \Rxy^+$. 
 
 If $p$ and $q$ lie in the same corner square, $\rho$, of~$\Rxy^+$ then we are 
 immediately done. Indeed, $\rho\cap \Rxy^+$ is $xy$-monotone, which implies
 that any shortest $pq$-path must be $xy$-monotone. 
 
 On the other hand, if $p$ and $q$ do not lie in the same corner square,
 then the conditions of \lemref{x-mono} are satisfied for $h:= \topR$,
 and we are done as well. Indeed, condition~(i) is satisfied because $p,q \in \Rxy^+$,
 condition~(ii) is satisfied because of the properties of $\Rxy^+ \cap\fre$,
 and condition~(iii) is satisfied because $p$ and $q$ do not lie in the same corner square.
\end{proof}

%----------------------------------------------------------------------------
\begin{figure}
	\centering
	\includegraphics{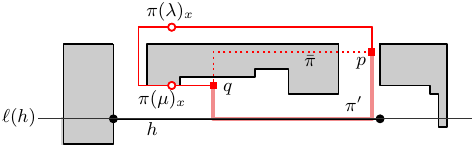}
	\caption{Illustration of the proof of 
   \protect\lemref{x-mono}, using 
   $\overline{\pi}$ and $\pi'$, where $h \subset \F$ is the horizontal segment between the black points.
   \lemref{x-geodesic} follows
   from setting $h = \topR$.
    }
	\label{fig:x-geodesic}
\end{figure}
%----------------------------------------------------------------------------
%

%-----------------------------------------------------------------------------------
We are now ready to prove the main property of the influence region.
%---------------------------------------------------------------------------
\begin{lemma}
\label{lem:xy-monotone-iff}
	For a pair of points $p,q$ lying in the same connected component of $\I(\R)$, any shortest path from $p$ to $q$
    in $\freesp$ is $x$-monotone if and only if $(p,q)$ is not a blocked pair.
\end{lemma}
%---------------------------------------------------------------------------
\begin{proof} 
    If $(p,q)$ is a blocked pair lying in $\Rud^-$, then any path from $p$ to $q$ must leave $\Rud^-$
	because each connected component of $\Rud^- \cap \freesp$ either lies in the top or bottom corner square of $\Rud^-$,
	and any such path is not $x$-monotone.
\medskip

	Next, assume that $(p,q)$ is not a blocked pair, and let $\pth$ be a shortest path from $p$ to~$q$. We argue that $\pth$ is $x$-monotone.
	If $p,q$ lie in a tiny component of $\I(\R)$, then $\pth$ is obviously $xy$-monotone because, as mentioned above,
	each tiny component of $\I(\R)$ is $xy$-monotone.
	We can therefore assume that $p,q \in \gamma(\R)$.

	If at least one of $p,q$ lies in $\R$, then \lemref{xy-shortest path} implies the claim. So, assume
	that neither $p$ nor $q$ lies in $\R$. If both $p$ and $q$ lie in $\Rxy^+$ or both lie in $\Rxy^-$,
	then \lemref{x-geodesic} implies the claim. We now assume wlog that $p \in \Rxy^+$
	and $q \in \Rxy^-$. Since $(p,q)$ is not a blocked pair, $p$ and $q$ cannot both lie in
	$\Rud^-$, and they cannot both lie in $\Rud^+$. Therefore, $p$ and $q$ lie in opposing quadrants of $\gamma(R)$.
    Let $r$ be an arbitrary point in $R$, and let $\pth(p,r)$ and $\pth(r,q)$ be shortest
    paths from $p$ to $r$ and from $r$ to $q$, respectively. 
    Then $\pth(p,r)$ and $\pth(r,q)$ are $xy$-monotone by \lemref{xy-shortest path}, 
    and so $\pth(p,r)\cup\pth(r,q)$ 
    is an $xy$-monotone path from $p$ to~$q$.
	Hence, if $\pth$ is not $xy$-monotone then it is not a shortest path.
	% Putting all the cases together, we conclude that $\pi$ is $x$-monotone. 
    This completes the proof of the lemma.
\end{proof}
%---------------------------------------------------------------------------

%----------------------------------------------------------------------------
\subsection{Unsafe and swap intervals}\label{subsec:unsafe}
%----------------------------------------------------------------------------
Recall that $R$ is a corridor containing a bad horizontal segment~$e$ of $\pi_A$, 
that $\pi_A(\lambda_1)$ is the last point on $\pi_A$ before $e$ that lies on 
$\topR$ or $\botR$, and that $\pi_A(\lambda_2)$ is the first point on $\pi_A$ after $e$ 
that lies on $\topR$ or $\botR$.
We now define critical subintervals of the time interval $[\lambda_1, \lambda_2]$ during which
the surgery of $\plan$ requires more care, and then prove some structural properties of $\plan$ during
these intervals. Again, we begin with a few definitions. A configuration $\bp = (p_A, p_B)$
is called \emph{$x$-separated} if $|(p_A)_x - (p_B)_x| \geq 1$ and \emph{$y$-separated}
if $|(p_{A})_y- (p_{B})_y| \geq 1$. We say that two $x$-separated
configurations $\bp,\bq$  have the \emph{same $x$-order} if
$\sign((p_{A})_x - (p_{B})_x) = \sign((q_{A})_x - (q_{B})_x)$;
%, i.e., if $A$ is to the right of $B$ at $p$ then it is the same at $q$; 
otherwise the $x$-order is \emph{swapped}.
An interval $[\nu_1,\nu_2]$ is called 
\emph{$x$-separated} if $\plan(\nu_1)$ and $\plan(\nu_2)$ are both $x$-separated; note that $\plan(\nu)$
is not necessarily $x$-separated for all $\nu\in [\nu_1,\nu_2]$.

%----------------------------------------------------------------------------
\begin{figure}
	\centering
    \includegraphics{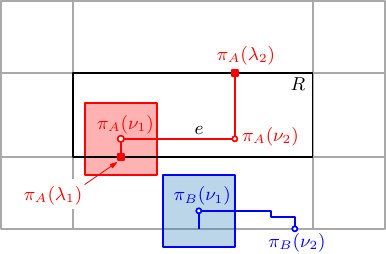}
    \hspace{1cm}
    \includegraphics{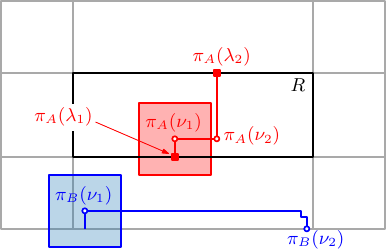}
	\caption{(Left) An unsafe swap interval $[\nu_1,\nu_2] \subseteq [\lambda_1, \lambda_2]$
                    The robots become $y$-separated at time $\nu_1$ due to the vertical movement of~$\robA$. 
                    At time~$\nu_2$, robot~$\robB$ leaves $\I(\R)$.
    (Right) A swap interval, which is an unsafe interval with different $x$-orders at its endpoints.}
	\label{fig:unsafe}
\end{figure}
%----------------------------------------------------------------------------
% Recall that $e$ is a horizontal bad segment of $\pth_A$
% and that $(\lambda_1,\lambda_2)$ is the maximal time interval such that 
% $e\subset\pth_A(\lambda_1,\lambda_2)$ and
% $\pth_A(\lambda_1,\lambda_2)\subset R\setminus(\topR\cup\botR)$,
% where $R$ is the corridor containing~$e$.
%----------------------------------------------------------------------------
\begin{defn}
An interval $[\nu_1,\nu_2]\subseteq [\lambda_1,\lambda_2]$ is \emph{unsafe} if 
\begin{enumerate}[(i)]
\item $[\nu_1,\nu_2]$ is a maximal interval such that  $\pth_B(\nu)\in \I(\R)$ and $\plan(\nu)$ is $y$-separated
		for all $\nu\in [\nu_1,\nu_2]$; and 
\item there is a time $\nu\in [\nu_1,\nu_2]$ such that $\plan(\nu)$ is not $x$-separated.
\end{enumerate}
\end{defn}
%----------------------------------------------------------------------------
\smallskip
Observe that for any unsafe interval~$[\nu_1,\nu_2]$, robot~$B$ either lies above $A$ throughout the
interval---that is, $\yco{\pth_B(\nu)} \geq \yco{\pth_A(\nu)}$ for all $\nu\in[\nu_1,\nu_2]$---or 
$B$ lies below $A$ throughout the interval. In fact, we can make the following stronger 
observation.
%----------------------------------------------------------------------------
\begin{lemma} \label{lem:BinZ}
If $[\nu_1,\nu_2]$ is an unsafe interval then $\pth_B[\nu_1,\nu_2]$ lies in a 
connected component of $Z^+\cap\fre$ or in a connected component of~$Z^-\cap\fre$.
\end{lemma}
%----------------------------------------------------------------------------
\begin{proof}
Because the robots must remain $y$-separated during $[\nu_1,\nu_2]$
and $\pth_B[\nu_1, \nu_2] \subset \I(\rect)$,
they cannot change 
their $y$-order. Now assume wlog that $\pth_B(\nu)_y \leq \pth_A(\nu)_y-1$ for 
all $\nu\in [\nu_1,\nu_2]$. If $\pth_A(\nu) \not\in\topR$, then this immediately
implies that $\pth_B(\nu) \in Z^-$. Thus, the only possibility for $\pth_B(\nu)$
not to be in $Z^-$ is when $\nu=\nu_1=\lambda_1$ or $\nu=\nu_2=\lambda_2$.
By definition of~$[\lambda_1,\lambda_2]$ we know that $\pth_A(\nu) \not\in\topR$ 
for any $\nu\in (\lambda_1,\lambda_2)$. Since the plan is decoupled,
robot~$\robB$ is parked just after after time $\lambda_1$ and just before time~$\lambda_2$,
and so $\pth_B(\nu)$ must be in $Z^-$ at times $\nu=\nu_1=\lambda_1$ and $\nu=\nu_2=\lambda_2$ as well.
\end{proof}
%----------------------------------------------------------------------------
% Also observe that if a time $\lambda\in[\lambda_1,\lambda_2]$ does
% not lie in an unsafe interval, then $\plan(\lambda)$ is $x$-separated or
% $\pth_B(\lambda)$ lies outside~$\I(\R)$.
% Hence, for such a time~$\lambda$ we can safely push $\pth_A(\lambda)$ to $\topR$ or $\botR$
% without creating a collision with~$B$.
%----------------------------------------------------------------------------
\begin{lemma}\label{lem:unsafe-vertical-edge}
	Let $[\nu_1, \nu_2]$ be an unsafe interval. For $X \in \{A,B\}$ and 
    for $i\in\{1,2\}$, the point~$\pth_X(\nu_i)$ lies on a (possibly zero-length) 
    vertical segment of $\pth_X$ or on a vertical 1-line.
    % \mdb{In fact the second option can also be stated as "\ldots or on the intersection of a vertical 1-line and a horizontal segment of $\pth_X$." But perhaps we do not need this.}
\end{lemma}
%----------------------------------------------------------------------------
\begin{proof}
Let $i=1$; the argument for $i=2$ is symmetric.
There are three cases for $\nu_1$ to be an endpoint of an unsafe interval.

The first case is that $\pth_B$ enters $\I(\R)$ at time $\nu_1$.
If it does so through a horizontal edge of $\I(\R)$ then clearly
$\pth_B(\nu_1)$ lies on a vertical segment of~$\pth_B$, and if
$\pth_B$ enters through a vertical edge of $\I(\R)$ then 
$\pth_B(\nu_1)$ lies on a vertical 1-line.
Furthermore, $\robA$ is parked when $\pth_B$ enters $\I(\R)$
because the plan $\plan$ is decoupled. Since $\plan$ is alternating,
this implies that $\pth_A(\nu_1)$ lies on a vertical segment of~$\pth_A$.

The second case is that $\nu_1=\lambda_1$. Then robot~$\robA$ must be moving vertically
just after time~$\nu_1$, and we can apply the same argument as above.

The third case is that $\robA$ and $\robB$ become $y$-separated at time~$\nu_1$. 
Then one of the two robots, say $\robA$, must be moving vertically 
immediately before time $\nu_1$. Hence, $\pth_A(\nu_1)$ lies on a vertical segment of~$\pth_A$.
The other robot, $\robB$, must then be parked at time $\nu_1$, which again implies that
it lies on a vertical segment of its path.
\end{proof}
%----------------------------------------------------------------------------
Consider an unsafe interval~$[\nu_1,\nu_2]$.
If $\nu_1 \not= \lambda_1$ and $\nu_2 \not= \lambda_2$, and
$\robB$ does not enter $\I(\R)$ at time $\nu_1$ and
and $\robB$ does not leave $\I(\R)$ at time~$\nu_2$,
then by the maximality condition,
$\plan(\nu_1)$ and $\plan(\nu_2)$ are $x$-separated configurations---this is true
because a configuration that is not $y$-separated must be $x$-separated. 
An important type of unsafe intervals are so-called swap intervals, as defined next.
%----------------------------------------------------------------------------
\begin{defn}
A \emph{swap interval} is an $x$-separated unsafe interval such that the $x$-orders
at $\plan(\nu_1)$ and $\plan(\nu_2)$ are different.
\end{defn}
%----------------------------------------------------------------------------
%
Roughly speaking, the role of unsafe intervals is as follows.
If a time $\lambda\in [\lambda_1,\lambda_2]$ does not lie in an 
unsafe interval, translating $\pth_A(\lambda)$ vertically within $\R$ 
does not cause the point to conflict with $\pth_B(\lambda)$ because $\plan(\lambda)$ is $x$-separated. 
If an unsafe interval is not $x$-separated, 
it turns out that we can modify $\pth_B$ to avoid conflicts 
after translating $\pth_A(\lambda)$ vertically without increasing the length of~$\pth_B$.
If a swap interval~$I$ 
is $x$-separated and the two endpoints have the same $x$-order
then we show in \lemref{plan-between-swap-endpoints} that one can re-parametrize the plan so that $I$ is not an unsafe
interval and we can do surgery as above. 
So the challenging case is when $I$ is
a swap interval. In this case, the surgery is considerably more involved, but
we will keep the situation under control by proving some 
desirable properties of unsafe and swap intervals.
To prove the existence of an optimal plan with these properties,
we modify a given optimal plan in a controlled manner, as defined
below. 

%----------------------------------------------------------------------------
\begin{defn}
Consider an interval $[\lambda, \mu]$. We call $\plan^*[\lambda, \mu]$ a 
\emph{compliant modification} of $\plan[\lambda, \mu]$ if the 
{\sc Feasibility}~(P1), {\sc Optimality}~(P2), {\sc Vertical alignment}~(P4) 
and {\sc Alternation}~(P5) properties hold, and the following {\sc No Regress} property:
each bad segment of $\plan^*$ is already present in $\plan$ or 
it is a segment of $\pi^*_A[\lambda_1,\lambda_2]$.
\end{defn}
%----------------------------------------------------------------------------
The fact that a compliant modification has the {\sc No Regress} property instead
of the {\sc Progress} property is not problematic, because when handling a bad segment~$e$
of the modified plan $\plan^*$,
we will actually get rid of all bad segments of $\pi^*_A[\lambda_1,\lambda_2]$.

A special type of compliant modification is a \emph{compliant re-parametrization}.
Here $\plan^*[\lambda, \mu]$ is a re-parametrization of $\plan^*[\lambda, \mu]$, that is, 
the paths traced by $A$ and $B$ remain the same but the
parametrization of $\plan^*$ is different from $\plan$
and thus the breakpoints and parking spots in $\plan^*$ may change. Note that
a re-parametrization will never introduce a collision of a robot with an obstacle.
Thus, to check the {\sc Feasibility} property, we only need to ensure that the robots do not collide with each other.
Also note that {\sc Optimality} is automatically satisfied in a re-parametrization.
We can ensure the {\sc Alternation}
property by adding zero-length segments at breakpoints, where
necessary.
The remaining properties, {\sc Vertical alignment} and {\sc No Regress},
are only violated if these zero-length segments are unaligned.
To verify that such violations do not happen,
it suffices to prove that 
the parking spots in a compliant re-parametrization~$\plan^*[\lambda, \mu]$ are of one of the following types:
{
\renewcommand{\labelenumi}{(T\arabic{enumi}):}
\begin{enumerate}
\item grid points,
\item vertices of $\pth_A[\lambda, \mu]$ or $\pth_B[\lambda, \mu]$, or
\item intersection points of horizontal (resp. vertical) grid lines with vertical (resp. horizontal) segments of
$\pth_A[\lambda, \mu]$ or $\pth_B[\lambda, \mu]$.
\end{enumerate}
}
%----------------------------------------------------------------------------
Note that compliant modification is composable: if $\plan'$ is a compliant modification of $\plan$
and $\plan''$ is a compliant modification of $\plan'$, then
$\plan''$ is a compliant modification of~$\plan$. Thus we can apply
compliant modifications repeatedly while preserving the desired properties
with respect to the original path.
\medskip

In \lemref{unsafe-is-swap}, we perform a sequence of compliant modifications to $\plan$
to ensure that the unsafe swap intervals in the resulting plan
have certain desirable properties.
We begin with the following technical lemma that we will
use repeatedly.
% \mdb{We used $\phi$ for component of the free space at some point, and now we use $\varphi$ for paths.}
%----------------------------------------------------------------------------
\begin{lemma}\label{lem:x-sep}
	Let $\bp = (p_A, p_B)\in \BF$ and $\bq = (q_A,q_B) \in \BF$ be two $x$-separated configurations that
	have the same $x$-order. If there is an $x$-monotone path $\varphi_A \subset \freesp$ from $p_A$ to $q_A$
	and an $x$-monotone path $\varphi_B \subset \freesp$ from $p_B$ to $q_B$,
	then there is a decoupled, feasible plan $\varphi:[0,1] \to \BF$ in which $A$ 
	first follows $\varphi_A$ and then $B$ follows $\varphi_B$, or vice versa.
    Moreover, $\varphi(\lambda)$ is $x$-separated
	for all $\lambda \in [0,1]$ in the plan. The claim also holds if $x$-direction is replaced with the $y$-direction.
\end{lemma}
%----------------------------------------------------------------------------
\begin{proof}
	We prove the lemma for the $x$-direction; the case of $y$-direction is symmetric.
	Without loss of generality, assume that $(q_{A})_x \geq (p_{A})_x$. There are two cases,
    illustrated in \figref{xmon-order}. 
    \begin{figure}
	\centering
	\includegraphics{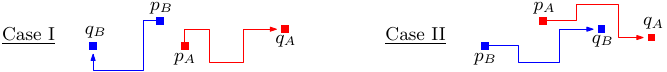}
		\caption{\emph{Case~I:} $A$ and $B$ move in opposite directions. Since $\bp$ and $\bq$ have the same
        $x$-order, the movements do not interfere with each other. \emph{Case~II:} $A$ and $B$ move in the same direction. Since $\bp$ and $\bq$ have the same
        $x$-order and we assumed $(q_{A})_x \geq (p_{A})_x > (p_B)_x$, robot~$A$ can be moved before~$B$.}
		\label{fig:xmon-order}
	\end{figure}
	\\[2mm]
	 \emph{Case~I:  $(q_{B})_x < (p_{B})_x$.} Since $\bp$ and $\bq$ have the same $x$-order, $\bp,\bq \in \BF$,
     and they are $x$-separated,
		it is easily seen that the $x$-intervals $[(p_{A})_x, (q_{A})_x]$ and $[(q_{B})_x, (p_{B})_x]$
		are disjoint and at least distance~1 apart. Therefore, we can first move $A$ from $p_A$
		to $q_A$ along $\varphi_A$ and then $B$ from $p_B$ to $q_B$ along $\varphi_B$.
		Since the two intervals are at least distance~1 apart, the resulting plan $\varphi$ is feasible and
		$\varphi(\lambda)$ is $x$-separated for all $\lambda \in [0,1]$.
	\\[2mm]	
	\emph{Case~II: $(q_{B})_x \geq (p_{B})_x$.} Without loss of generality, assume that $(p_{A})_x > (p_B)_x$;
		otherwise, we can switch the roles of $A$ and $B$. We first move $A$ along $\varphi_A$
		and then $B$ along $\varphi_B$. Note that for any point $z \in \varphi_A$, we have
		$z_x - (p_{B})_x \geq (p_{A})_x - (p_{B})_x \geq 1$. Similarly, for any point $z \in \varphi_B$,
		we have $(q_{A})_x - z_x \geq (q_{A})_x - (q_{B})_x \geq 1$. Hence, the resulting plan $\varphi$ is feasible
		and $\varphi(\lambda)$ is $x$-separated for all $\lambda \in [0,1]$.
\end{proof}
%----------------------------------------------------------------------------
\begin{lemma}\label{lem:plan-between-swap-endpoints}
	Let $\lambda,\mu\in [\lambda_1,\lambda_2]$ be endpoints of $x$-separated unsafe intervals,
	not necessarily endpoints of the same unsafe interval, such that $\lambda < \mu$ and
	$\plan(\lambda), \plan(\mu)$ have the same $x$-order. If $\pth_B(\lambda)$, $\pth_B(\mu)$ 
	lie in the same component of $\I(\R)$ and are not a blocked pair,
	then there is a compliant re-parametrization $\plan^*[\lambda,\mu]$ of $\plan[\lambda, \mu]$ such that 
	the configuration $\plan^*(\nu)$ is $x$-separated for all $\nu\in [\lambda,\mu]$.
\end{lemma}
%----------------------------------------------------------------------------
\begin{proof}
Applying \lemref{x-sep} to the points $\bp := \plan(\lambda)$ 
and $\bq := \plan(\mu)$ to obtain a plan~$\plan^*[\lambda,\mu]$ that
is $x$-separated throughout
would create new parking spots in $\plan^*[\lambda,\mu]$
at the time instances $\lambda$ and $\mu$.
These parking spots may violate the properties of a
compliant re-parametrization: they are by definition
breakpoints of $\plan^*[\lambda,\mu]$, but they need not be breakpoints
of $\plan[\lambda,\mu]$ and thus may not lie on a grid line or be aligned with a vertical 
	segment of $\pth_A\cup\pth_B$. We therefore apply 
\lemref{x-sep} at carefully chosen time instances, $\lambda'$ and $\mu'$, defined as follows.
\begin{itemize}
\item If $\pth_B$ enters or leaves $\I(\R)$ at time $\lambda$
     or $\lambda=\lambda_1$,
    then we set $\lambda':=\lambda$. Otherwise, $\plan$ starts or stops 
    being $y$-separated at time~$\lambda$.
    Let $X\in \{A,B\}$ be the robot that moves at time~$\lambda$,
	thus causing the robots to start or stop being $y$-separated; $\pth_X(\lambda)$ lies
		on a vertical segment of $\pth_X$.
    We now set $\lambda'$ to be
    the first time after $\lambda$ such that $\pth_X(\lambda')$ is a vertex 
    of $\pth_X$ or $\pth_X(\lambda')$ lies on a horizontal grid line. 
    See Figure~\ref{fig:swap-lem-breakpoints}.
    Since the top and bottom edges of $\I(\rect)$ are horizontal grid lines, we conclude that
    $\pth_X(\lambda') \in \I(\rect)$, that
    $\pth_X(\lambda)\pth_X(\lambda')$ is a vertical segment, and that
    $\pth(\lambda)$ and $\pth(\lambda')$ have the same $x$-order.
\item The time $\mu'\in [\lambda,\mu]$ is defined symmetrically. 
      More precisely, if $\pth_B$ enters or leaves $\I(\R)$ at time $\mu$
    or $\mu=\lambda_2$,
      then we set $\mu':=\mu$. Otherwise, let $X\in \{A,B\}$ be the robot that 
      moves at time~$\mu$, causing the robots to start or stop being $y$-separated.
      We set $\mu'$ to the last time before $\mu$ such that $\pth_X(\mu')$ is a vertex of $\pth_X$ or $\pth_X(\mu')$ lies on a horizontal grid line. Again, we can argue that $\pth_x(\mu') \in \I(\rect)$, that $\pth_x(\mu)\pth_x(\mu')$
      is a vertical segment, and that $\plan(\mu)$ and $\plan(\mu')$ have
      the same $x$-order.
\end{itemize}
\begin{figure}
	\centering
	\includegraphics{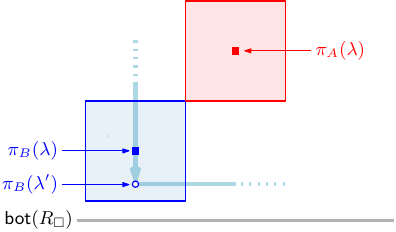}
	\caption{An example of the definition of $\lambda'$ in Lemma~\ref{lem:plan-between-swap-endpoints}. 
    Since $\pth_B$ does not enter or leave $\I(\rect)$ at time $\lambda$ and $\lambda\neq \lambda_1$, 
    we take the robot $X$ that moves at time $\lambda$ (here: $X=B$), and set $\lambda'$ to the first time
    after $\lambda$ where $\pth_B$ reaches a vertex or a grid line (here $\pth_B(\lambda')$ is a vertex of $\pth_B$).}
	\label{fig:swap-lem-breakpoints}
\end{figure}

Note that if $\lambda'> \mu'$, then $\plan[\lambda,\mu]$ is a plan in which one
robot moves along a vertical segment while the other is parked,
and setting $\plan^*:=\plan$ proves the claim.
So, assume $\lambda'\leq \mu'$ in the remainder of the proof.

Note that $\plan(\lambda)$ and $\plan(\lambda')$ have the same $x$-order,
and that the same holds for $\plan(\mu)$ and $\plan(\mu')$. This implies that
$\plan(\lambda')$ and $\plan(\mu')$ have the same $x$-order.
Since $\lambda',\mu'\in[\lambda_1,\lambda_2]$ and 
$\pth_A[\lambda_1,\lambda_2]\subset \R$,  we have $\pth_A(\lambda'), \pth_A(\mu') \in \rect$.
Hence, there trivially exists an $x$-monotone path $\varphi_A\subset\fre$ 
from $\pth_A(\lambda')$ to $\pth_A(\mu')$. 
By construction, as argued above,
$\pth_B(\lambda')\pth_B(\lambda)$ and $\pth_B(\mu')\pth_B(\mu)$ are
(possibly zero-length) vertical segments in $\I(\rect)$. 
Therefore $\pth_B(\lambda')$ and $\pth_B(\mu')$ lie in the same component of $\I(\rect)$
and are not a blocked pair, since these properties hold for $\lambda$ and $\mu$ by assumption.
Hence, by \lemref{xy-monotone-iff}, any shortest path from $\pth_B(\lambda')$ to $\pth_B(\mu')$ is $x$-monotone.
We can thus apply \lemref{x-sep} to $\bp := \plan(\lambda')$ 
and $\bq := \plan(\mu')$. This gives us a $(\bp, \bq)$-plan
$\plan^*[\lambda',\mu']$ in which $A$ only parks at $\pth_A(\lambda')$ or $\pth_A(\mu')$ and
$B$ only parks at $\pth_B(\lambda')$ or $\pth_B(\mu')$.\footnote{This plan is not 
necessarily a re-parametrization. However, since there exists a valid plan from 
$\plan(\lambda')$ to $\plan(\mu')$ that uses shortest paths in $\fre$, 
the paths $\pth_A[\lambda',\mu']$ and $\pth_B[\lambda',\mu']$ must be 
shortest paths in $\fre$, because $\plan$ is optimal. 
This means $\pth_A[\lambda',\mu']$ is $xy$-monotone and that, by \lemref{x-sep}, 
the path $\pth_B[\lambda',\mu']$ is $x$-monotone. So, we apply \lemref{x-sep} 
with these paths to obtain a plan with the required properties 
that is also a re-parametrization of $\plan$.}  
We extend the domain of $\plan^*$ to $[\lambda,\mu]$ by setting 
$\plan^*[\lambda,\lambda']:= \plan[\lambda,\lambda']$ and $\plan^*[\mu',\mu]:=\plan[\mu',\mu]$.
Observe that $\plan^*(\nu)$ is $x$-separated for all $\nu\in [\lambda,\mu]$.
Finally, we argue that the parking spots of $\plan^*[\lambda,\mu]$
are of type (T1)--(T3). We do this for the possible parking spots
at time $\lambda'$; the argument for $\mu'$ is similar. We consider three cases.
\begin{itemize}
\item 
    If $\pth_B$ enters or leaves $\I(\rect)$ at time $\lambda$, then $\lambda'=\lambda$.
    Hence, $\pth_A(\lambda')$ is a parking spot of~$\plan$ (which is allowed) and $\pth_B(\lambda')$ is
    the intersection of a segment of $\pth_B$ and an edge of $\I(\rect)$.
    Since the edges of $\I(\rect)$ are contained in grid lines, we conclude
    that the potential parking spot $\pth_B(\lambda')$ is also allowed.
\item
    Second, we may have $\lambda'=\lambda = \lambda_1$. 
    Since $\pth_A$ enters $\R$ at time $\lambda_1$, the point $\pth_A(\lambda_1)$ 
    lies on a vertical segment of $\pth_A$ and a horizontal grid line. 
    Additionally, since $A$ is moving at time $\lambda_1$, 
    the point $\pth_B(\lambda_1)$ is a breakpoint of $\pth_B$, and thus is allowed.
\item
    Finally, suppose that $\pth_B$ does not enter or leave $\I(\rect)$ at time~$\lambda$
    and that $\lambda\neq\lambda_1$.
    Then $\plan$ starts or stops being $y$-separated at time $\lambda$
    due to the movement of a robot~$X\in \{A,B\}$.
    Thus, the other robot is at a breakpoint of $\plan$ at time $\lambda'$. 
    Moreover, by construction, $\pth_X(\lambda')$ lies on a vertex of $\pth_X$, 
    or on the intersection of a vertical segment of $\pth_X$ and a horizontal grid line. 
    Thus, in this case, the parking spots at time~$\lambda'$ are allowed as well.
\end{itemize}
This concludes the proof of the lemma.
\end{proof}
%----------------------------------------------------------------------------
By applying \lemref{plan-between-swap-endpoints} repeatedly we obtain the following lemma.
%----------------------------------------------------------------------------
\begin{lemma}\label{lem:unsafe-is-swap}
	There is a compliant re-parametrization $\plan^*[\lambda_1, \lambda_2]$ of $\plan[\lambda_1, \lambda_2]$
	such that any $x$-separated unsafe interval is a swap interval.
\end{lemma}
%----------------------------------------------------------------------------
\begin{proof}
	Let $[\mu_0, \mu_1]$ be an $x$-separated unsafe interval with $\plan(\mu_0)$ 
    and $\plan(\mu_1)$ having the same $x$-order. The points $\pth_B(\mu_0), \pth_B(\mu_1)$ 
    lie in the same component of $\I(\rect)$, since $\pth_B[\mu_0,\mu_1] \subset \I(\rect)$. 
    Additionally, $\plan$ is $y$-separated during $[\mu_0, \mu_1]$,
	so $B$ cannot completely cross $\R$ during $[\mu_0, \mu_1]$.
    This implies that $(\pth_B(\mu_0), \pth_B(\mu_1))$
	is not a blocked pair. Hence, by \lemref{plan-between-swap-endpoints}, 
    there is a compliant re-parametrization $\plan^*[\mu_0, \mu_1]$ of 
    $\plan[\mu_0, \mu_1]$ such that $\plan^*[\mu_0, \mu_1]$ is $x$-separated
	at all times. We now modify $\plan$ by replacing $\plan[\mu_0, \mu_1]$ 
    with $\plan^*[\mu_0, \mu_1]$. Since $\plan^*[\mu_0, \mu_1]$ is $x$-separated, 
    the interval $[\mu_0, \mu_1]$ is no longer unsafe.
	We repeat this procedure for every unsafe interval
	where the configuration at the endpoints have the same $x$-order. 
    After doing so, any $x$-separated unsafe interval is a swap interval. 
\end{proof}
%----------------------------------------------------------------------------
For the remainder of this section, we assume that we have applied the compliant 
re-parametrizations of \lemref{unsafe-is-swap} to $\plan$. We now, finally, prove the main property of swap intervals.
Recall that during an unsafe interval, the robots are always $y$-separated.
Hence, either $A$ is above $B$ throughout an unsafe interval, or vice versa.
%----------------------------------------------------------------------------
\begin{lemma}\label{lem:limit-swap-intervals}
    There is a compliant re-parametrization $\plan^*[\lambda_1, \lambda_2]$ of 
    $\plan[\lambda_1,\lambda_2]$ such that
    $\plan^*[\lambda_1, \lambda_2]$ has at most one swap interval with $B$ above $A$ 
    and at most one swap interval with $A$ above $B$.
\end{lemma}
%----------------------------------------------------------------------------
\begin{proof}
Suppose $\plan[\lambda_1,\lambda_2]$ has more than one swap interval with $A$ above $B$. 
Let $[\mu_0, \mu_1]$ and $[\mu_2,\mu_3]$ be the first and the last such intervals.
During both swap intervals, $\pth_B$ lies in the giant 
component $\gamma(\R)$ of $\I(\rect)$. Indeed, if $\pth_B[\mu_0,\mu_1]$
is contained in a tiny component of $\I(\rect)$, then $\pth_B[\mu_0,\mu_1]$
either lies entirely to the left or entirely to the right of $\R$, and
$\plan(\mu_0)$ and $\plan(\mu_1)$ have the same $x$-order,
contradicting that $[\mu_0, \mu_1]$ is a swap interval.
A similar argument applies to~$[\mu_2,\mu_3]$.

Since $[\mu_2, \mu_3]$ is a swap interval,
there is an $i\in \{2,3\}$ such that $\plan(\mu_0)$ and 
$\plan(\mu_i)$ have the same $x$-order.
Furthermore, $B$ lies below $A$ in both swap intervals and, hence, 
$(\pth_B(\mu_0),\pth_B(\mu_i))$ is not a blocked pair. 
By \lemref{plan-between-swap-endpoints}, there is a compliant
re-parametrization $\plan^*[\mu_0,\mu_i]$ of $\plan[\mu_0,\mu_i]$ such
that $\plan^*[\mu_0,\mu_i]$ is $x$-separated at all times, and thus
there are no unsafe intervals during $[\mu_0, \mu_i]$. 
Hence, $\plan^*$ contains at most one swap interval with $A$ lying above $B$.	
(Recall that $[\mu_2,\mu_3]$ was defined to be the last interval with $A$ above~$B$.)
A symmetric argument implies that there is a compliant re-parametrization $\plan^*$
that, in addition to having at most one swap interval with $A$ above $B$,
also has at most one swap interval with $A$ below~$B$.
\end{proof}
%----------------------------------------------------------------------------
Next, we show that if $\plan[\lambda_1, \lambda_2]$ has two swap intervals, then we can either eliminate at least one of them or $\plan[\lambda_1, \lambda_2]$
has additional structure, which will help with the surgery and with proving
its correctness.
%----------------------------------------------------------------------------
\begin{lemma}\label{lem:central-move-no-swap}
    Suppose there are exactly two swap intervals $[\mu_0,\mu_1]$ and $[\mu_2,\mu_3]$ in $[\lambda_1,\lambda_2]$. 
    If $\plan(\mu_1)$ and $\plan(\mu_2)$ have the same $x$-order and $(\pth_B(\mu_1),\pth_B(\mu_2))$ 
    is not a blocked pair, then there is a compliant modification $\plan^*$ such that $\plan^*[\lambda_1,\lambda_2]$ has no swap intervals.
\end{lemma}
%----------------------------------------------------------------------------
\begin{proof}
    Note that $\plan(\mu_0)$ and $\plan(\mu_3)$ have the same $x$-order,
    and that $\pth_B(\mu_1),\ldots,\pth_B(\mu_3)$ all lie in the giant component 
    of~$\I(\rect)$ because a swap interval can happen only when $B$
    is inside the giant component.
    If $(\pth_B(\mu_0),\pth_B(\mu_3))$ is not a blocked pair, then we can apply
    \lemref{plan-between-swap-endpoints} to $\mu_0,\mu_3$ to obtain a 
    compliant modification $\plan^*$ without swap intervals, so from now on
    we assume that $(\pth_B(\mu_0),\pth_B(\mu_3))$ is a blocked pair. 
    We furthermore assume wlog that $\pth_B(\mu_0)$ lies in the bottom-left corner square of $\Rsq$
    and that $\pth_B(\mu_3)$ lies in the top-left corner square.
    %$\pth_B(\mu_0) \in \Rud^-\cap \Rxy^-$ and  $\pth_B(\mu_3)\in \Rud^-\cap \Rxy^+$.
    We define five time instances $\mu', \xi_1, \xi_2, \nu_1,\nu_2$
    such that
    \[\mu_0 \leq \xi_1 \leq \nu_1 \leq \mu_1 \leq \mu' \leq \mu_2 \leq \nu_2 \leq \xi_2 \leq \mu_3\]
    and we construct a new plan $\plan^*$ by performing a compliant modification 
    on $\plan[\xi_1, \xi_2]$ such that $\plan^*(\mu)$ is $x$-separated for all
    $\mu \in [\mu_0, \mu_3]$, thereby implying that
    $\plan^*[\lambda_1, \lambda_2]$ does not have any swap intervals.
    \begin{itemize}
    \item 
    We first define $\mu'$.
    Observe that $\pth_B(\mu_0)$ lying in the bottom-left corner square 
    implies that $\pth_B(\mu_1)\in \Rxy^-$ because $[\mu_0,\mu_1]$ is an unsafe interval
    and thus $\pth_B[\mu_0,\mu_1]$ cannot cross~$\R$.
    Similarly, $\pth_B(\mu_3)$ lying in the top-left corner square 
    implies that $\pth_B(\mu_2)\in \Rxy^+$.
	Since $\pth_B(\mu_1)$ and $\pth_B(\mu_2)$ lie to the right of $\pth_A(\mu_1)$ and $\pth_A(\mu_2)$ respectively---this is true because the $x$-orders of $\plan(\mu_1)$ and $\plan(\mu_2)$
	are different from $\plan(\mu_0)$ and $\plan(\mu_3)$ respectively---we have
    $\pth_B(\mu_1)\notin \Rud^-$ and $\pth_B(\mu_2)\not\in \Rud^-$. 
    Because $(\pth_B(\mu_1),\pth_B(\mu_2))$ is not a blocked pair, there
    is a $\mu'\in \{\mu_1,\mu_2\}$ such that $\pth_B(\mu')\not\in\Rud^+$.
    Thus, $\pth_B(\mu')\in \Rud$.
    \item 
    Next, we define $\xi_1, \xi_2$.
    Let $V^*$ denote the intersection points of a horizontal segment of
    $\pth_A \cup \pth_B$ or a horizontal grid line with a vertical segment of $\pth_A \cup \pth_B$ or a vertical grid line. Then,
    \[
    \xi_1:= \min\{\xi\in [\mu_0,\mu_1]: \{\pth_A(\xi), \pth_B(\xi)\} \subset V^* \}.
    \]
    Note that $\xi_1$ is well defined. Indeed, one of the robots, say $X$, is 
    moving immediately after time $\mu_0$ while the other is parked, and, by definition, a 
    parking spot is a point of~$V^*$. 
    By \lemref{unsafe-vertical-edge}, $\pth_X(\mu_0)$ lies on a vertical edge
    of $\pth_X$ or on a vertical 1-line. If $\pth_X(\mu_0)$
    lies on a horizontal segment of~$\pth_X$, 
	then $\pth_X(\mu_0) \in V^*$ (because it is an intersection point of a horizontal 
    segment of $\pth_X$ with a vertical grid line) and $\xi_1=\mu_0$.
    Otherwise, $X$ moves along a vertical segment~$g$ immediately after time~$\mu_0$
    and the time $\xi_1$ is the first time at which $X$
    reaches a point of $V^*$ lying on~$g$. Observe that 
    $\xi_1 \leq \mu_1$ because the robots remain $x$-separated as $X$ moves
    along $g$, while $\plan[\mu_0,\mu_1]$ must contain a configuration
    that is not $x$-separated
    by the definition of an unsafe interval.
    So, we conclude that $\xi_1 = \mu_0$ or $\plan[\mu_0,\xi_1]$ consists of one 
    robot~$X$ moving along a vertical segment while the other is parked at a 
    position of $V^*$. Similarly, let
    \[
    \xi_2:= \max\{\xi\in [\mu_2,\mu_3]: \{\pth_A(\xi),\pth_B(\xi)\}\subset V^* \}
    \]
     and note that $\xi_2 = \mu_3$ or $\plan[\xi_2,\mu_3]$ 
    consists of one robot~$X$ moving along a vertical segment
    while the other robot is parked at a point of $V^*$.
    \item 
    Finally, we define $\nu_1, \nu_2$.
    Let $\edge$ denote the left edge of $\Rud$ and define
    $\nu_1 := \min \{ \nu\in [\xi_1,\mu_1] : \pth_B(\nu)\in \edge\}$ and
    $\nu_2 := \max \{ \nu\in [\mu_2,\xi_2] : \pth_B(\nu)\in \edge\}$;
    this is well defined since $\pth_B[\xi_1,\mu_1]$ and
    $\pth_B[\mu_2,\xi_2]$ must cross~$\edge$.
    % $\pth_B(\mu_3)$ lies in the top-left corner square, which is to the left of~$e$.      
    \end{itemize}
    We now explain how to construct the new plan $\plan^*$.
    Let $q:= (\pth_B(\mu')_x, \pth_A(\xi_1)_y)$ and note that $q\in R$.
    Moreover, $q$ has distance at least~$1$ to $\psi$
    because $\plan(\mu')$ is $x$-separated,
    with $\pth_B(\mu')$ lying to the right of~$\pth_A(\mu')$.
    We construct $\plan^*$ from $\plan$ by replacing $\plan[\xi_1,\xi_2]$ 
    with $\plan^*[\xi_1,\xi_2]$ consisting of the following three moves. 
    See also Figure~\ref{fig:switch-U-turn}.
    \begin{enumerate}[M1.]
    \item Move $A$ horizontally from $\pth_A(\xi_1)$ to~$q$. 
    \item We first let $B$ follow the subpath $\pth_B[\xi_1,\nu_1]$,
          then we move $B$ to $\pth_B(\nu_2)$ along~$\edge$, and finally we
          let $B$ follow $\pth_B[\nu_2,\xi_2]$.
    \item Move $A$ horizontally from $q$ to $(\pth_A(\xi_2)_x,q_y)$
          and then vertically to~$\pth_A(\xi_2)$.
    \end{enumerate}
    \begin{figure}
    	\centering
    	\includegraphics{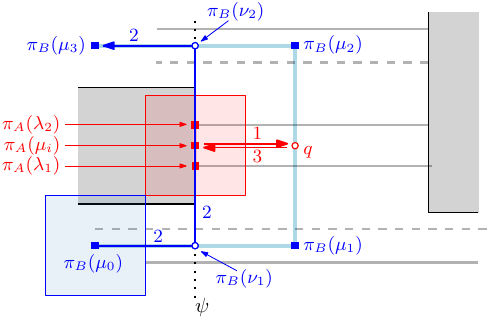}
	    \caption{The modification $\plan^*$ in the proof of Lemma~\ref{lem:central-move-no-swap} when $(\pth_B(\mu_0),\pth_B(\mu_3))$ is a blocked pair. In this example, we have $\xi_1 =\mu_0$, $\xi_2 = \mu_3$, and $\pth_A(\mu_0) = \pth_A(\mu_3)$ (denoted by $\pth_A(\mu_i)$). The new plan $\plan^*$ is dark colored, with the moves labeled by their number, while the trace of the original plan $\plan$ is light colored.}
    	\label{fig:switch-U-turn}
    \end{figure}
    
    Note that $\plan^*(\mu)$ is $x$-separated for all $\mu\in [\mu_0,\mu_3]$;
    this is true for $\plan^*[\mu_0,\xi_1]$ because $\plan(\mu_0)$ is $x$-separated, 
    for $\plan^*[\xi_1, \xi_2]$ 
    because $q$ has distance at least~$1$ to~$\psi$,
    for $\plan^*[\xi_2, \mu_3]$ because $\plan(\mu_3)$ is $x$-separated.
    Hence, there are no swap intervals left in~$\plan^*[\lambda_1,\lambda_2]$. 
    It remains to show that the modification is compliant. 
    \smallskip
    \begin{itemize}
    \item The three moves of $\plan^*[\xi_1,\xi_2]$ obviously lie in $\freesp$,
    and there are no conflicts since $\plan^*[\xi_1,\xi_2]$ is $x$-separated throughout.
    Thus, $\plan^*$ is feasible. 
    \item
    Note that $\xi_1\leq \mu'\leq \xi_2$. Hence, the modification of $\pth_A$ increases 
    its length of by at most $2\cdot|\pth_A(\mu')_x-q_x|=2\cdot|\pth_A(\mu')_x-\pth_B(\mu')_x|$. 
    On the other hand, the modification of $\pth_B$ decreases its length by at 
    least $2(\pth_B(\mu')_x - x(\edge))$, since $B$ no longer has to go from $\edge$ to $\pth_B(\mu')$ 
    and back. Because $x(\edge) \leq \pth_A(\mu')_x < \pth_B(\mu')_x$, 
    it follows that $\plan^*$ is still optimal.
    \item
    Each vertex of $\pi^*_A[\xi_1,\xi_2]$ is either a point of~$V^*$, or a vertex of~$\pth_A$,
    or the point~$q$. The latter lies on a vertical line through a vertex of $\pth_B$, 
    so we have the
    {\sc Vertical Alignment} property for $\pi^*_A$. Furthermore, the only new segment
    of $\pi^*_B$ lies on $\edge$, which implies  the {\sc Vertical Alignment} property for $\pi^*_B$.
    \item
    The parking spots of $A$ and $B$ are points of $V^*$ by construction. %while those of $B$ 
    %lie on vertices by \lemref{unsafe-vertical-edge}. 
    We can add zero-length 
    segments to make $\plan^*$ alternating, while still satisfying alignment.
    \end{itemize}
    \smallskip
    Finally, it is clear that all bad segments of $\plan^*$ are either 
    present in $\plan$ or are a segment of $\pth_A[\lambda_1,\lambda_2]$.
    We conclude that $\plan^*$ is a compliant modification of $\plan$.
\end{proof}
%----------------------------------------------------------------------------
\begin{lemma}\label{lem:tiny-one-swap}
    If $\I(\rect)$ has a tiny component, there is a compliant modification~$\plan^*$ 
    of $\plan$ with at most one swap interval.
\end{lemma}
%----------------------------------------------------------------------------
\begin{proof}
    If $\plan$ already has only one swap interval then we are done. 
    Otherwise, we can assume wlog that $\plan$ has two swap intervals $[\mu_0,\mu_1]$ 
    and $[\mu_2,\mu_3]$ with the properties of \lemref{limit-swap-intervals} and
    such that $B$ lies below $A$ during $[\mu_0,\mu_1]$. 
    
    Let $\tau$ be a tiny component of~$\I(\rect)$, which we can assume, without loss
    of generality, lies in the bottom-right corner square of~$\I(\rect)$.
    The subpath $\pth_B[\mu_0,\mu_1]$ lies inside a single connected component
    $Z^*$ of $Z^-\cap \freesp$, which is a rectangle.
    We claim that $Z^* \cap \Rud^+ = \emptyset$. Assume for a contradiction
    this is not the case. Then the right edge $\eta$
	of $\Rud\cap \Rxy^-$ lies in~$\freesp$. (Indeed, $Z^*$ intersects $\eta$, and $\fre\cap\Rud$ is a 
	$x$-monotone polygon, and the upper endpoint of $\eta$ lies on $\bd\fre\cap\Rud$.)
    But then, by Observation~\ref{obs:free-segments}, we can connect $\tau$ to $\eta$ by a horizontal segment
	that lies in $\fre$, contradicting that $\tau$ is a tiny component.
    
    Let  $i\in \{0,1\}$ be such that $\pth_B(\mu_i)_x > \pth_A(\mu_i)_x$;
    such an $i$ exists because $[\mu_0,\mu_1]$ is a swap interval. Thus,
    $\pth_B(\mu_i)_x\notin \Rud^-$. Because $Z^* \cap \Rud^+ = \emptyset$
    and $\pth_B(\mu_i)\in Z^*$,
    we know that $\pth_B(\mu_i)\in \Rud$.
    Hence, $\pth_B(\mu_i)$ cannot be part of a blocked pair.
    Now let $j\in \{2,3\}$ be such that $\plan(\mu_i)$ and $\plan(\mu_j)$ have the same $x$-order. 
    If $j-i > 1$, then applying \lemref{plan-between-swap-endpoints} to $\mu_i,\mu_j$ gives 
    a compliant re-parametrization with at most one swap interval.
    (We can apply \lemref{plan-between-swap-endpoints} because $(\pth_B(\mu_i),\pth_B(\mu_j))$ 
    is not a blocked pair and $\pth_B$ lies in the giant component during a swap interval.) 
    Otherwise, we have $j-i = 1$, and therefore $\plan(\mu_1)$ and $\plan(\mu_2)$ have the same $x$-order 
    and $(\pth_B(\mu_1),\pth_B(\mu_2))$ is not a blocked pair. Hence, \lemref{central-move-no-swap} gives a complaint modification without swap intervals.
\end{proof}
%----------------------------------------------------------------------------------
We can now prove the main structural property of an optimal plan
in which $\plan[\lambda_1, \lambda_2]$ has two swap intervals.
    %----------------------------------------------------------------------------------
    \begin{lemma}
	    \label{lem:mu1inR}
	    If $\plan$ has two swap intervals $[\mu_0,\mu_1]$ and $[\mu_2,\mu_3]$,
    and $\pth_B[\mu_0,\mu_3]$ does not leave $\I(\rect)$, then
    $\pth_B(\mu_1)\in \Rud$.
    \end{lemma}
    %----------------------------------------------------------------------------------
    \begin{proof}
    Assume wlog that $B$ lies below $A$ during $[\mu_0,\mu_1]$ and above $A$ during~$[\mu_2,\mu_3]$.
    Suppose for a contradiction that $\pth_B(\mu_1)\notin \Rud$ and assume wlog 
    that $\pth_B(\mu_1)\in \Rud^+$. Let $Z^*$ be the rectangle of $Z^- \cap \freesp$ that 
    contains $\pth_B[\mu_0,\mu_1]$. 
    Since $\pth_B[\mu_0,\mu_3]$ does not leave $\I(\rect)$, 
	the subpath~$\pth_B[\mu_1,\mu_2]$ visits~$\rect$.
	Thus we can define 
    \[
	    \mu':= \min\{\mu\in [\mu_1,\mu_2]: \pth_B(\mu)\in \Rud \mbox{ and } \pth_B(\mu)_y\geq y(\mytop(Z^-))\}.
    \]
    We now give a feasible plan~$\plan'$ 
    from $\plan(\mu_0)$ to $\plan(\mu')$ and show that 
    $\cost{\plan'[\mu_0,\mu']} < \cost{\plan[\mu_0,\mu']}$,
    thus contradicting the optimality of~$\plan$. 
    The new subplan~$\plan'[\mu_0,\mu']$ is as follows. See also Figure~\ref{fig:shortcut-below-Z}.
    \begin{quotation} \noindent \vspace*{-3mm}
    \begin{enumerate}[M1.] 
        \item Move $B$ horizontally from $\pth_B(\mu_0)$ to $(\pth_B(\mu')_x, \pth_B(\mu_0)_y)$.
        \item Move $A$ horizontally from $\pth_A(\mu_0)$ to $(\pth_A(\mu')_x, \pth_A(\mu_0)_y)$, then move $A$ vertically to $\pth_A(\mu')$.
        \item Move $B$ vertically to $\pth_B(\mu')$.
    \end{enumerate}
    \end{quotation}
    \begin{figure}
    	\centering
    	\includegraphics{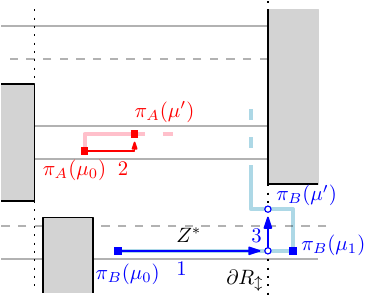}
    	\caption{The modification $\plan'$ of $[\mu_0,\mu']$ in Lemma~\ref{lem:mu1inR}. In this example, $\pth_B[\mu_0,\mu']$ enters $Z\cap \Rud$ via the right edge of $\Rud$.}
    	\label{fig:shortcut-below-Z}
    \end{figure}
	Since $Z^*$ intersects $\Rud^+$ and $\mu_0\in Z^*$, we conclude that the horizontal 
	segment from $\pi_B(\mu_0)$ to $(\pth_B(\mu')_x, \pth_B(\mu_0)_y)$ lies in $Z^*$.
	Hence, move~M1 lies in $Z^*\subset\freesp$. 
    Moreover, move~M3 lies in $\fre$ because $\Rud\cap\freesp$ is $x$-monotone.
    Finally, the moves of~$A$ do not leave $\R$ so they lie in $\fre$ as well.
    The horizontal movements in M1 and M2 are without conflict
    because $\plan(\mu_0)$ is $y$-separated. Moreover, the vertical movements
    in M2 and M3 are without conflict because $\plan(\mu')$ is $x$-separated.
    (The fact that $\plan(\mu')$ is $x$-separated follows from $\pth_B(\mu')\in Z$
    together with $\lambda_1 < \mu'<\lambda_2$ (so that $\pth_A(\mu')\not\in \topR$).
    Hence, $\plan'$ is a valid $\plan(\mu_0)\plan(\mu')$-plan. 
    
    Note that $\pi'_A$ and $\pi'_B$ are $xy$-monotone, while $\pth_B[\mu_0,\mu']$ is not; 
    indeed, $\pth_B(\mu_0),\pth_B(\mu')\not\in \Rud^+$ while $\pth_B[\mu_0,\mu']$ 
    is in $\Rud^+$ at time $\mu_1$, and so $\pth_B[\mu_0,\mu']$ properly 
    crosses the right edge of $\Rud$ at least twice. 
    So, we have a shortcut plan, which contradicts the optimality of~$\plan$. 
    \end{proof}
%----------------------------------------------------------------------------------
\begin{lemma}\label{lem:two-swap-leaves-influence}
    If $\plan$ has two swap intervals $[\mu_0,\mu_1]$ and $[\mu_2,\mu_3]$,
    and $\pth_B[\mu_0,\mu_3]$ does not leave $\I(\rect)$, then there is a 
    compliant modification $\plan^*$ with at most one swap interval.
\end{lemma}
%----------------------------------------------------------------------------------
\begin{proof}
    If there are two swap intervals, then we can assume by \lemref{limit-swap-intervals}
    that $A$ is above $B$ in $[\mu_0,\mu_1]$ and $B$ is above $A$ in $[\mu_2,\mu_3]$.
	By \lemref{mu1inR}, $\pth_B(\mu_1)\in \Rud$, so 
	the point $\pth_B(\mu_1)$ is not part of a blocked pair. 
    If the $x$-orders of $\plan(\mu_1)$ and $\plan(\mu_3)$ are identical, 
    then \lemref{plan-between-swap-endpoints} implies there is compliant 
    re-parametrization such that $\plan[\lambda_1,\lambda_2]$ has 
    at most one swap interval.
    Otherwise, the $x$-orders of $\plan(\mu_1)$ and $\plan(\mu_2)$ are identical. 
    Then, by \lemref{central-move-no-swap}, there is a compliant modification 
    without swap intervals.
\end{proof}
%----------------------------------------------------------------------------
The following lemma summarizes the main properties of swap intervals in an 
optimal plan, proved in
Lemmas \ref{lem:unsafe-is-swap}--\ref{lem:two-swap-leaves-influence}.

%----------------------------------------------------------------------------
\begin{restatable}{lemma}{mainprop}
    \label{lem:main-properties}
    Let $\plan$ be an optimal \stplan, let $e$ be a bad horizontal segment on~$\pth_A$,
    and let $[\lambda_1,\lambda_2]$ be as defined in (\ref{def:lambdas}).
    Then there is a compliant modification $\plan^*$
    of $\plan$ such that any unsafe $x$-separated interval is a swap interval
    with the following properties.
    \begin{enumerate}[(i)]
    \item The plan $\plan^*[\lambda_1,\lambda_2]$ has at most one swap interval with $B$ above $A$
          and at most one swap interval with $B$ below~$A$.
    \item If $\plan^*[\lambda_1,\lambda_2]$ has two swap intervals $[\mu_0,\mu_1]$ 
    and $[\mu_2,\mu_3]$, then $\I(\rect)$ does not have tiny components, $\pi^*_B[\lambda_1,\lambda_2]$ contains a blocked pair, and
    $\plan^*_B[\mu_0,\mu_3]$ leaves~$\I(\rect)$.
    \end{enumerate}
\end{restatable}    
%----------------------------------------------------------------------------
The following property of an optimal plan with two swap intervals will be useful 
to prove the optimality of the plan resulting from the surgery described later. 
%-------------------------------------------------------------------
\begin{lemma}\label{lem:lambda2}
	Suppose $\plan[\lambda_1,\lambda_2]$ has a swap interval $[\mu_0,\mu_1]$ with $B$ below $A$, 
    and a swap interval $[\mu_2,\mu_3]$ with $B$ above $A$,
    where $\mu_1<\mu_2$.
    If $\lambda_1$ is an endpoint of an unsafe interval 
    then $\pth_B(\lambda_1)_y < \pth_A(\lambda_1)_y$, 
    and if $\lambda_2$ is an endpoint of
    an unsafe interval then $\pth_B(\lambda_2)_y > \pth_A(\lambda_2)_y$.
\end{lemma}
%-------------------------------------------------------------------
    \begin{figure}
    	\centering
    \includegraphics{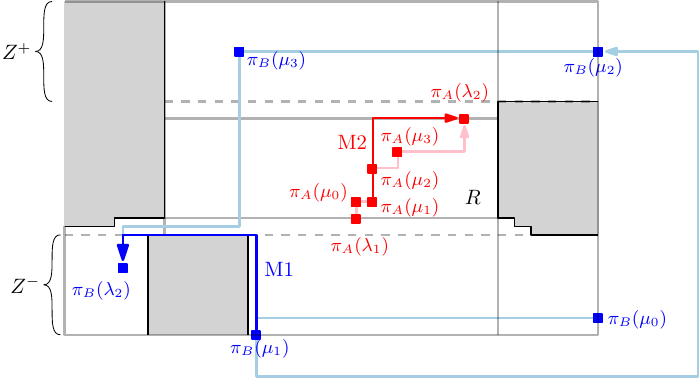}
 	    \caption{Illustration of the proof of \protect\lemref{lambda2}, using the shortcuts defined by M1 and M2. 
        % \ben{Hopefully everything fixed now with this figure, and keeping Alex's suffix.}
        }
    
    	\label{fig:lambda2}
    \end{figure}
%-------------------------------------------------------------------
\begin{proof}
    We prove the lemma for $\lambda_2$; the proof for $\lambda_1$ is analogous.
    Assume for a contradiction that $\lambda_2$ is an endpoint
    of an unsafe interval and that $\pth_B(\lambda_2)_y \leq \pth_A(\lambda_2)_y$. 
    Note that $\plan(\lambda_2)$ is $y$-separated and that $\pth_B(\lambda_2)\in Z^-$. 
    We will show there is a shortcut plan from a time in $[\mu_0,\mu_1]$ to $\lambda_2$ in which $\pth_B$ avoids visiting $\Rxy^+$.
    
    If $\pth_B(\lambda_2)$ lies in the same rectangle of $Z^-\cap \freesp$ as $\pth_B[\mu_0,\mu_1]$, then there are $xy$-monotone paths in $\freesp$ from $\pth_B(\mu_0)$ to $\pth_B(\lambda_2)$ and from $\pth_A(\mu_0)$ to $\pth_A(\lambda_2)$. 
    Note that $\plan(\mu_0)$ and $\plan(\lambda_2)$ are both $y$-separated configurations with the same $y$-order. 
    Hence, \lemref{x-sep} implies there is a valid plan from the configuration
    $\plan(\mu_0)$ to the configuration~$\plan(\lambda_2)$ in which $A$ and $B$ use $xy$-monotone paths.
	%\footnote{\lemref{x-sep} is stated with respect to the $x$-coordinates, but, unlike many of the other lemmas, it also applies if we exchange the coordinate axes.}
    This plan is strictly shorter than $\plan[\mu_0,\lambda_2]$, as $\pth_B[\mu_0,\lambda_2]$ visits $Z^+$.
    
    Otherwise, $\pth_B(\lambda_2)$ lies in a different rectangle 
    of $Z^-\cap \freesp$ than $\pth_B[\mu_0,\mu_1]$. Assume wlog that 
    the rectangle of $\pth_B(\lambda_2)$ lies to the left of $\pth_B[\mu_0,\mu_1]$, 
    so we have $\pth_B(\lambda_2)_x < \pth_B(\mu_0)_x$ and $\pth_B(\lambda_2)_x < \pth_B(\mu_1)_x$. 
    Let $i\in \{0,1\}$ be such that $\pth_B(\mu_i)_x < \pth_A(\mu_i)_x$.
    By \lemref{xy-monotone-iff}, there is an $x$-monotone path~$\varphi_B$
    from $\pth_B(\mu_i)$ to $\pth_B(\lambda_2)$ that is shortest in~$\freesp$. 
    Now consider the following plan $\plan'$ from $\plan(\mu_i)$ to $\plan(\lambda_2)$,
    illustrated in \figref{lambda2}.
    \smallskip
    \begin{enumerate}[M1.]
        \item Move $B$ along $\varphi_B$ from $\pth_B(\mu_i)$ to $\pth_B(\lambda_2)$.
        \item Move $A$ from $\pth_A(\mu_i)$ to $\pth_A(\lambda_2)$ along an
              L-shaped path whose first link is vertical.
    \end{enumerate}
    \smallskip
    The plan~$\plan'$ clearly lies in $\freesp$. Move~M1 is without conflict, 
    as the configuration remains $x$-separated because $\varphi_B$ is $x$-monotone. 
    Move~M2 is without conflict, as the configuration is $x$-separated during 
    the vertical move, and $y$-separated during the horizontal move.
    So, $\plan'$ is a valid plan of which the paths are shortest paths in~$\freesp$.
    The plan $\plan$ does not consist of two shortest paths,
    however, since $\pth_B[\mu_i,\lambda_2]$ visits $Z^+$.
    
    Thus, we obtain a shortcut in both cases, which contradicts the optimality of~$\plan$.
    Hence, if $\lambda_2$ lies in an unsafe interval, then $\pth_B(\lambda_2)_y > \pth_A(\lambda_2)_y$. 
\end{proof}

Next, we prove two lemmas that narrow down the possible locations of $A$ before time~$\lambda_1$.
These lemmas will be useful in proving the feasibility
of the modified plan if we perform surgery on $\pth_B$ outside
$[\lambda_1, \lambda_2]$. 
For technical reasons that will become clear later, we not only need to work 
with $\horF$ but also with $\horzF$, the rectangular subdivision of $\F$ induced by the 
horizontal lines of $L_0$. (Recall that these are the horizontal lines containing the horizontal edges of $\F$
and the horizontal lines containing the points $s_A,s_B,t_A,t_B$.)
%----------------------------------------------------------------------------
\begin{lemma}\label{lem:before-lambda0-not-other-side}
	Suppose that $\pth_A(\lambda_1)\in \botR$ and $\pth_A(\lambda_2)\in \topR$, 
    that there is a swap interval $[\mu_0,\mu_1]$ with $B$ below $A$, and that
	$y(\botR)\geq y(\topQ)$ for the rectangle $Q$ of $\horzF$ that contains $\pi_B(\mu_1)$.
	Let $\nu_0<\lambda_1$ be the first time such that
    $\pth_B[\nu_0,\mu_1]\subset Q$. Then 
	for any $\nu\in[\nu_0,\lambda_1]$, we have $\pth_A(\nu)\notin \interior(\rect) \cup (\Rxy^+\cap \gamma(\rect))$.
\end{lemma}
%----------------------------------------------------------------------------
\begin{proof}
	% \mdb{Whole proof is new.}
	% Analogous to Lemma~\ref{lem:case-III-shortcut-above}.
	For the sake of contradiction, suppose there is a time $\nu \in [\nu_0,\lambda_1]$
	such that $\pth_A(\nu)\in \interior(\rect) \cup (\Rxy^+ \cap \gamma(R))$. 
	Let $\mu^- \in [\mu_0,\mu_1]$ be a time such that 
    $\xco{\pth_A(\mu^-)} = \xco{\pth_B(\mu^-)}$ and $\pth_B(\mu^-) \in \Rud \cap Z^-$, 
    which must exist since $[\mu_0,\mu_1]$ is a swap interval with $B$ below $A$ 
    and $\pth_A[\mu_1,\mu_2]\subset R$. Note that $\lambda_1\leq\mu_0<\mu^-<\mu_1$;
	indeed, since $\plan(\mu_0)$ and $\plan(\mu_1)$ are $x$-separated, we cannot have $\mu^-=\mu_0$ or $\mu^-=\mu_1$.
	
	We now show how to create a shortcut $\plan'$ of $\plan[\nu,\mu^-]$, thus obtaining the desired contradiction.
	Note that $\pth_A(\nu)_y> y(\botR)$ because $\pth_A(\nu)\in \interior(\rect) \cup \Rxy^+$.
    Moreover, $\pth_A(\mu^-)_y > y(\botR)$ 
    because $\lambda_1<\mu^-<\lambda_2$. Since $\nu<\lambda_1<\mu^-$ and $\pth_A(\lambda_1)\in \botR$,
    the path $\pth_A[\nu,\mu^-]$ is not $xy$-monotone. 
	Thus, it suffices to show that there exists a feasible 
	$(\plan(\nu),\plan(\mu^-))$-plan $\plan'=(\pth'_A,\pth'_B)$ such that $\pth'_A$ and $\pth'_B$ are 
	both $xy$-monotone; see \figref{shortcut-above}.
	Let $p_A=(\xco{\pth_A(\nu)}, \yco{\pth(\mu^-)})$ and 
	$p_B=(\xco{\pth_B(\nu)},\yco{\pth_B(\mu^-)})$.
	We have two cases.
    \begin{itemize}
    \item 
    If $\yco{\pth_A(\nu)} \geq \yco{\pth_A(\mu^-)}$ then $\plan'$ consists of the following moves:
	%\mst{These sets of moves are mostly the same, I don't think we need to repeat the common part here. One way of phrasing that is ``If $\pth_A(\nu)_y < \pth_A(\mu^-)_y$, first move $A$ along a vertical segment to $p_A$. '' and after that execute the first set of moves M1-M2 (where ``$\pi_A(\nu)$'' in M2 may also be $p_A$. I suggest either giving the point $A$ reaches after the initial ``conditional'' move a name, or not mentioning the start point of the moves at all)}
    %\mdb{I think keeping it as it is (with the repetition) is fine.}
    %----------------------------------------------------------------------------
    %\vspace*{-7mm}
	%\begin{quotation} \noindent 
		%\begin{enumerate}[M1.]
			%\item Move $B$ along a vertical segment from $\pth_B(\nu)$ 
				%to $ p=(\pth_B(\nu)_x,\pth_B(\mu^-)_y)$
			%\item Move $B$ along a horizontal segment from $p$ to $\pth_B(\mu^-)$.
			%\item Move $A$ along an $xy$-monotone path from $\pth_A(\nu)$ 
				%to $\pth_A(\mu^-)$.
		%\end{enumerate}
	%\end{quotation}
		\begin{enumerate}[M1.]
			\item Move $B$ along an L-shaped path from $\pth_B(\nu)$ via $p_B$ to $\pth_B(\mu^-)$.
			\item Move $A$ along an $xy$-monotone path from $\pth_A(\nu)$ 
				to $\pth_A(\mu^-)$.
		\end{enumerate}
	\item 
    Otherwise---that is, if $\yco{\pth_A(\nu)} < \yco{\pth_A(\mu^-)}$---then~$\plan'$ consists of the 
	following moves:
		\begin{enumerate}[M1.]
			\item Move $A$ along a vertical segment from $\pth_A(\nu)$ to $p_A$.
			\item Move $B$ along an L-shaped path from $\pth_B(\nu)$ via $p_B$ to $\pth_B(\mu^-)$.
			\item Move $A$ along a horizontal segment from $p_A$ to to $\pth_A(\mu^-)$.
		\end{enumerate}
    \end{itemize}
    %----------------------------------------------------------------------------
	\begin{figure}
		\centering
        \includegraphics{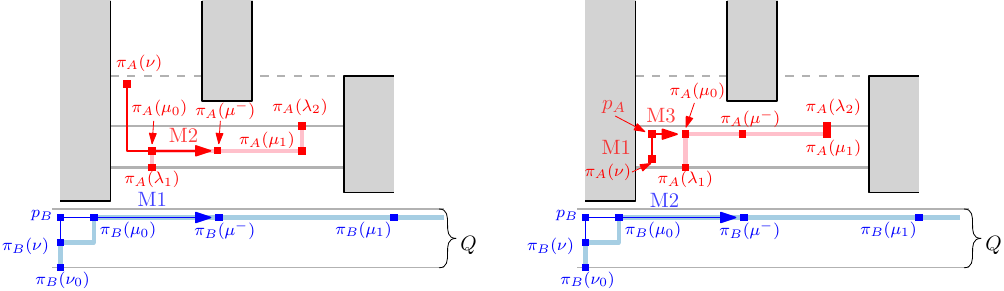}
		\caption{The shortcut plan from $\plan(\nu)$ to $\plan(\mu^-)$ in \lemref{before-lambda0-not-other-side}
        when $\yco{\pth_A(\nu)} \geq \yco{\pth_A(\mu^-)}$ (left) 
        and when $\yco{\pth_A(\nu)} < \yco{\pth_A(\mu^-)}$
        (right).}
        % \mdb{Try to avoid the arrows when the labels can be placed clearly without them. For example,
        % the arrow for $\pth_A(\nu)$ is not needed, and with some tweaking some other arrows can probably be avoided as well.
        % Same for other figures.}\ben{Removed some arrows throughout. I don't mind the remaining arrows, hopefully they're not too distracting?}
        \label{fig:shortcut-above}
	\end{figure}
    %----------------------------------------------------------------------------
    We first argue that the new paths $\pth'_B$ and $\pth'_A$ lie in~$\fre$.
    Note that $\pth'_B$ is the same in both cases, and that it lies in $Q$,
    which is contained in~$\fre$. We now consider $\pth'_A$.
    By \lemref{xy-shortest path},
    there is an $xy$-monotone $(\pth_A(\nu),\pth_A(\mu^-))$-path in $\fre$,
    so we can ensure that $\pth'_A\subset \fre$ in the first case.
    Furthermore, 
	if $\yco{\pth_A(\nu)} < \yco{\pth_A(\mu^-)}$ then $\pth_A(\nu) \in R$ and thus $\pth'_A\subset \fre$
    in the second case as well. 
    Hence, both $\pth_A', \pth'_B$ lie in $\fre$.
	By construction, both $\pth_A'$ and $\pth_B'$ are $xy$-monotone. 
    
    It remains to show that the two robot do not collide in $\plan'$.
	Recall that 
	$\plan(\mu^-)$ is $y$-separated.
    If a robot moves along a horizontal segment in $\plan'$ at 
	some time $\lambda$, then we always have $\yco{\pth_A(\lambda)} \geq \yco{\pth_A(\mu^-)}$ and 
	$\yco{\pth_B(\lambda)} \leq \yco{\pth_B(\mu^-)}$, which implies that
    $\plan'(\lambda)$ is 
	$y$-separated. Hence,  horizontal motions of $\plan'$ are conflict-free. If a robot 
	moves along a vertical segment at some time $\lambda$ in $\plan'$, then either the robot is 
	moving away from the other robot %\mdb{Removed ``$B$ moves in $-y$-direction or $A$ moves in 
	%$+y$-direction'' because this is not always away from each other} 
    or they remain $y$-separated along the segment because the robots are $y$-separated
    at the end of the vertical move. Thus, the 
	vertical moves are also conflict free. 
    Hence, the two robot do not collide in $\plan'$. This completes the proof of the lemma.
	%Since $\pth'_A$ and $\pth'_B$ are both $xy$-monotone while $\pth_A[\nu,\mu^-]$ is not $xy$-monotone, the plan $\plan'$ is indeed a shortcut, giving us the desired contradiction.
\end{proof}
%----------------------------------------------------------------------------
\begin{lemma}\label{lem:before-lambda1-x-sep-below}
	Suppose that $\pth_A(\lambda_1)\in \botR$ and $\pth_A(\lambda_2)\in \topR$, 
	that there is a swap interval $[\mu_0,\mu_1]$ with $B$ below $A$, and that
	$y(\botR)\geq y(\topQ)$ for the rectangle $Q$ of $\horzF$ that contains $\pi_B(\mu_1)$.
	%Suppose that $\pi_A(\lambda_1)\in \botR$ and $\pth_A(\lambda_2)\in \topR$ and $\plan$ has a swap interval $[\mu_0,\mu_1]$ with $A$ above~$B$. 
    %Let $Q$ be the rectangle of $\horzF$ that contains $\pth_B(\mu_1)$, and let $y(\botR)\geq y(\topQ)$
    Let $Q^+$ be the closed rectangle with bottom edge $\topQ$ and height~$1$.
    Suppose there is a time $\nu < \lambda_1$ at which  $\pth_B[\nu,\mu_1]\subset Q$ and $\pth_A(\nu)\in Q\cup Q^+\cup \gamma(R)$ and $\plan(\nu)$ is $x$-separated. %}%
    %{ Let $\nu < \lambda_1$ be the first time at which  $\pth_B[\nu,\mu_1]\subset Q$, and $\pth_A(\nu)\in Q\cup Q^+\cup \gamma(R)$, and $\plan(\nu)$ is $x$-separated.}
    Then there is a compliant 
    modification $\plan^*[0,\mu_1]$ of $\plan[0,\mu_1]$ such that 
    $\plan^*(\nu')$ is $x$-separated for all $\nu'\in [\nu,\mu_0]$.
\end{lemma}
%\mst{I had originally planned for $Q^+$ to be open at the top boundary, but then we get in trouble looking for the first time $A$ is in $Q^+$ when entering via the top edge. So, $Q^+$ is closed.}
	
%\mst{Also, I changed the first sentence of the lemma to a copy of the first sentence of Lemma~\ref{lem:before-lambda0-not-other-side}. This lemma holds without the additional assumptions, but the assumptions are met whenever we use this lemma (i.e.~we are in Case~III) and making these lemmas consistent makes it easier to follow their usage. With the additional assumptions we can also get rid of the lines handling their cases.}
%----------------------------------------------------------------------------
	% \begin{figure}
	% 	\centering
	% 	\includegraphics{figures/lemma-3.18-times.pdf}
	% 	\caption{The relevant time instances defined for \lemref{before-lambda1-x-sep-below}. \ben{Double check and draw argument when stable}}
	% 	\label{fig:before-lambda1-x-sep}
	% \end{figure}
    \begin{figure}
		\centering
        \makebox[\textwidth][c]{\includegraphics[width=1.2\textwidth]{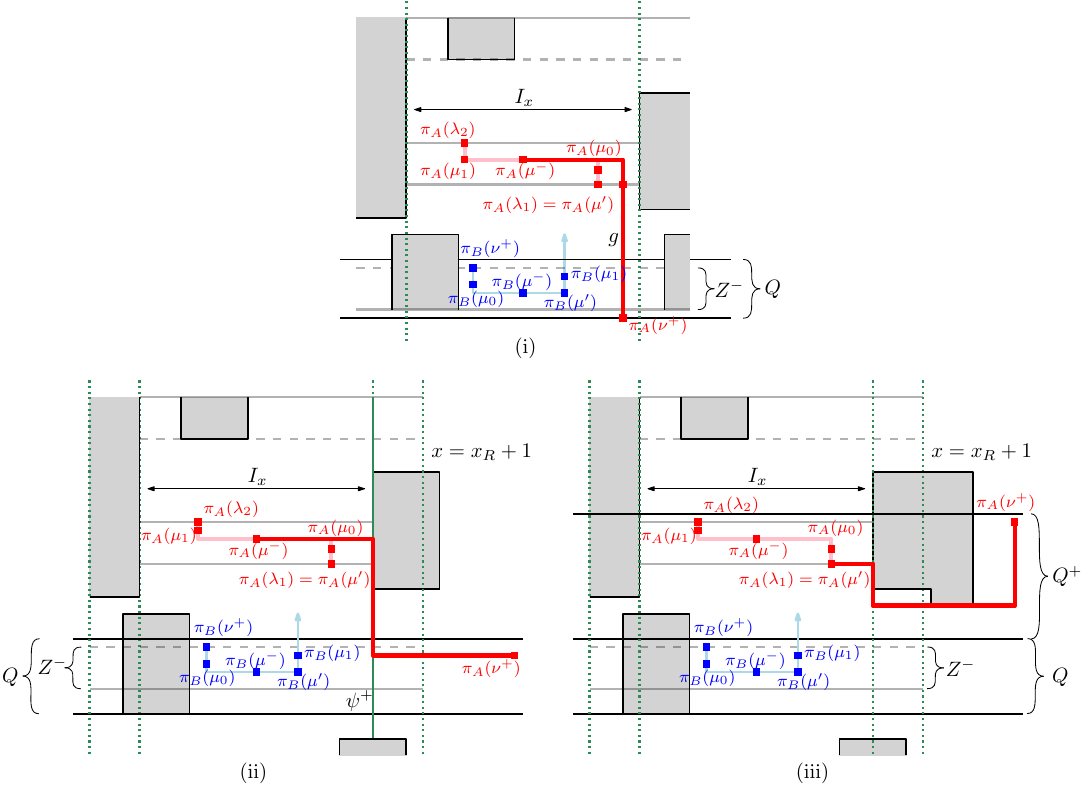}}
		\caption{Illustrations of the proof of \lemref{before-lambda1-x-sep-below}. For readability, the original path $\pth_A[\nu^+,\lambda_1]$ is not shown and the horizontal axes are compressed. The dotted lines support the vertical edges of $R$ and $R_\Box$. (i) $\pth_A(\nu^+)_x \in I_x$; (ii) $\pi_A(\nu^+)_x \notin I_x$; (iii) $\pth_A(\nu^+) \in Q^+$.}
		\label{fig:before-lambda1-x-sep}
	\end{figure}
%----------------------------------------------------------------------------
\begin{proof}
    %\deleted{We observe that the conditions of the lemma are satisfied for $\plan(\mu_0)$, so $\nu$ is well defined.}
      We will construct a compliant re-parametrization $\plan^*[\nu^+,\mu_i]$ of $\plan[\nu^+,\mu_i]$,
       for a suitable $\nu^+\geq \nu$ and for some $i\in\{0,1\}$ so that
       $\plan^*(\nu')$ is $x$-separated for all $\nu' \in [\nu, \mu_i]$.
       By taking $\plan^*[0,\nu^+]=\plan[0,\nu^+]$
       and $\plan^*[\mu_i,\mu_1]= \plan[\mu_i,\mu_1]$ we will obtain the desired 
       re-parametrization of $\plan[0,\mu_1]$. We start by defining~$\nu^+$, for
       which we have two cases.
       \begin{itemize}
       \item If $\pth_B$ enters $Q$ or $\pth_A$ enters $Q\cup Q^+\cup \gamma(R)$ at time $\nu$,
             then set $\nu^+:= \nu$.
       \item Otherwise $\plan$ is $x$-separated at time $\nu$.
       % \ajs{The statement of the lemma does not define $\nu$ to be the earliest time with its properties, so it is possible that the entire move containing $\plan(\nu)$ is $x$-separated.}
       If $\plan(\nu)$ is a breakpoint of  $\plan$ (i.e., $\pth_A(\nu),\pth_B(\nu)$ are vertices of $\pth_A,\pth_B$, respectively), set $\nu^+ \assign \nu$. Otherwise,
             let $X$ be the robot that is moving at time $\nu$, and set 
             \[
            \nu^+:=\min\{\nu'\in [\nu,\mu_0]: \text{$\pth_X(\nu')$ lies on a vertical grid line or is a vertex of $\pth_X$}\}.
            \] 
       \end{itemize}
     Note that the three assumptions on $\nu$ stated in the lemma also hold for~$\nu^+$.
     This is trivial when $\nu^+=\nu$. In the other case, we clearly 
     have $\pth_B[\nu^+,\mu_1]\subset Q$ since $\nu^+\geq \nu$.  
     We also have $\pth_A(\nu^+)\in Q\cup Q^+\cup \gamma(R)$. Indeed, if $X=B$ 
     (that is, $B$ is the moving robot) this trivially holds, 
     and if $X=A$ then $A$ is moving horizontally and does not cross a grid line during $[\nu,\nu^+)$, 
     so it cannot move out of~$Q\cup Q^+\cup \gamma(R)$.
        %\ajs{Remind the reader why? E.g.: (recall that there are vertical $0$-lines supporting the vertical edges of $R$ so there are vertical $1$-lines supporting the vertical edges of $R_\Box \supseteq \gamma(R)$.)}
        % \mdb{Not sure if this argument is complete, as it does not talk about $Q$ and $Q^+$. Moreover, I am a bit hesitant to add extra explanations, as the statement is relatively easy to verify (I think) and making the proof longer is not necessarily helping the reader. What do others think ?}
        % \ajs{Good points. What took me a while to remember was why $A$ couldn't leave $\gamma(R)$. It cannot cross the vertical edges of $Q$ since the vertical edges of $Q$ are free space edges by definition, so it cannot cross the vertical edges of $Q^+$ either, but the vertical edges of $\gamma(R)$ are not necessarily in any $0$-line. However, if a vertical edge of $\gamma(R)$ can be crossed, it must be contained in a vertical edge of $R_\Box$, which are supported by $1$-lines. Happy to remove all of this discussion, just something that took me a moment while reading.}
     Finally, $\plan(\nu^+)$ is $x$-separated because the robots become $x$-separated 
     at time~$\nu$; thus, $X$ moves horizontally at time $\nu$, and $X$ continues
     to move away from the other robot during $[\nu,\nu^+]$ because this time interval
     does not contain a vertex of $\pth_X$ by definition.
     Note that this argument actually shows that $\plan[\nu,\nu^+]$ is $x$-separated at all times.
	
	Since $[\mu_0,\mu_1]$ is a swap interval, the configuration $\plan(\nu^+)$ has 
    the same $x$-order as $\plan(\mu_i)$ for some $i\in \{0,1\}$. We will re-parametrize
    $\plan[\nu^+,\mu_i]$ for this value of~$i$, as described next.
    \medskip
    
    Define $\mu'$ in the same way as in \lemref{plan-between-swap-endpoints}, as follows:
	If $\pth_B$ enters or leaves $\I(\rect)$ at time $\mu_i$, or $\mu_i=\lambda_2$, then we 
    set $\mu':= \mu_i$. Otherwise, let $X\in \{A,B\}$ be the robot that moves at time~$\mu_i$,
    causing the robots to start or stop being $y$-separated. We set $\mu'$ to the last time
    before $\mu_i$ such that $\pth_X(\mu')$ is a vertex of $\pth_X$ or $\pth_X(\mu')$ 
    lies on a horizontal grid line.
	
    We now construct a compliant re-parametrization $\plan^*$ of $\plan[\nu^+,\mu_i]$ 
    such that $\plan^*(\nu')$ is $x$-separated for all $\nu'\in [\nu^+,\mu_i]$. 
    Since, as observed above, $\plan[\nu,\nu^+]$ is $x$-separated at all times, 
    this will finish the proof.    
    Our approach is similar to that in the proof of \lemref{plan-between-swap-endpoints}.
    We wish to use \lemref{x-sep} to obtain a $(\plan(\nu^+),\plan(\mu'))$-plan 
    $\plan^*[\nu^+,\mu']$ in which $A$ only parks at $\pth_A(\nu^+)$ or $\pth_A(\mu')$ 
    and $B$ only parks at $\pth_B(\nu^+)$ or $\pth_B(\mu')$. To be able to apply
    \lemref{x-sep}, we first note that $\plan(\mu')$ has the same $x$-order as $\plan(\mu_i)$
    because both robots have the same $x$-positions at these times.
    Hence, $\plan(\mu')$ has the same $x$-order as $\plan(\nu^+)$.
    Next, we note that there is an $xy$-monotone $(\pth_B(\nu^+),\pth_B(\mu'))$-path 
    $\varphi_B$ in $\fre$, as $\pth_B[\nu,\mu_1]\subseteq Q$.
    The following claim asserts that we also have a short $x$-monotone $(\pth_A(\nu^+),\pth_A(\mu'))$-path in~$\fre$.
    %----------------------------------------------------------------------------------
    \begin{claim}\label{clm:two-swap-lemma2}
    There exists an $x$-monotone $(\pth_A(\nu^+),\pth_A(\mu'))$-path $\varphi_A$ in~$\fre$ with $\|\varphi_A\|\leq \|\pi_A[\nu^+,\mu']\|$.
    \end{claim}
%----------------------------------------------------------------------------------
    \begin{claimproof}
    If $\pth_A(\nu^+)\in \gamma(R)$, the claim follows from \lemref{xy-shortest path}.
    Now suppose $\pth_A(\nu^+)\in (Q\cup Q^+)\setminus \gamma(R)$. 
    %If $y(\topQ) > y(\botR)$, then the $x$-range of $Q$ is identical to that of $R$ and we have $R\subset Q$ and therefore $\pth_A(\mu')\in Q$. If $\pth_A(\nu^+)\in Q$, there is an $xy$-monotone $(\pth_A(\nu^+),\pth_A(\mu'))$-path in $\fre$. If $\pth_A(\nu^+)\in Q^+$, note that there is a vertical segment connecting $\pth_A(\nu^+)$ with $\topQ$ (Observation~\ref{obs:free-segments}), so in this case there also is an $xy$-monotone $(\pth_A(\nu^+),\pth_A(\mu'))$-path in $\fre$.
   	%We assume for the remainder of the proof that $y(\topQ) \leq y(\botR)$.
    Let $\mu^-\in [\mu_0,\mu_1]$ be such that 
    $\pth_B(\mu^-)_x = \pth_A(\mu^-)_x$; such a time exists because 
    $[\mu_0,\mu_1]$ is a swap interval.
    Note that $\pth_B(\mu^-)\in \Rud \cap Q\cap Z^-$. Indeed, by \lemref{BinZ} 
    and the assumption that $B$ is below $A$, we have $\pth_B(\mu^-)\in Z^-$; 
    and since $\pth_B(\mu^-)_x = \pth_A(\mu^-)_x$ and $\pth_A(\mu^-)\in R$, we have
    $\pth_B(\mu^-)\in \Rud$; and since $\pth_B[\nu,\mu_1]\subset Q$, we have $\pth_B(\mu^-)\in Q$.

    We distinguish two cases
    depending on whether $\pth_A(\nu^+)$ lies in $I_x$, the $x$-range of $\rect$.
    \begin{itemize}
    \item 
     $\pth_A(\nu^+)_x \in I_x$.
     
     In this case, $\pth_A(\nu^+)_x$ must lie in the $x$-range of the rectangle~$Z^*$,
     where $Z^*$ is the component of $Z^-\cap \fre$ that contains~$\pth_B(\mu^-)$.
     To see this, consider the bounding box $\B$ of $\pth_A(\nu^+)$ and $\pth_B(\mu^-)$. 
     Note that if $\pi_A(\nu^+)_x\in I_x$ and $\pi_A(\nu^+)\in Q^+$, then $\pi_A(\nu^+)\in \Rud$ (as $y(\topQ) \leq y(\botR)$) and therefore $\pi_A(\nu^+)\in \gamma(R)$. Since we assumed $\pi_A(\nu^+)\notin \gamma(R)$, we 
     must therefore have $\pth_A(\nu^+)\in Q$. 
     Thus, $\pth_A(\nu^+)\in Q$ and $\pth_B(\mu^-) \in Q$, and so 
     $\B \subset Q$.
     Therefore $\pth_A(\nu^+)_x$ must lie within the $x$-range of~$Z^*$,
     otherwise the interior of $\B$ would intersect a vertical edge of
     $Z^*$, contradicting that $\B\subset Q\subset \fre$.
     This implies that the vertical segment $g$ from $\pth_A(\nu^+)$ to $\botR$
		    lies in $\I(\rect) \subset\F$. See \figref{before-lambda1-x-sep}(i).
     Hence, the path composed of $g$ followed by 
     an L-shaped path from the top endpoint of $g$ to $\pth_A(\mu')$
     yields an $xy$-monotone $(\pth_A(\nu^+),\pth_A(\mu'))$-path
     that lies in $\fre$; this is the desired $x$-monotone shortest path~$\varphi_A$.
    \item 
    $\pth_A(\nu^+)_x \not\in I_x$. %and $\pi_A(\nu^+)\in Q$:
    
    Then the horizontal segment between $(\pth_A(\nu^+)_x,\pth_B(\mu^-)_y)$ and $\pth_B(\mu^-)$
    crosses a vertical edge $\psi$ of $\Rud$; this is true because $\pth_B(\mu^-)\in \Rud$ and 
    $\pth_A(\nu^+)\notin \Rud$.
    Let $\psi'\supset \psi$ be the maximal vertical segment in $\fre$ that contains~$\psi$.
    We distinguish two subcases depending on whether $\pi_A(\nu^+)$ lies in $Q$ or in~$Q^+$.
    \begin{itemize}
    \item $\pi_A(\nu^+)\in Q$.
        
    Note that $\psi'$ is reachable from $\pth_A(\nu^+)$ via a horizontal segment in $\fre$,
    because the point $(x(\psi),\pth_B(\mu^-)_y)$ lies on $\psi$
    and the vertical segment between this point and $(x(\psi),\pth_A(\nu^+)_y)$ lies in $Q$. 
    (The latter follows because $\pth_A(\nu^+)\in Q$ and $\pth_B(\mu^-)\in Q$; see the first paragraph
    of the proof of the current claim.) % Claim~\ref{clm:two-swap-lemma}.)
    Thus, the horizontal segment $\pth_A(\nu^+) (x(\psi),\pth_A(\nu^+)_y)$ lies in~$Q$
    and, hence, in~$\fre$.
    The segment~$\psi'$ is also reachable via a horizontal line from $\pth_A(\mu')$,
    because $\pth_A(\mu')\in\R$. Hence, there is an $xy$-monotone $(\pth_A(\nu^+),\pth_A(\mu'))$-path $\varphi_A\subset\fre$, which consists of these two horizontal segments
    and the appropriate portion of~$\psi'$ (see \figref{before-lambda1-x-sep}(ii)).
    \item $\pi_A(\nu^+)\in Q^+$. 

     % Recall that $q$ is the top endpoint of~$\psi'$.
     We will construct an $x$-monotone 
     path $\varphi_A$ from $\pth_A(\nu^+)$ to $\pth_A(\mu')$ %that goes via~$q$ and 
     with 
     $\|\varphi_A\|\leq \|\pi_A[\nu^+,\mu']\|$.
     \smallskip
    
    Let $p \assign \pth_A(\nu^+)$, let $h \assign \topQ$, and let $q \assign (\pth_A(\mu')_x,y(\botR))$ be the vertical projection of $\pth_A(\mu')$ onto $\botR$.
    Then the triple $p,q,h$ satisfies the conditions of \lemref{x-mono}. 
    To see this, we first note that $p,q,h$ all lie in $\fre$ by definition.
    Moreover, condition~(i) is satisfied: we have $|p_y-y(\topQ)|\leq 1$ because $p\in Q^+$,
    and we have $|q_y-y(\topQ)| = |y(\botR)-y(\topQ)| \leq 1$ because we assumed 
    $y(\topQ) \leq y(\botR)$ and $y(\topQ) \geq y(\botRsq) = y(\botR) -1$.
    (The latter holds since $\Rsq\cap Q\neq \emptyset$, which is true because $\pi_B(\mu_1)\in \Rsq\cap Q$.)
    We next show that condition~(ii) is satisfied.
    Observe that the point $p'$ is the vertical projection of $p$ onto $\topQ$ because $p \in Q^+$. 
    Thus, $pp'$ is a vertical segment in~$\fre$. 
    There is also an $xy$-monotone path from $q$ to $q'$:
    first move horizontally along $\botR$ from $q$ to $(q'_x,q_y)$---this move is empty if
    $q$ lies in the $x$-range of $\topQ$---and then move vertically to~$q'$.
    % { If $q_x$ lies in the $x$-range of $\topQ$ then $q'$ is the vertical projection of $q$ onto $\topQ$, 
    % and otherwise, we have that $q_x$ lies to the left of $\topQ$ (indeed, $q$ cannot lie to 
    % the right of $\topQ$ since $q$ lies above or left of $\psi'$ which is intersected by $\botQ$).
    % In either case, there is an $xy$-monotone path from $q$ to $q'$: in the former case, it is the vertical segment $qq'$ in~$\fre$, and otherwise it is an L-shaped path with vertex $(q'_x,q_y) \in \R$.} 
    % \newtext{MdB: We seem to assume that $\psi'$ is aligned with the right edge of $R$,
    % which need not be the case? Anyway, I don't see why we need this.}
    % \ajs{Your edit is good to me, commented out my lines. I thought it may be unclear why the segment from $q$ to $(q'_x,q_y)$ was indeed contained in $\botR$ to be a feasible move, so I tried to cover that. But perhaps it is clear. You are right, $\psi'$ may be aligned with the left---for some reason I thought we had $R$ left of $A$ time $\nu^+$ wlog, which is not the case upon rereading the above.}
    Hence condition~(ii) is satisfied.
    Moreover, condition~(iii) is satisfied because $p$ lies in the $x$-range of~$\topQ$.
    By \lemref{x-mono} we therefore have that all shortest $pq$-paths in $\fre$ are $x$-monotone.
    Take such an $x$-monotone path $\varphi_A'$, and extend it with a vertical segment from $q$ up to $\pth_A(\mu')$ (see \figref{before-lambda1-x-sep}(iii)).
    % (Such an L-shaped path exists because both $q$ and $\pth_A(\mu')$ lie in~$R$.)
    Let $\varphi_A$ be the resulting $(\pth_A(\nu^+),\pth_A(\mu'))$-path.
    Clearly $\varphi_A$ is $x$-monotone since $\varphi_A'$ is $x$-monotone.
    % Because $\varphi_A$ does not make a U-turn at $q$ by construction, we conclude that $\varphi_A$ is $x$-monotone.
    % \newtext{Moreover, $\varphi_A$ must not cross below $\topQ$.}
    % \mdb{What is not being crossed (or: should not be crossed?) and why is it relevant?}
    \smallskip

    Next, we argue that $\|\varphi_A\|\leq \|\pi_A[\nu^+,\mu']\|$.
    Since $\lambda_1\in [\nu,\mu']$ and $A$ does not move vertically during $[\nu,\nu^+)$,
    the path $\pi_A[\nu^+,\mu']$ visits $\botR$ at some time $\nu' \in [\nu^+,\mu']$.
    Let $\varphi_A''$ be the $(\pth_A(\nu^+),\pth_A(\mu'))$-path obtained by concatenating to
    $\pi_A[\nu^+,\nu']$ the $L$-shaped path from $\pth_A(\nu')$ to $\pth_A(\mu')$ via~$q$.
    It follows that $\| \varphi_A \| \leq \|\varphi_A''\| \leq \|\pth_A[\nu^+,\mu']\|$.
    Indeed, the first inequality holds since $\varphi_A$ and $\varphi_A''$ are $(\pth_A(\nu^+),\pth_A(\mu'))$-paths 
    that visit~$q$ and $\varphi_A$ is a shortest such path by construction, and the second inequality holds since $\varphi''$ 
    is the same as $\pi_A[\nu^+,\mu']$ until $\pth_A(\nu')$ is reached and 
    its subsequent path to $\pth_A(\mu')$ is an L-shape.

    \newif\ifOldproof
    \Oldprooffalse
    \ifOldproof
     Since $\lambda_1\in [\nu,\mu']$ and $A$ does not move vertically during $[\nu,\nu^+)$,
        %\ajs{I am not sure why $A$ remains parked in this interval, or why it is needed.}
        %\mdb{We know that the path $[\nu,\mu']$ visits $\botR$, and we want to rule out that it does so during $[\nu,\nu^+)$. But now I don't see anymore why $A$ remains parked during $[\nu,\nu^+)$.cBut if $A$ is not parked it moves horizontally during $[\nu,\nu^+)$, so it will not visit $\botR$ either.cRight?}
        % \ajs{I agree it does not visit $\botR$ while moving horizontally, unless $A$ were to be moving on the line supported by it, which I think is possible (with the assumptions made thus far)? Which is an easy case to handle separately, of course, if so.}\mst{All we need here is that $\pi_A[\nu, \nu^+)$ does not properly cross $\botR$, which follows from $A$ either parking or on a horizontal segment during that interval.}
    the path $\pi_A[\nu^+,\mu']$ visits $\botR$. 
     This means it suffices to obtain an $x$-monotone path from $\pth_A(\nu^+)$ to $\pth_A(\mu')$
     that is shortest among all such paths that visit $\botR$. 
    %  \mdb{New proof for this case, which uses the new \lemref{x-mono}.}
    
    Let $q:= \botR \cap \psi'$ and let $\varphi'_A$ be a shortest $(\pth_A(\nu^+),\pth_A(\mu'))$-path 
    among all such paths that visit $\botR$. We may assume (see below) that $\varphi'_A$ crosses $\psi'$.
        \ajs{Specifically below $\botR$?}
        \mdb{I think so.}\mst{I don't think this restriction is needed, but we can prove it just as well. It also is more intuitive.}
    This implies that there exists a $(\pth_A(\nu^+),\pth_A(\mu'))$-path $\varphi_A$
    of equal length to $\varphi'_A$ that visits $q$: follow $\varphi'_A$ until the first intersection 
    with $\psi'$, then move along $\psi'$ to $q$, followed by an L-shaped path to $\pth_A(\mu')$).
    (Note that any $(\pth_A(\nu^+),\pth_A(\mu'))$-path that visits $\botR$ but does not cross
    $\psi'$ must be longer than $\varphi_A$; hence the assumption we just made is valid.  \mdb{Hopefully this is sufficiently clear.} 
        \ajs{I am finding this circular as it is not evident to me that, without the assumption, $\varphi_A$ is feasible. The details in the next para argue that it is, namely the subpath from $\pth_A(\nu^+)$ to $q$, but not so far, right?} 
        \mdb{I think feasibility is clear, and the next paragraph only argues about $x$-monotonicity?} \ajs{Oops, you're right about the next para. I see that any path from $\pi_A(\nu^+)$ to $\botR$ must cross the vertical line supporting $\psi'$ but I don't see why it is evident (from what is said here) that there is a $xy$-monotone path from $\pi_A(\nu^+)$ that visits $\psi'$ (approaching from the right side of $\psi'$). For concreteness, I was expecting details here similar to the start of the previous case: the vertical segment from $\pi_A(\nu^+)$ downwards to $\topQ$ is free, etc.}
        \mst{In general, there does not have to be an $xy$-monotone path from $\pi_A(\nu^+)$ to $\botR$. $\phi'_A$ is an arbitrary $(\pi_A(\nu^+),\pi_A(\mu'))$-path that is shortest among all such paths that visit $\botR$. $\phi'_A$ exists since $\pi_A[\nu^+,\mu']$ visits $\botR$. We modify this path to argue there exists a path $\phi_A$ that visits the point $q$ in addition to the properties of $\phi'_A$. The part of $\phi_A$ from $q$ to $\pi_A(\mu^+)$ is $xy$-monotone, as it is shortest and lies in $R$. The part of $\phi_A$ from $\pi_A(\nu^+)$ to $q$ does not have to be $xy$-monotone. In the text below, we show that it must be $x$-monotone. This is the part where we use the fact that, among others, the segment between $\pi_A(\nu^+)$ and its projection on $\topQ$ lies in $\fre$.}
        \ajs{I still don't understand the argument, sorry. My understanding is that (1) we assume $\varphi'_A$ visits $\psi'$ (and wish to show this indeed holds), (2) let $\varphi_A$ be a path whose existence depends on the assumption that $\varphi'_A$ visits $\phi'$, then (3) conclude that if $\varphi'$ (or any other path between its endpoints) did not cross $\psi'$ that it would be longer than $\varphi_A$. But in (3), if $\varphi'_A$ does not visit $\psi'$ at all, $\varphi_A$ is undefined, so there is no contradiction. What am I missing? Sorry if it is so obvious!}
        \mdb{We can define $\varphi'_A$ as the shortest path that visits $\botR$ and must also cross
        $\psi'$. We can turn this into a path $\varphi_A$, and any path that does not cross $\psi'$
        is longer than $\varphi_A$. Does that make sense?}
        \ajs{I am not sure a shortest path that crosses is well-defined (if it makes a $U$-turn on $\psi'$ then it may only visit and leave on the same side). $\pi_A(\mu')$ could also be $q$ or directly above it. But I think it all works if we say "visits" $\psi'$.}
    We will argue that $\varphi_A$ is $x$-monotone.
    
    Let $p:= \pth_A(\nu^+)$ and let $h := \topQ$.
    Then the triple $p,q,h$ satisfies the conditions of \lemref{x-mono}. 
    To see this, we first note that $p,q,h$ all lie in $\fre$ by definition.
    Moreover, condition~(i) is satisfied: we have $|p_y-y(\topQ)|\leq 1$ because $p\in Q^+$,
    and we have $|q_y-y(\topQ)| = |y(\botR)-y(\topQ)| \leq 1$ because we assumed 
    $y(\topQ) \leq y(\botR)$ and $y(\topQ) \geq y(\botRsq) = y(\botR) -1$.
    (The latter holds since $\Rsq\cap Q\neq \emptyset$, which is true because $\pi_B(\mu_1)\in \Rsq\cap Q$.)
    To show that conditions~(ii) and~(iii) are satisfied, we observe that the points $p'$ and $q'$ 
    are the vertical projections of $p$ and $q$ onto $\topQ$, respectively,
    because $p\in Q^+$ and $q_x = x(\psi)$ (and $x(\psi)$ lies in the
    $x$-range of $Q$).
    Thus, $pp'$ and $qq'$ are vertical segments in~$\fre$, which implies
    conditions~(ii) and~(iii) are satisfied.
    By \lemref{x-mono} we therefore have that all shortest $pq$-paths in $\fre$ are $x$-monotone. %an $x$-monotone shortest $pq$-path.
    
    We now observe that the subpath of $\varphi_A$ from $p$ to $q$ is 
    shortest in $\fre$ and therefore $x$-monotone. 
    The subpath of $\varphi_A$ from $q$ to $\pth_A(\mu')$ must be $xy$-monotone 
    as it is shortest and $q\in R$ and $\pth_A(\mu')\in R$. Because the path does not make a U-turn
    at $q$ by construction we conclude that the path $\varphi_A$ is $x$-monotone. 
    \newtext{END OF OLD PROOF}
    \fi
    \end{itemize}
%    \begin{align*}
%    	\|\varphi''_A\|_y &= |\pth_A(\mu')_y - y(\topQ)| + |y(\topR) - y(\topQ)| \\
%    	&= 2\cdot \max(|\pth_A(\mu')_y - y(\topQ) |, |y(\topR) - y(\topQ) |)\\
%    	&\quad - ||\pth_A(\mu')_y - y(\topQ)| - |y(\topR) - y(\topQ)||\\
%    	&\leq 2 - |\pth_A(\mu')_y - y(\topR) |
%    \end{align*}
    
    %the projection of $\pth_A(\mu')$ onto $\topQ$, and finally along a vertical line to $\pth_A(\mu')$. The path $\psi_A$ lies in $\fre$: the first edge lies in $\fre$ by Observation~\ref{obs:free-segments}, the second edge lies in $Q\subset \fre$, and the final edge lies in $\Rud \cap \fre$ (as the region is $x$-monotone). The path $\psi_A$ is $x$-monotone by construction. We now show that the length of $\psi_A$ is shorter than any %TODO Note that there may be non-monotone paths that are shorter! Only a non-monotone path without a monotone shortcut is longer than \psi_A. 
     %First, note there is a shortest $(\pth_A(\nu^+),\pth_A(\mu'))$-path in $\fre$ that lies within $Q\cup Q^+\cup R\cup (\Rud \cap (\Rxy)^-)$: if not, then there is a shortest path that enters $R$ via $\topR$, as any other path between points in $Q\cup Q^+\cup R\cup (\Rud \cap (\Rxy)^-)$ that leaves the region is not shortest. (This is clear for points in $Q\cup R\cup (\Rud \cap (\Rxy)^-)$, who are connected with $xy$-monotone paths within this region. )
    
    %enter $R$ via $\botR$: if not, then there exists a shortest $(\pth_A(\nu^+),\pth_A(\mu'))$-path $\psi_A$ that enters $R$. %via $\topR$ and therefore 
    \end{itemize}
    Thus, in both cases we have an 
    an $x$-monotone $(\pth_A(\nu^+),\pth_A(\mu'))$-path in~$\fre$ with length at most $\|\pth_A[\nu^+,\mu']\|$, which proves the claim.
    \end{claimproof}
     %----------------------------------------------------------------------------------
     We conclude that the conditions of Lemma~\ref{lem:x-sep} are satisfied,
     giving us a $(\plan(\nu^+),\plan(\mu'))$-plan $\plan^*[\nu^+,\mu']$ in which $A$ only parks at $\pth_A(\nu^+)$ or $\pth_A(\mu')$ and $B$ only parks at $\pth_B(\nu^+)$ or $\pth_B(\mu')$.\footnote{Just as in Lemma~\ref{lem:plan-between-swap-endpoints}, this plan is not necessarily a re-parametrization. However, since we have a valid plan from $\plan(\nu^+)$ to $\plan(\mu')$ that uses $xy$-monotone paths in $\fre$, the paths $\pth_A[\nu^+,\mu']$ and $\pth_B[\nu^+,\mu']$ must be $xy$-monotone as well, since $\plan$ is optimal. So, we can apply Lemma~\ref{lem:x-sep} with these paths to obtain a plan with the required properties that is also a re-parametrization of $\plan$.}
    We extend the domain of $\plan^*$ to $[\nu^+,\mu_i]$ by setting $\plan^*[\mu',\mu_i]:= \plan[\mu',\mu_i]$. Observe that 
    for all $\nu'\in [\nu^+,\mu_i]$, $\plan^*(\nu')$ is $x$-separated,
    so it remains to argue that the potential parking spots at times $\nu^+$ and $\mu'$ 
    are in valid positions.
    
    Note that in the cases where $\nu^+ = \nu$, the time $\nu$ is a moment where a robot moves onto or leaves a grid line, so the moving robot $X$ is at the intersection of a grid line and an edge of $\pth_X$, which is a valid parking spot. In the remaining case, the positions of $\plan^*(\nu^+)$ are valid parking spots by construction.
    The argument for $\mu'$ is similar: if $\pth_B$ enters or leaves at time $\mu_i$, the moving robot $B$ is at the orthogonal intersection of an edge of $\pth_B$ with a grid line. If $\lambda_2 = \mu_i$, then $A$ is the moving robot and is at an orthogonal intersection of an edge of $\pth_A$ and $\topR$.
    The remaining option is that $\mu_i$ is a time when $\plan$ stops or becomes $y$-separated. In this case, both locations of $\plan(\mu')$ are valid parking spots by construction.
\end{proof}

%----------------------------------------------------------------------------
\subsection{Surgery on the paths}\label{sec:surgery}
%----------------------------------------------------------------------------
We are now ready to describe the surgery we perform on $\pth_A$ and $\pth_B$
to align a given bad horizontal segment~$e$ of $\pth_A$ with a grid line. 
By performing compliant modifications as necessary, we can assume that the optimal
plan $\plan$ satisfies the conditions of 
Lemmas~\ref{lem:main-properties}--\ref{lem:before-lambda1-x-sep-below}.
Ideally, we would like to restrict the surgery to the interval $[\lambda_1, \lambda_2]$,
but we have to perform surgery beyond this interval in some cases. 
The surgery consists of a sequence of one or more \push operations, which 
push a part of $\pth_A$ or $\pth_B$ onto a grid line. We first describe
the \push procedure and then describe the sequence of pushes we perform in
the various cases that can arise. We use $\plan^*$ to denote the (modified) 
plan after the surgery.

%---------------------------------------------------------------------------
\paragraph{The \push procedure.}
%---------------------------------------------------------------------------
A \push operation takes three parameters: a robot $X \in \{A,B\}$, a time
interval $I = [\mu_1, \mu_2]$, and a $y$-coordinate $y^*$ of a horizontal grid line---recall
that a grid line is a line from the set $L_{\leq 2}$---and it pushes the subpath $\pth_X[I]$
onto the grid line $\ell: y=y^*$. More formally, the new plan $\plan_X^*$
resulting from the operation $\push(X,I,y^*)$  is obtained by setting
\begin{equation}\label{eqn:push}
\xco{\pth_X^*(\mu)} := \xco{\pth_X(\mu)} \mbox{ for all $\mu$} \hspace*{5mm} \mbox{ and } 
\hspace*{5mm} \yco{\pth_X^*(\mu)} :=
\begin{cases}
   y^*  & \mbox{ if $\mu\in I$},\\
  \yco{\pth_X(\mu)} & \mbox {otherwise.}
\end{cases}
\end{equation}
Thus, a \push operation only changes the $y$-coordinates on the subpath $\pth_X[I]$.
The exact surgery---the values of the parameters $X,I,y^*$ used in the \push
operations that we perform---depends on the location of $\pth_A(\lambda_1)$, 
the set of unsafe and swap intervals during $[\lambda_1, \lambda_2]$, 
and the position of $B$ during the swap intervals, as explained later.
\medskip

A \push operation may have side effects that we need to address. 
First, $\push(X,I,y^*)$ may introduce a discontinuity at $\mu_1$ or $\mu_2$, 
the endpoints of~$I$. 
In other words, for $\mu\in\{\mu_1,\mu_2\}$ we may have 
$p_X^-(\mu)\neq p_X^+(\mu)$, where $p_X^-(\mu) := \lim_{\lambda \uparrow \mu} \pth^*_X(\lambda)$ 
and $p_X^+(\mu) := \lim_{\lambda \downarrow \mu} \pth^*_X(\lambda)$. 
This happens when $\pth_X(\mu)\neq \pth^*_X(\mu)$.
Recall that $\xco{\pth^*_X(\lambda)} = \xco{\pth_X(\lambda)}$ for all $\lambda$ and so
$p_X^-(\mu)_x = p_X^+(\mu)_x$.  We can therefore resolve the discontinuity
by adding the vertical segment $g_X(\mu) = p_X^-(\mu) p_X^+(\mu)$, called a \emph{ghost segment},
to $\pth_X^*$ to make the path continuous.
We will also have to re-parametrize $\plan^*$ to incorporate this ghost segment
while satisfying the properties (P1)--(P5). Since $\plan^*$ is decoupled, it suffices
to describe the ``moves'' of $A$ and $B$ in the neighborhood of $\mu$,
which we will do after describing the surgery.
%--------------------------------------------------------------------------
\begin{figure}
	\centering
	\includegraphics{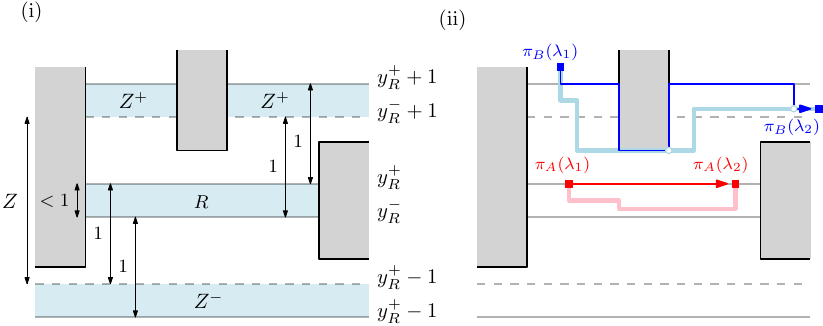}
	\caption{(i) Example of a corridor $\R$, the surrounding lines and areas within $\R+2\square$.
		(ii) The modification in Case~I. The original paths are shown in pink and light blue. 
        The light-blue disks mark the endpoints of a swap interval on $\pth_B$.}
	\label{fig:short-corridor}
\end{figure}
%--------------------------------------------------------------------------

The second side effect of pushing $\pth_X[I]$ to $\pth^*_X[I]$ is that $X$ may collide with
the other robot during the interval $I$. This can only happen when $X=A$ and during 
unsafe intervals~$[\nu_1,\nu_2]\subset I$.
During such intervals, we then perform a set of 
\emph{secondary pushes}. As we will see below, secondary pushes are always performed
on $\pth_B$, that is, collisions created by primary pushes on $\pth_A$
are always resolved by letting robot~$\robA$ push $\robB$ out of the way.
More precisely, let $\nu \in I$ be a time instance such that $\pth_B(\nu)$
is in conflict with $\pth_A^*(\nu)$ when $\nu$ lies in an unsafe interval
$[\nu_1, \nu_2]$. We execute $\push(B, (\nu_1,\nu_2), y^*)$, where
\begin{equation}\label{eq:remove-U-turn-B-Case-I}
	y^* := \begin{cases}
		y(\topR) + 1 & \mbox{if } \yco{\pth_B(\nu)} > \yco{\pth_A(\nu)},\\
		y(\botR) -1 & \mbox{if } \yco{\pth_B(\nu)} < \yco{\pth_A(\nu)}.
	\end{cases}
\end{equation}
Note that this is well defined, since by definition of unsafe intervals,
$B$ is either above $A$ during the entire interval, or $B$ is below
$A$ during the entire interval. A secondary 
push operation may also create discontinuities at times $\nu_1$ and $\nu_2$.

A third side effect is that a \push operation may collapse a vertical segment of $\plan$ 
into a 0-length vertical segment. This is not a problem, as long as we keep the
0-length segment so that the alternating property is maintained.

%--------------------------------------------------------------------------
\paragraph{The overall surgery.}
%--------------------------------------------------------------------------
Without loss of generality, we assume that $\pth_A(\lambda_2)\in \topR$. 
The overall surgery of $\plan$ depends on the location of $\pth_A(\lambda_1)$,
the number of swap intervals during $[\lambda_1, \lambda_2]$, which is at most two
by \lemref{main-properties}, and the position of $B$ in the swap intervals.
% Recall that $\R = [x_{\R}^-, x_{\R}^+] \times [y_{\R}^-, y_{\R}^+]$. 
Our modification of $\pth_A$ into $\pth_A^*$ follows three main cases.
%--------------------------------------------------------------------------
\begin{itemize}
\item 
    \textit{Case I: $\pth_A(\lambda_1) \in \topR$.} \\[1mm]
    % Thus both $\pth_A(\lambda_1)$ and $\pth_A(\lambda_2)$ lie on $\topR$. 
    We perform one push operation: $\push(A, [\lambda_1, \lambda_2], y(\topR))$; see \figref{short-corridor}(ii).
\item 
    \textit{Case II: $\pth_A(\lambda_1) \in \botR$.
    Furthermore, either there is at most one swap interval,
	or there are two swap intervals and $\pth_B$ lies above $\pth_A$ during the first swap interval.} 
    \\[1mm] 
	In this case, we choose a suitable time value $\bar{\lambda} \in [\lambda_1,\lambda_2]$
    and perform two push operations: $\push(A, [\lambda_1, \bar{\lambda}], y(\botR))$ and
	$\push(A,[\bar{\lambda}, \lambda_2], y(\topR))$. 
    Thus, robot $A$ moves along $\botR$ from time $\lambda_1$ to $\bar{\lambda}$, 
    and then it switches to moving along $\topR$ until time~$\lambda_2$.
    %----------------------------------------------------------------
    \begin{figure}
	\centering
	\includegraphics{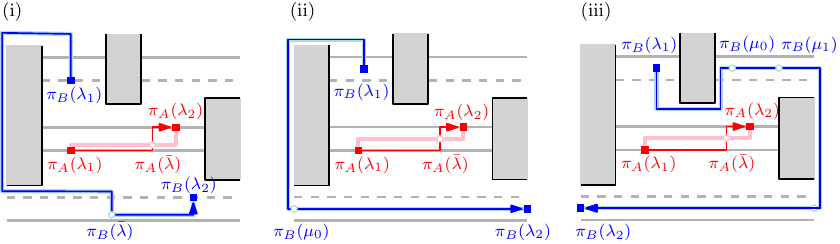}
	\caption{The modifications in the subcases of Case~II. The original paths are shown in pink and light blue.
	In all cases, $\plan^*_A$ jumps to $\topR$ at a suitable time $\bar{\lambda}$.
		(i) The modification in Case~II(a). $\bar{\lambda}$ is the moment when $\plan$ enters the last $y$-separated interval.
        (ii) The modification in Case~II(b). $\bar{\lambda}=\mu_1$ is the moment when $\plan$ enters the swap interval with $B$ below $A$.
        (iii) The modification in Case~II(c). $\bar{\lambda}=\mu_2$ is the moment when $\plan$ leaves the swap interval with $B$ above $A$.}
	\label{fig:short-corridor-II}
    \end{figure} 
    %------------------------------------------------------------
	We have three subcases for the choice of $\bar{\lambda}$,
    illustrated in \figref{short-corridor-II}.
	\begin{itemize}
	\item \emph{Case II(a): There are no swap intervals.} 
		If all configurations in $\plan[\lambda_1,\lambda_2]$ are $y$-separated then either $B$ lies below $A$ during the interval $[\lambda_1,\lambda_2]$ or above $A$ during the entire interval.
			We set $\bar{\lambda}=\lambda_1$ (resp.\ $\bar{\lambda}=\lambda_2$) if
			$B$ lies below (resp.\ above) $A$ during $[\lambda_1,\lambda_2]$. 
			Otherwise there is a configuration in $\plan[\lambda_1,\lambda_2]$ that is not 
			$y$-separated. We then set 
		$$\bar{\lambda}:= \sup\{\lambda\in [\lambda_1,\lambda_2] : \text{$\plan(\lambda)$ is 
			not $y$-separated} \}.$$
			Thus, $\plan(\lambda)$ is $y$-separated for all 
			$\lambda \in [\bar{\lambda}, \lambda_2]$.
		%Otherwise we set $\bar{\lambda}:=\lambda_2$.
    \item \emph{Case II(b): There is one swap interval and $B$ is  below $A$ during this swap interval.}
        Let $[\mu_0,\mu_1]$ be this swap interval. Set $\bar{\lambda}:= \mu_0$.
    \item \emph{Case II(c): There is at least one swap interval and $B$
		lies above $A$ during the first one.}
        Let $[\mu_0,\mu_1]$ be this swap interval. Set $\bar{\lambda}:= \mu_1$.
	\end{itemize}
    Note that we may create a discontinuity in $\pth_A$ at $\bar{\lambda}$,
    and we thus add a ghost segment at $\pth_A(\bar{\lambda})$.
    Using \lemref{unsafe-vertical-edge}, it can be verified that this ghost 
    segment is aligned with a grid line or a vertical segment of $\pth_A$ or $\pth_B$.
    Furthermore if $\bar{\lambda} \not= \{\lambda_1, \lambda_2\}$ then
	$\plan(\bar{\lambda})$ is both $x$- and $y$-separated.
    %------------------------------------------------------------
    \begin{figure}[htb]
	\centering
    \includegraphics{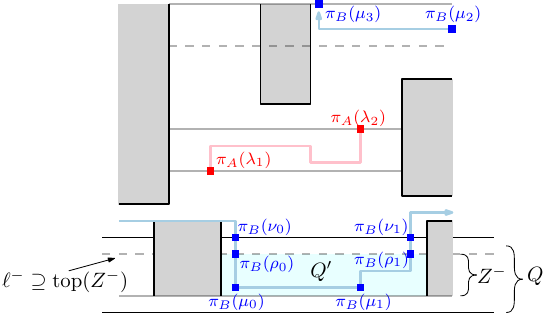}
	    \caption{Set up for Case III. $Z^-$, rectangles $Q$ and $Q'$,
	    and time instances $\nu_0, \nu_1, \rho_0, \rho_1$.}
	    \label{fig:Case-III-set-up}
    \end{figure} 
    %------------------------------------------------------------
\item
    \textit{Case III: $\pth_A(\lambda_1) \in \botR$ and there are exactly two swap intervals, 
    and $\pth_B$ lies below $\pth_A$ during the first swap interval.}
    \\[1mm]
	Let $[\mu_0,\mu_1]$ and $[\mu_2,\mu_3]$ be the two swap intervals. 
    Assume wlog that $\mu_1<\mu_2$, that is, $[\mu_0,\mu_1]$ is the 
    first swap interval. Since all configurations in $\plan[\mu_0,\mu_1]$
    are $y$-separated and $B$ lies below $A$ during $[\mu_0,\mu_1]$, 
    we conclude that $\pth_B[\mu_0,\mu_1]$ lies in a connected component 
    of $\F \cap Z^-$, which is a rectangle $Q'$. No horizontal line
    of $L_0$, the set of lines supporting the edges of $\bd\fre$,
    intersects the interior of $Z^-$
    because then a horizontal primary grid line
    would intersect the interior of~$R$, which is impossible by the 
    construction of the subdivision $\horF$.
    
    By the above discussion, $\pth_B[\mu_0,\mu_1]$ 
    is contained in a rectangle $Q \supseteq Q'$ of $\horzF$, the subdivision
    of $\F$ induced by the horizontal lines of $L_0$. 
    Let $[\nu_0,\nu_1]$ be the maximal interval containing $[\mu_0,\mu_1]$ 
    such that $\pth_B(\nu_0,\nu_1) \subset \interior(Q)$.
    % this is well defined since $s_B,t_B \not \in \interior(Q)$.
    % We have $\nu_0 \leq \mu_0 < \mu_1 \leq \nu_1$. 
    Note that at most one of $\pth_B(\nu_0),\pth_B(\nu_1)$ lies on $\botQ$, 
    as otherwise pushing $\pth_B[\nu_0,\nu_1]$ to $\botQ$ yields a strictly shorter plan.
	Indeed, the original path $\pth_B$ entered the influence region~$\I(\rect)$
    and so $\pth_B[\nu_0,\nu_1]\not\subset \botQ$; hence, pushing $\pth_B[\nu_0,\nu_1]$ 
    onto $\botQ$ yields a strictly shorter path.

    Let $[\rho_0,\rho_1] \supseteq [\mu_0,\mu_1]$ be the maximal interval 
    such that $\pth_B[\rho_0,\rho_1]$ lies on or below the horizontal line 
    $\ell^-\assign y=y(\topR) - 1$ that contains $\topx{Z^-}$. 
    Note that $\rho_1 < \lambda_2$ since there is a swap interval with $B$ above~$A$.
	Furthermore, we may have $\rho_0<\lambda_1$, and if both $\pth_B(\nu_0)$ and $\pth_B(\nu_1)$ lie above or on~$\ell^-$ then we have $[\rho_0,\rho_1]\subseteq[\nu_0,\nu_1]$.

    %We now describe the push operations we perform in Case~III.
    There are three subcases, as illustrated in \figref{short-corridor-III}, 
    depending on which of $\pth_B(\nu_0)$ and $\pth_B(\nu_1)$ lie on $\botQ$.
    These are indeed all subcases since, as argued above,
    we cannot have that $\pth_B(\nu_0)$ and $\pth_B(\nu_1)$ both lie on~$\botQ$.  
    In each subcase, we perform two push operations:
		\begin{itemize}
		\item \emph{Case III(a): $\pth_B(\nu_0)\in \botQ$ and $\pth_B(\nu_1)\in \topQ$.} See Figure \ref{fig:short-corridor-III} (i).
            We perform $\push(A, [\lambda_1, \lambda_2], y(\botR))$
            and $\push(B, [\nu_0, \mu_1], y(\botQ))$. 
		\item \emph{Case III(b): $\pth_B(\nu_0)\in \topQ$ and $\pth_B(\nu_1)\in \botQ$.} See Figure \ref{fig:short-corridor-III} (ii).
			We perform $\push(A, [\lambda_1, \lambda_2], y(\botR))$
            and $\push(B, [\mu_0, \nu_1], y(\botQ))$.
		\item \emph{Case III(c): $\pth_B(\nu_0)\in \topQ$ and $\pth_B(\nu_1)\in \topQ$.} See Figure \ref{fig:short-corridor-III} (iii).
            We perform  $\push(A, [\lambda_1, \lambda_2], y(\topR))$
            and $\push(B, [\rho_0, \rho_1], y(\topR) - 1)$. 
	    Note that, unlike in the previous two subcases, we push $\robA$ to $\topR$ instead of to $\botR$ 
				and $\robB$ to $\topin{Z^-}$ instead of to $\botin{Q}$.
		\end{itemize}
\end{itemize}
The push operations may introduce a discontinuity in $\plan^*_A$ at time 
$\lambda_1$ or~$\lambda_2$. 
Furthermore, in Case~III(a), we may introduce a discontinuity in $\plan^*_B$
at time $\mu_1$, and in Case~III(b), we may do so at time 
$\mu_0$. In Case~III(c), we do not introduce a discontinuity 
at $\rho_0$ or $\rho_1$ since $\yco{\pth_B(\rho_0)}=\yco{\pth_B(\rho_1)}= 
y(\topR) - 1$. Again, it can be verified that all ghost segments added
to remove discontinuities are (vertically) aligned with a vertical
grid line or a vertical segment of $\pth_A$ or $\pth_B$. 
%Furthermore, $\plan(\mu_0)$ and $\plan(\mu_1)$ are both $x$- and $y$-separated, as they are endpoints of the swap intervals.
%\end{itemize}

     %---------------------------------------------------------------
  	\begin{figure}
	\centering
	\includegraphics{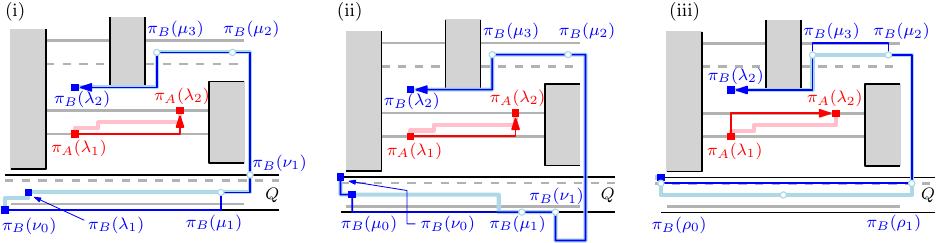}
	\caption{
	The modifications in the subcases of Case~III. The original paths are shown in pink and light blue. 
		(i) In Case~III(a) we push $\pth_A[\lambda_1,\lambda_2]$ onto $\botR$,
            and we push $\pth_B[\nu_0, \mu_1]$ onto $y(\botQ)$.
        (ii) In Case~III(b) we also push $\pth_A[\lambda_1,\lambda_2]$ onto $\botR$,
             and we push $\pth_B[\lambda_1, \nu_1]$ onto $y(\botQ)$.
             In this example, we have $\mu_0=\lambda_1$.
        (iii) In Case~III(c) we push $\pth_A[\lambda_1,\lambda_2]$ onto $\topR$,
              and we push $\pth_B[\rho_0,\rho_1]$ onto $\topR- 1$.
              We also have a secondary push, namely $\pth_B[\mu_2,\mu_3]$ onto $\topR+ 1$.}
	\label{fig:short-corridor-III}
    \end{figure} 
    %---------------------------------------------------------------
    
% \smallskip
This completes the description of the main surgery procedure. Recall that whenever a primary push 
operation on $\robA$ causes $\pth_A^*$ to conflict with $\pth_B^*$, we perform a secondary push 
operation on $\robB$ as described above.  We now describe how to resolve discontinuities introduced by the push operations.

%---------------------------------------------------------------------
\paragraph{Resolving discontinuities.}
\label{subpara:discontinuities}
%---------------------------------------------------------------------
We now describe how to re-parametrize $\plan^*$ to parametrize the motion along the ghost segments 
that were added to resolve discontinuities so that the resulting plan remains feasible. 
We begin by proving a few key properties of the time instance at which a ghost segment is
created. For $\lambda\in[\lambda_1,\lambda_2]$, we call $\plan(\lambda)$ \emph{$y$-semi-separated} 
if $\plan$ transitions from being $y$-separated to not being $y$-separated at $\lambda$, or vice-versa. 
In other words, $\plan(\lambda)$ is $y$-semi-separated if
$|\yco{\pth_A(\lambda)}-\yco{\pth_B(\lambda)}| = 1$ and there is a $\delta>0$ such 
that $\plan$ is not $y$-separated during $[\lambda-\delta,\lambda)$ or during $(\lambda,\lambda+\delta]$. 
%---------------------------------------------------------------------
\begin{lemma}\label{lem:ghost-segment-endpts}
    Let $\lambda$ be a time instance at which a ghost segment is created by a primary
    or secondary $\push$ operation performed during the surgery on $\plan$. Then
    the following properties hold.
    \begin{enumerate}[(i)]
    \item \label{ghosti} $\lambda \in [\lambda_1, \lambda_2]$.
	\item \label{ghostii} If $\lambda \not\in \{\lambda_1, \lambda_2\}$ and $\pth_B(\lambda)\not\in\partial\I(\rect)$, then
		$\plan(\lambda)$ is $y$-semi-separated.
	\item \label{ghostiii} If $\pth_B(\lambda)\in \topRsq\cup\botRsq$ then $\pth_B^*$ may have a discontinuity 
        at time~$\lambda$
		only in Cases~III(a,b) because of the primary push operation on $\robB$,
		    and the ghost segment $g_B(\lambda)=p_B^-(\lambda)p_B^+(\lambda)$ 
		    lies below the interior of $\Rsq$. 
    \end{enumerate}
\end{lemma}
%---------------------------------------------------------------------
\begin{proof}
    It follows from the above discussion---see Cases~II and~III---that a
    primary $\push$ creates a discontinuity only during the interval 
	$[\lambda_1, \lambda_2]$. On the other hand, if the discontinuity at $\lambda$ 
    is created by a secondary push then $\plan(\lambda)$ is the endpoint of an unsafe interval. 
    By the definition of unsafe intervals, we thus have~$\lambda\in[\lambda_1,\lambda_2]$. 
    Hence, (i)~holds.

	To prove~(ii), suppose a discontinuity occurs at $\lambda \not \in \{\lambda_1, \lambda_2\}$. 
    If this happens in Case~II(a) then $\plan(\lambda)$ is $y$-semi-separated.
    In Cases II(b,c) and~III, and for secondary pushes,
    we know that $\lambda$ is an endpoint of an unsafe interval. 
    If $\pth_B(\lambda) \not\in \partial\I(\rect)$ 
	and $\lambda\not\in\{\lambda_1,\lambda_2\}$, then
	the maximality condition of an unsafe interval implies that $\plan$ transitions from being 
	$y$-separated to not being $y$-separated or vice-versa at time~$\lambda$, and so
	$\plan(\lambda)$ is $y$-semi-separated. We conclude that (ii)~holds.

	Finally, suppose $\pth_B(\lambda) \in \topRsq\cup\botRsq$. 
    Note that a	discontinuity in $\pth_B^*(\lambda)$ can exist
	only if $\lambda$ is an endpoint of an unsafe interval~$I$.
	If a secondary push is performed on $B$ during $I$ and 
	$\pth_B(\lambda) \in \topRsq$ (resp.~$\pth_B(\lambda) \in \botRsq$) then
	the secondary $\push$ sets $\pth_B^*(\lambda)$ to $\topRsq$ (resp.~$\botRsq$),
    so there is no discontinuity in $\pth_B^*$ at $\lambda$. Thus, the discontinuity
    is caused by a primary push in Case~III(a) or Case~III(b),
    which implies that $\pth_B(\lambda)$ is the upper endpoint of the ghost segment~$g_B(\lambda)$.
	In Cases II~(a,b), $B$ lies below $A$ during the interval $[\mu_0,\mu_1]$ and a discontinuity 
	due to the primary push of $B$ occurs at $\mu_0$ or $\mu_1$, so $\pth_B(\lambda) \in \botRsq$, and thus 
	$g_B(\lambda)$ lies below the interior of $\Rsq$.
	%Similarly, $\pth_B^*$ does not have a discontinuity at $\lambda$ if $\pth_B(\lambda) \in \botRsq$. Hence (iii) also holds.
	This completes the proof of the lemma.
\end{proof}

We now describe how we re-parametrize the plan in the neighborhood of
$\lambda \in [\lambda_1, \lambda_2]$ if $\plan^*(\lambda)$ has a discontinuity and 
a ghost segment $g := g_X(\lambda)$ was created, for $X \in \{A,B\}$.
Let $Y \not= X$ be the other robot. 
Note that $\pth_Y^*$ may also have a discontinuity at~$\lambda$.

If $\lambda \in \{ \lambda_1,\lambda_2\}$ then we insert a breakpoint on~$\pth_A(\lambda)$. 
(If $\pth_A^*(\lambda)$ has a discontinuity then this breakpoint corresponds to $p_A^-(\lambda)$.)
Similarly, if $\pth_B(\lambda)\in\partial \I(\rect)$, we insert a breakpoint at $\pth_B(\lambda)$.
It is easily seen that the insertion of these breakpoints is a compliant modification of the plan.
Let $\plan[\lambda^-,\lambda^+]$ be the line segment (between two consecutive breakpoints)
of $\plan$ that contains $\plan(\lambda)$; 
if $\plan(\lambda)$ is a breakpoint of $\plan(\lambda)$ then $\lambda^-=\lambda^+=\lambda$.
In this case, we first add a tiny time interval $I_\lambda$\label{page:I_lambda} at $\lambda$, set 
$\plan(\lambda')=\plan(\lambda)$, for all $\lambda'\in I_\lambda$, and re-parametrize the plan 
during $I_\lambda$ so that $\pth^*_X[I_\lambda]=g$. 

Since $\plan$ is decoupled, one of the robots is parked at time $\lambda$ 
and the other is moving along a segment during the interval $[\lambda^-,\lambda^+]$. 
We describe the re-parametrization of
$\plan^*$ during the interval $[\lambda^-,\lambda^+]$ by describing a sequence 
of moves we perform.
%First assume that $\lambda \not \in \{\lambda_1, \lambda_2\}$, then, by \lemref{ghost-segment-endpts}, $\pth(\lambda)$ is $x$-separated. 
First, assume that $\pth_Y^*$ does not have a discontinuity at $\lambda$.
%Let $[\lambda^-, \lambda^+]$ be the time interval corresponding to the vertical segment on which the robot in motion at $\lambda$ is moving.  (There is only one such robot, since $\pth$ is decoupled.) Note that $\lambda \in [\lambda^-, \lambda^+]$ and $\lambda$ may be an endpoint of $[\lambda^-, \lambda^+]$.
There are three cases to consider:
\begin{enumerate}[(i)]
	\item \emph{$Y$ is parked at time $\lambda$.} Keep $Y$ parked at its parking location
		$\pth_Y^*(\lambda)$. First, move $X$ along $\pth_X^*[\lambda^-,\lambda)$, then along 
		the ghost segment $g$, and finally along $\pth_X^*(\lambda,\lambda^+]$.
    \item \emph{$\pth_Y^*(\lambda)$ is a breakpoint.} 
	    In this case, $\lambda^-=\lambda^+=\lambda$.
		Park $Y$ at $\pth_Y^*(\lambda)$ and move $X$ along the ghost segment $g$. (Recall that we have replaced $\lambda$ with a tiny interval $I_\lambda$ to parametrize the motion of $X$ 
		along $g$.)
		%and then follow $\plan^*$.
    \item \emph{$\pth_Y^*(\lambda)$ lies in the relative interior of the segment 
	    $e_Y := \pth_Y(\lambda^-)\pth_Y(\lambda^+)$ of $\pth_Y^*$.}
		Since $\pth_Y^*(\lambda)$ is not a breakpoint, 
		neither $\lambda\not\in\{\lambda_1,\lambda_2\}$ nor $\pth_B(\lambda)\not\in\partial \I(\rect)$.
		By \lemref{ghost-segment-endpts}, $\plan(\lambda)$ is  a $y$-semi-separated 
		configuration. The definition of a $y$-semi-separated 
		configuration implies that $\pth_Y(\lambda)$ is $x$-separated and 
		$e_Y$ is a vertical segment. 
		Since $\plan$ is decoupled and $Y$ is moving at time $\lambda$,
		$X$ is parked at $p_X^-(\lambda)$ at time $\lambda$. We park $Y$ at time 
		$\lambda^-$ (at location $\pth_Y^*(\lambda^-)$), move $X$ along the ghost 
		segment $g$, park $X$ at $p_X^+(\lambda)$, 
		move $Y$ along the segment $\pth_Y^*(\lambda^-)\pth_Y^*(\lambda^+)$.
\end{enumerate}

\begin{figure}
    \centering
    \includegraphics{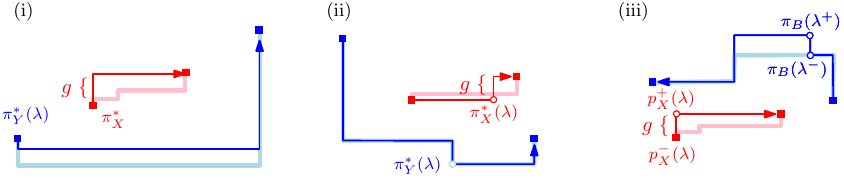}
    \caption{The cases to resolve a discontinuity at time
    $\lambda$. In case (i), we move $X$ along $g$ directly. In case (ii), we keep $Y$ parked at the breakpoint and move $X$ directly. In case (iii), we park $Y$ at $\pth_Y^*(\lambda^-)$, move $X$ along $g$, move $Y$ to $\pth_Y^*(\lambda^+)$, and then proceed with $\plan^*$.}
    \label{fig:enter-label}
\end{figure}

If both $\pth_A^*$ and $\pth_B^*$ have  discontinuities at $\lambda$,
then we reprameterize $\plan^*[\lambda^-,\lambda^+]$ as follows: first follow
$\plan^*[\lambda^-,\lambda^+)$, park $A$ at $p_A^-(\lambda)$ and 
%let $B$ play the role of $X$ and apply the above procedure but with the following caveat:
move $B$ along the ghost segment $p_B^-(\lambda)p_B^+(\lambda)$, park $B$ at $p^+_B(\lambda)$ 
and move $A$ along the ghost segment $p_A^-(\lambda)p_A^+(\lambda)$, and finally follow $\plan^*(\lambda,\lambda^+]$. 

This completes the description of the re-parametrization of $\plan^*[\lambda^-,\lambda^+]$.
\lemref{ghost-correctness} in \secref{surgery-correct} proves the feasibility of 
$\plan^*[\lambda^-,\lambda^+]$ after the re-parametrization.

\def\ph{\varphi}
%-------------------------------------------------------------
\section{Proof of Correctness}\label{sec:surgery-correct}
%-------------------------------------------------------------
We now prove that the plan $\plan^*$ resulting from the surgery described
above preserves properties (P1)--(P5).
%We start with the easiest properties, (P3)--(P5), and then prove (P1) and~(P2).
%
%-------------------------------------------------------------
%\subparagraph{Progress (P3), Vertical alignment (P4), and Alternation (P5).}
%-------------------------------------------------------------
By construction, the resulting plan is decoupled and alternating, 
so (P5) is satisfied. Furthermore, the selected bad segment, $e$,
is pushed to a grid line and no new horizontal bad segments are 
introduced, so (P3) is satisfied. 
Ghost segments and zero-length segments at the breakpoints introduced by the reprarametrization procedure are the only new vertical segments in $\plan^*$.  The breakpoints are introduced at the intersection points of grid lines with the segments of $\plan$, so they are contained in the vertical grid lines or 
aligned with the vertical 
segments of $\plan$. The discussion during the surgery procedure implies that the ghost segments 
because of primary push operations also satisfy this property.
Finally, the  ghosts segments introduced on $\pth_B^*$ because of secondary 
pushes are on an endpoint of an unsafe interval, therefore by \lemref{unsafe-vertical-edge}, they 
also have the desired property.
Hence, the surgery preserves (P4) as well. In the rest of the section, we thus focus on proving 
feasibility (P1) and optimality (P2) of $\plan^*$.

%-------------------------------------------------------------
\subsection[Feasibility of $\pi^*$]{Feasibility of $\plan^*$}
\label{subsec:feasibility}
%-------------------------------------------------------------
We prove the feasibility of $\plan^*$ by showing that both $\pth_A^*$
and $\pth_B^*$ lie inside $\fre$ 
and that $\pth_A^*(\lambda)$ and $\pth_B^*(\lambda)$ are not in conflict for any $\lambda$.
The plan of robot~$\robA$ is only modified during the time 
interval~$[\lambda_1,\lambda_2]$, and the modifications are
such that $\pth^*_A[\lambda_1,\lambda_2]$, including the ghost segments added to $\pth_A$,
lie in $\R \subset \fre$.
Similarly, the modifications in $\pth_B$ (including the ghost segments) 
because of primary push operations (which happen in Case~III) lie in a 
rectangle $Q\subset \fre$.
%and $\pth_B$ does not conflict with $\pth_A$.
Thus, it remains 
To prove that the modifications in $\pth_B$
because of secondary push operations stay in the free space~$\fre$ and that 
$\pth^*_B$ does not conflict
with $\pth^*_A$ (after re-parameterizing the plan around ghost segments.) To this end, we first prove that our
secondary push operations resolve conflicts
during unsafe intervals before adding ghost segments.
%-------------------------------------------------------------
\begin{lemma}\label{lem:secondary-push-correctness}
If a secondary push is performed on $\pth_B[\nu_1,\nu_2]$, 
then the new subpath $\pth^*_B[\nu_1,\nu_2]$, including its ghost segments, lie in $\fre$. 
Furthermore, $\pth^*_A[\nu_1,\nu_2]$ and $\pth_B^*[\nu_1,\nu_2]$, before adding the 
ghost segments, are not in conflict.
\end{lemma}
%-------------------------------------------------------------
\begin{proof}
A secondary push is of the form $\push(B,(\nu_1,\nu_2),y^*)$, where
$y^*=y(\topR)+1$ or $y^*=y(\botR)-1$. Since $y(\botR)\leq \yco{\pth^*_A(\nu)}\leq y(\topR)$
for all $\nu\in (\nu_1,\nu_2)$, this immediately implies that 
$\pth^*_B[\nu_1,\nu_2]$ and $\pth^*_A[\nu_1,\nu_2]$ do not conflict before the ghost 
segments are added.

To prove that $\pth^*_B[\nu_1,\nu_2]\subset\fre$, we first observe that
$\pth_B(\nu)$ must be $y$-separated from $\pth_A(\nu)$, for all 
$\nu\in (\nu_1,\nu_2)$, by the definition of unsafe intervals. Suppose $\robB$ lies below $\robA$
during $[\nu_1,\nu_2]$---the case of $\robB$ lying above $\robA$ is symmetric---then
$\pth_B[\nu_1,\nu_2]\subset Z^-$.  Recall that every connected
component of $Z^-\cap \I(\R)$ is a rectangle.
Thus $\pth_B[\nu_1,\nu_2]$ and its ghost segments are contained in one such rectangle, and
the secondary push keeps the path inside this rectangle and thus inside $\fre$.
\end{proof}
%-------------------------------------------------------------
Note that  for any $\lambda$, if $\plan(\lambda)$ is $x$-separated 
then $\pth^*_A(\lambda)$ and $\pth^*_B(\lambda)$ are not in conflict either
since we do not change the $x$-coordinate of any point on the paths.
Any $\lambda\in[\lambda_1,\lambda_2]$ such that $\plan(\lambda)$
is not $x$-separated must lie in an unsafe interval. Hence, 
\lemref{secondary-push-correctness} implies that the configurations 
$\pth^*_A(\lambda)$ and $\pth^*_B(\lambda)$ are not in conflict after we have
applied all secondary pushes but before we added ghost segments. 
The following lemma proves the feasibility of $\plan^*[\lambda_1,\lambda_2]$ even after adding the 
ghost segments and re-parametrizing the plan.

%-------------------------------------------------------------
\begin{lemma}
	\label{lem:ghost-correctness}
	Suppose $\plan^*$ has a discontinuity at time $\lambda$ before adding the ghost segments
	such that $\plan(\lambda)$ lies on a segment $\plan[\lambda^-,\lambda^+]$. Then the
    re-parametrization of $\plan^*$ in the interval $[\lambda^-, \lambda^+]$
	ensures that $A$ and $B$ are not in conflict during the interval $[\lambda^-,\lambda^+]$.
\end{lemma}
%-------------------------------------------------------------
\begin{proof}
We distinguish three cases.
\begin{itemize}
\item 
	First, assume that  $\lambda \ne \{\lambda_1, \lambda_2\}$ and
	$\pth_B(\lambda) \not\in \partial\I(\rect)$.
    \lemref{ghost-segment-endpts}\,(\ref{ghostii}) then implies that
	$\plan(\lambda)$ is $y$-semi-separated. 
	Suppose $X$ is parked at time $\lambda$ and $Y$ is moving 
	along the vertical  segment $\pth_Y[\lambda^-,\lambda^+]$. Since $\plan(\lambda)$ is 
	$x$-separated, $A$ and $B$ are not in conflict during $\plan^*[\lambda^-,\lambda^+]$ after 
	re-parametrizing $\plan^*[\lambda^-,\lambda^+]$.
\item
	Next, suppose $\pth_B (\lambda) \in \bd\I(\rect)$. In this case, 
	$\lambda^-=\lambda^+=\lambda$.
	First, suppose $\pth_B^*$ has a discontinuity at $\lambda$ and a ghost segment 
	$g_B(\lambda)$ was added.
	If $\pth_B(\lambda) \in \topRsq\cup\botRsq$, then 
	$g_B(\lambda)$ lies outside $\Rsq$ by \lemref{ghost-segment-endpts}\,(\ref{ghostiii}).
	On the other hand, if $\pth_B(\lambda)$ lies on the left or right edge of $\Rsq$ then
	$g_B(\lambda)$ lies on $\partial\Rsq$. In either case, $g_B(\lambda)$ does not intersect
	the interior of $\Rsq$ while $A$
	lies inside $\rect$ during $[\lambda^-,\lambda^+]$.
	If $\pth^*_B$ does not have a discontinuity at $\lambda$, then it is parked at 
	$\partial\I(\rect)$ while $A$ traverses its ghost segment inside $\rect$.
	Hence, $A$ and $B$ are not in conflict in $\plan^*[\lambda^-,\lambda^+]$ after 
	its re-parametrization.
\item
	Finally, suppose $\lambda=\lambda_1$. (The case of $\lambda=\lambda_2$ is symmetric.)
	Then $B$ is parked at $\pth_B(\lambda_1)$ in $\plan$, $A$ is 
	moving along a vertical segment, and $\lambda^-=\lambda^+=\lambda$ since we 
	introduced a breakpoint at $\pth_A(\lambda)$. 
	If $\plan(\lambda_1)$ is $x$-separated then
	moving $\robA$ and/or $\robB$ along their ghost segments will not cause a conflict.
	So assume $\plan(\lambda_1)$ is not $x$-separated. Then $\plan(\lambda_1)$ is 
	$y$-separated and $\lambda_1$ is the endpoint of an unsafe interval~$I$ that is not a 
	swap interval. 
    \begin{claim} \label{claim:B-in-top-or-bot}
	$\pth_B(\lambda_1)\in\topRsq\cup\botRsq$ and $\pth_B^*$ does not have a discontinuity 
	at $\lambda_1$.
    \end{claim}
    \begin{claimproof}
	First, note that a discontinuity at $\pth_B^*(\lambda_1)$ in this case 
	can happen only because of a secondary push. (Indeed, a discontinuity in $\pth_B^*$ because of a primary 
    push, which may happen in Case~III, always occurs at the endpoint of a swap interval.)
    Hence, by \lemref{ghost-segment-endpts}, it suffices to show that $\pth_B(\lambda)\in\topRsq\cup\botRsq$.
	%which implies that  $A$ is also being pushed at $\lambda_1$. A case analysis shows that $\pth_B(\lambda_1) \in \topRsq\cup\botRsq$ in this case (\pka{This requires an argument}) but 
	In Case~I, a secondary push on $B$ happens only if $B$ lies above $A$. Since 
	$\pth_A(\lambda_1)\in\topR$, we have $\pth_B(\lambda_1)\in\topRsq$, as claimed.
	In Case~III, by \lemref{lambda2}, $B$ lies below $A$ at time~$\lambda_1$. 
	Since $\pth_A(\lambda_1)\in\botR$, we have $\pth_B(\lambda_1)\in\botRsq$, as claimed. 
	Finally,  in Case~II, $B$ lying above $A$ at $\lambda_1$ would imply that the value of 
	$\bar\lambda$ is greater than the right endpoint of\footnote{Recall that $I_\lambda$
    is the tiny interval we add when $\lambda^-=\lambda^+=\lambda$; see page~\pageref{page:I_lambda}.} $I_\lambda$ 
	and that a secondary push is not performed for $I_\lambda$ because
	$\pth_A^*[\lambda_,\bar\lambda] \subset \botR$. But then neither $\pth^*_A$ nor 
	$\pth^*_B$ has a discontinuity at $\lambda_1$, a contradiction. Hence,  $\robB$ lies 
	below $\robA$ at $\lambda_1$. 
	Since $\pth_A(\lambda_1)\in \botR$, we have $\pth_B(\lambda_1)\in\botRsq$, as claimed.
     \end{claimproof}

	If $\plan(\lambda_1)$ is not $x$-separated and only $\pth_A^*(\lambda_1)$ has 
	a discontinuity then we are in Case~II or~III of the surgery.
	Together with Claim~\ref{claim:B-in-top-or-bot} this implies that $\pth_B(\lambda_1)\in\botRsq$, 
    so our conflict-resolving procedure parks $B$ at $\pth_B(\lambda_1)$ while $A$ moves along the ghost segment
	$p_A^-(\lambda_1)p_A^+(\lambda_1)$, ensuring that $A$ and $B$ are not in conflict 
	in $\plan^*[\lambda^-,\lambda^+]$ after its re-parametrization.
	This completes the proof of the lemma.
\end{itemize}
\end{proof}

%Moreover, if \lemref{secondary-push-correctness} therefore has the following corollary, which, together with \lemref{ghost-segments}, implies that the ghost segments do not cause collisions.
%-------------------------------------------------------------
%\begin{corollary}\label{cor:feasibility-general}
%After applying all secondary pushes and before adding ghost segments, there are no collisions between $\pth^*_A(\lambda)$ and $\pth^*_B(\lambda)$ for any $\lambda\in[\lambda_1,\lambda_2]$.
%\end{corollary} 

%-------------------------------------------------------------
We conclude that $\plan^*[\lambda_1,\lambda_2]$ is feasible, which
proves the feasibility of the plan~$\plan^*$ for the surgery in Cases~I, II, and III(b) 
because in those cases
the plan is modified only during~$[\lambda_1,\lambda_2]$. This also \emph{almost} implies 
the feasibility for Cases~III(a) and III(c), except that the surgery may modify $\plan$ beyond the 
interval $[\lambda_1,\lambda_2]$. In particular, 
it is possible that $\nu_0 < \lambda_1$ in Case~III. Then, when we push 
$\pth_B(\nu_0, \mu_1)$ to $\botQ$ in Case III(a), or $\pth_B(\rho_0,\rho_1)$ to $\topin{Z^-}$ 
in Case~III(c), it might happen that $\pth_B^*(\nu)$ conflicts with $\pth_A(\nu)$ 
for some $\nu \in (\nu_0, \lambda_1)$. 
We will argue that, in fact, this cannot happen: if such a $\nu$ were to exist
then we can shortcut~$\plan$, contradicting the optimality of~$\plan$, 
or show that $\plan[\nu_0,\lambda_1]$ can be replaced by another plan $\pi'$ 
such that we remain in the case III  but 
the surgery in the modified plan does not cause a conflict. 

We begin with a few definitions and some properties of~$\plan$.
\medskip

In the following, suppose $\plan^*$ has a conflict at some time in $[\nu_0,\lambda_1]$
caused by the surgery in Case~III(a) or~III(c). Let 
%-------------------------------------------------------------
\begin{equation}
	\label{eq:first-conflict}
\fconflict = \inf \{\nu \in [\nu_0,\lambda_1) : \text{$\plan^*(\nu)$ has a conflict} \}.
\end{equation}
%-------------------------------------------------------------
From now on we assume $\fconflict$ exists, otherwise there is nothing left to show.
Let $Q$ be the rectangle of $\horzF$
defined in Case~III of the surgery. We now prove a few properties of~$\plan$.
%-------------------------------------------------------------
\begin{figure}
		% \vspace{1in}
        \centering
        \includegraphics{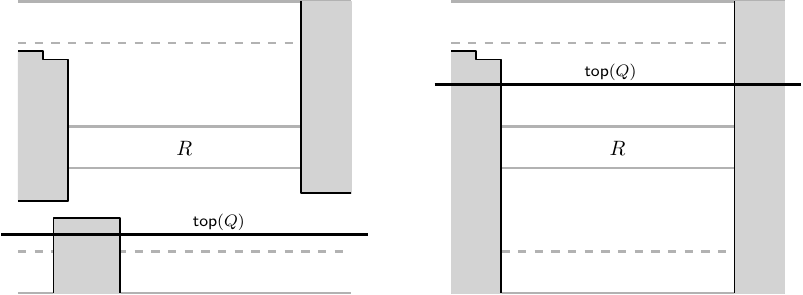}
		\caption{(left) $Q$ lies below $\rect$; (right) $\topQ$ lies above $\topR$.}
		\label{fig:QbelowR}
\end{figure}
%-------------------------------------------------------------
\begin{lemma}
	\label{lem:QbelowR}
		Consider Case~III of the surgery. Then the rectangle $Q$ lies 
		below $\botR$, that is, $y(\topQ) \leq y(\botR)$.
\end{lemma}
%-------------------------------------------------------------
\begin{proof}
	No horizontal grid line can intersect the interior of $\rect$.
    Hence, $y(\topQ) > y(\botR)$ would imply $y(\topQ) \geq y(\topR)$. 
	Furthermore, the argument in the Case~III of the surgery implies that 
	$y(\botQ)\leq y(\botin{Z^-})=y(\botRsq)$. Since no vertex of $\fre$ can lie in the range $(y(\botQ),y(\topQ))$,
    no vertex of $\fre$ has 
	$y$-coordinate in the (open) interval $(y(\botRsq), y(\topR))$. 
    %%  ARGUMENT  %%
    %% Indeed, $\pi_B(\mu_0,\mu_1)\subset Q' \subset Z^-$. Moreover, a vertex of $\fre$ with 
    %% $y$-coordinate in $(y(\botRsq), y(\topR))$ must be above $Z^-$ (since there are no vertices
    %% with $y$-coordinates in the $y$-range of $Z^-$) and below $R$, which contradicts that $y(top(Q)) > y(top(R))$.
    Thus, the bottom endpoints of the vertical 
	edges of $\fre$ that contain the left and right edges of $\rect$ lie on 
	or below $\botRsq$, and $\I(\rect)$ does not intersect the bottom corner squares or 
	$\Rsq$ (see \figref{QbelowR}(left)). Hence,
	$\I(\rect)$ does not admit any blocked pair, and therefore by \lemref{main-properties}, 
	$\plan$ has at most one swap interval, implying that Case~III does not occur, 
	a contradiction. We conclude that $y(\topQ) \leq y(\botR)$.
\end{proof}
%-------------------------------------------------------------------
\begin{lemma} 
	\label{lem:AaboveB}
	Consider Case~III(c). If $\pth_A(\fconflict) \not\in\I(\rect)\cup Q$
	then $\yco{\pth^*_A(\fconflict)} \geq \yco{\pth^*_B(\fconflict)}$.
\end{lemma}
%-------------------------------------------------------------------
\begin{proof} 
	For the sake of contradiction, suppose 
	$\pth^*_A(\fconflict)\notin \I(R)\cup Q$ and 
	$\pth^*_A(\fconflict)_y < \pth^*_B(\fconflict)_y$. Recall 
	that $\pth^*_A(\fconflict) = \pth_A(\fconflict)$ and that $\pth_B[\rho_0,\rho_1]$ is pushed to 
	$\topin{Z^-}$ in Case~III(c). Since $\pth_A$ conflicts with 
	$\pth_B^*$ immediately after $\fconflict$, but not with $\pth_B$, we have 
	$\pth_A(\fconflict)_y \geq \pth_B(\fconflict)_y + 1 > y(\botQ)$.
	Since 
	$$y(\botQ) < \yco{\pth_A(\fconflict)}=\yco{\pth_A^*(\fconflict)}<\yco{\pth_B^*(\fconflict)} 
			=y(\topin{Z^-})\leq y(\topQ)$$ 
	and $\pth_A(\fconflict)\notin Q$, there is a vertical edge of $\freesp$ in-between 
	$\pth_A(\fconflict)$ and $Q$, so the $x$-distance from $\pth_A(\fconflict)$ to any point 
	in $Q$ is strictly greater than $1$. 
	This contradicts that $\pth^*_A(\fconflict)+\square$ touches $\pth_B(\fconflict)+\square$ (namely, $|\xco{\pth_A^*(\fconflict)} - \xco{\pth_B^*(\fconflict)}| \leq 1$), for $\pth_B^*(\fconflict)\in Q$.
\end{proof}
%-------------------------------------------------------------------
%We first prove some properties that are a consequence of such a shortcut in Lemma~\ref{lem:case-III-shortcut-below}. We will later uses these properties to show the feasibility of Case~III(a) and Case~III(c).
%-------------------------------------------------------------------
	\begin{figure}
		\centering
		\includegraphics[page=2]{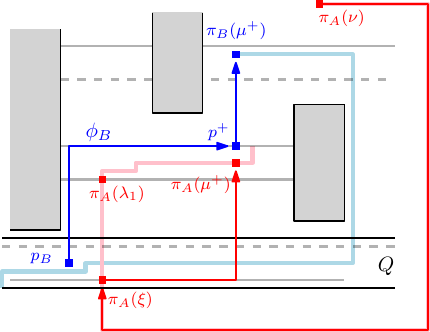}
		\caption{Illustration of the proof of \lemref{case-III-shortcut-below}.
        The original paths are shown in pink and light blue. Note that the figures are not drawn to scale; the horizontal axis is compressed for visualization purposes.
        % (i) In Case 1 of the proof, $\pth_A$ enters $\Rxy^+$ at time $\xi$, where $\nu < \xi < \lambda_1$, and we know that $\pth_B(\xi) \in Q$.
        % Moves~M1--M2 for this case are shown as solid blue and red paths.
		% (ii) In Case 2 of the proof,
        $\pth_A$ enters $\Rxy^-$ for the first time at time $\xi$. Moves~M1--M4 for this case are shown as solid blue and red paths.
        }
		\label{fig:short-corridor-case-III-a-overview}
	\end{figure}
%-------------------------------------------------------------------
Recall that $[\mu_0,\mu_1]$ and $[\mu_2,\mu_3]$ are the first and second swap interval,
respectively, that we have in Case~III. Define $\mu^+ \in [\mu_2,\mu_3]$ to be a time such that 
$\pth_A(\mu^+)_x = \pth_B(\mu^+)_x$; $\mu^+$ exists because the $x$-order of $\robA$ and $\robB$
changes during a swap interval. 
Define $p^+:= (\pth_B(\mu^+)_x, y(\topR)$. 
%-------------------------------------------------------------------
\begin{lemma}\label{lem:case-III-shortcut-below}
	Consider Case~III with $\nu_0 < \lambda_1$. 
	Let $p_B$ be a point in $Q$, and let $\nu\in [\nu_0,\lambda_1]$.
	If there exists an $xy$-monotone $p_Bp^+$-path in $\freesp$ that does not conflict with $\pth_A(\nu)$,
	then there exists a feasible $((\pth_A(\nu), p_B),\plan(\mu^+))$-plan $\plan'$ 
    such that $\|\pth'_A\| \leq \|\pth_A[\nu,\mu^+]\|$ and $\pth'_B$ is $xy$-monotone. 
\end{lemma}
\begin{proof}
	Let $\phi_B$ be an $xy$-monotone $p_Bp^+$-path in $\freesp$ that does not conflict with $\pth_A(\nu)$.
	Let $\xi := \min\{\xi'\in [\nu,\lambda_1]: \pth_A(\xi)\in \Rxy^-\}$; such a time exists 
	because $\robA$ enters $\rect$ through~$\botR$. 
	%since $\pth_A[\nu,\lambda_1]$ does not visit $\Rxy^+$ (\lemref{before-lambda0-not-other-side}, note that $\gamma(R)= \I(R)$ by \lemref{main-properties}(ii)) and $\pth_A(\lambda_1)\in \botR$.
	Note that $\pth_A(\xi')\notin \Rxy^-$ for all $\xi'\in [\nu,\xi)$.
	We define the $((\pth_A(\nu), p_B),\plan(\mu^+))$-plan $\plan'$
    to consist of the following four moves:
	\begin{quotation} \noindent \vspace*{-3mm}
		\begin{enumerate}[M1.]
			\item Move $B$ from $p_B$ to $p^+$ along $\phi_B$.
			\item  Move $A$ along $\pth_A[\nu,\xi]$.
			\item  Move $B$ on a vertical segment from $p^+$ to $\pth_B(\mu^+)$.
			\item Move $A$ along an $xy$-monotone path from $\pth_A(\xi)$ to $\pth_A(\mu^+) \in R$.
		\end{enumerate}
	\end{quotation}
	
	See Figure~\ref{fig:short-corridor-case-III-a-overview}. 
	Move~M4 exists because $\pth_A(\xi)\in\I(R)$---this is true because 
	by \lemref{main-properties}(ii) there are no tiny components 
	in Case~III---so we can apply \lemref{xy-shortest path}.
	By construction, moves M1--M4 are all in the free space.
	We now show that these moves do not cause conflicts. 
    Robot~$B$ does not conflict with $A$ during move M1 by assumption.
	By \lemref{before-lambda0-not-other-side}, $\pth_A[\nu_0,\lambda_1)$ does not enter 
	$\rect\cup\Rxy^+$, and by the definition of $\xi$ we know that 
    $\pth_A[\nu,\xi)$ does not enter $\Rxy^-$, 
	Hence, $\pth_A(\nu,\xi)$ does not enter the interior of $\Rsq$.
	Since $p^+ \in \topR$, we have $p^++2\Box \subset \Rsq$ and thus
	M2 does not conflict with $B$.
	The vertical move M3 of $B$ from $p^+$ to $\pth_B(\mu^+)$ is feasible
	because $B$ moves upward, while $A$ is parked at $\pth_A(\xi)$ without being in conflict with
	$p^+$ and we have $\pth_A(\xi)_y \leq y(\botR)\leq p^+_y$; thus, $\robB$ is moving away 
	from $\robA$.
	Finally, move~M4 does not conflict with $\pth_B(\mu^+)$ 
	since $\pth_A(\xi)_y\leq \pth_B(\mu^+)_y - 1$ and $\pth_A(\mu^+)_y \leq \pth_B(\mu^+)_y - 1$.
	
	Hence, the plan $\plan'$ is feasible. Note that $\phi_B$ is $xy$-monotone and $B$ moves 
	in the same vertical direction as $\phi_B$ during~M3---the latter is true bacause
	$p_B\in Q$ and by \lemref{QbelowR}, $Q$ lies below $\rect$. Hence, 
	$\pth'_B$ is $xy$-monotone. During M2, $A$ moves along $\pth_A[\nu,\xi]$, 
	while it moves $xy$-monotonically during M4. So, we indeed have 
	$\|\pth'_A\|\leq \|\pth_A[\nu,\mu^+]\|$.
\end{proof}
%-------------------------------------------------------------------
The above lemma will be sufficient to shortcut $\plan$ in Case~III(a), 
but we need the following variant of the lemma for Case~III(c).
%-------------------------------------------------------------------
	\begin{figure}
		\centering
		\includegraphics{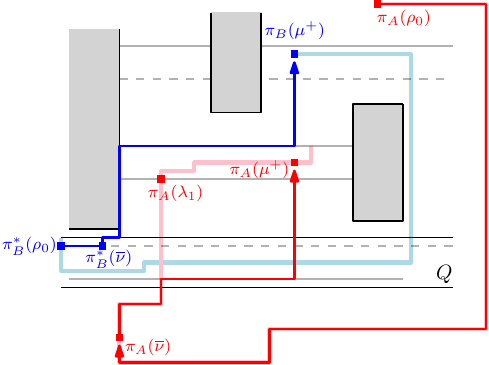}
		\caption{Illustration of the proof of \lemref{apply-shortcut-to-pi*} depicting the hypothetical situation in Case~III(c) where pushing $\pth_B(\nu)$ up to $\pth^*_B(\nu)$ causes a conflict (here, $A$ collides with $B$ when moving upwards). The paths of $\plan''$ are shown in solid blue and red, and the original paths are shown in pink and light blue. Note that the horizontal axis is compressed for visualization purposes.}
		\label{fig:apply-shortcut-to-pi*}
	\end{figure}
    
\begin{lemma}
	\label{lem:apply-shortcut-to-pi*}
	Consider Case~III(c). 
	If there exists an $xy$-monotone $\pth^*_B(\fconflict)p^+$-path in $\freesp$ that does not 
	conflict with $\pth_A(\fconflict)$, then $\plan$ is not optimal.
\end{lemma}
%-------------------------------------------------------------------
\begin{proof}
	Applying \lemref{case-III-shortcut-below} with $p_B:=\pth^*_B(\fconflict)$ and 
	$\nu:= \fconflict$, we obtain
	a $\plan^*(\fconflict)\plan(\mu^+)$-plan~$\plan'$ that is feasible and such that 
	$\|\pth'_A\| \leq \|\pth_A[\fconflict,\mu^+]\|$ and $\pth'_B$ is $xy$-monotone.
	% $\|\pth'_B\| = \|\pth_B^*(\fconflict)-\pth_B(\mu^+)\|_1$.
	% (Recall that $\pth_A^*(\fconflict) = \pth_A(\fconflict)$.) 
	
	Let $\plan'' = \plan^*[\rho_0,\fconflict] \circ \plan'$, which
	is a feasible $\plan(\rho_0)\plan(\mu^+)$-plan with 
    \[
    \|\pth''_A[\rho_0,\mu^+]\| = \|\pth_A[\rho_0,\fconflict]\| + \|\pth'_A \| \leq \|\pth_A[\rho_0,\mu^+]\|
    \mbox{\ \ \ and \ \ \ }
	\|\pth''_B[\rho_0,\fconflict]\| \leq \|\pth_B[\rho_0,\fconflict]\|. 
    \] 
See \figref{apply-shortcut-to-pi*}. Note that $\pth''_B[\fconflict,\mu^+]$, which is the same as $\pi'_B[\fconflict,\mu^+]$, is $xy$-monotone.
    On the other hand, $\pth_B[\fconflict,\mu^+]$ is not $y$-monotone,
    because $\pth_B[\mu_1,\mu_2]$ is not $xy$-monotone by \lemref{main-properties}(ii), 
    and $[\mu_1,\mu_2]\subset [\fconflict,\mu^+]$.
    Thus, $\|\pth''_B[\fconflict,\mu^+]\| < \|\pth_B[\fconflict,\mu^+]\|$. This means that
	$\plan''$ is a shortcut, contradicting the optimality of~$\plan$.
\end{proof}
%-------------------------------------------------------------------
We are now ready to prove the feasibility of $\plan^*$ in Cases~III(a) and~(c).
%-------------------------------------------------------------------

	\begin{figure}
		\centering
		\includegraphics{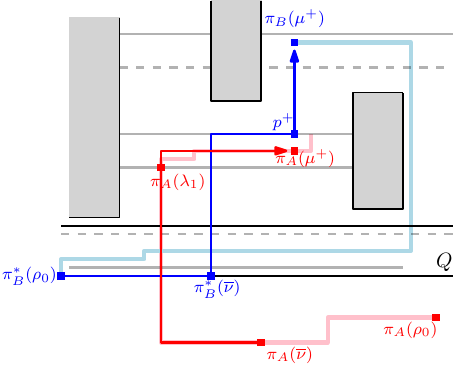}
		\caption{Illustration of the proof of \lemref{apply-shortcut-to-pi*} depicting the hypothetical situation in Case~III(a) where pushing $\pth_B(\nu)$ down to $\pth^*_B(\nu)$ causes a conflict (here, $B$ collides with $A$ when moving to the right). The paths of $\plan'$ are shown in solid blue and red, and the original paths are shown in pink and light blue. Note that the horizontal axis is compressed for visualization purposes.}
		\label{fig:feasible-case-IIIa}
	\end{figure}

\begin{lemma}\label{lem:feasible-case-IIIa}
	%Suppose we are in Case~III(a) and $\nu_0 < \lambda_1$. 
	Consider Case~III(a) with $\nu_0 < \lambda_1$. 
	Then the modified plan $\plan^*$ does not have a conflict for any 
	$\nu \in [\nu_0, \lambda_1]$.
\end{lemma}
%-------------------------------------------------------------------
\begin{proof}
	%Since $Q\subset \fre$, we trivially have $\pth_B^*[\nu_0, \lambda_1]\subset \fre$.
	Suppose for a contradiction that the plan $\plan^*$ obtained in Case~III(a) 
	has a conflict at some time $\nu\in [\nu_0,\lambda_1]$. Let $\fconflict$ be the first 
	point of conflict in $\plan^*$, as defined in~\eqref{first-conflict}.
	
	Since $\pth_B[\nu_0, \mu_1]$ is pushed down to $\botQ$,
	a conflict can arise immediately after $\fconflict$
	only if $\yco{\pth_B(\fconflict)} \geq \yco{\pth_A(\fconflict)}+1$.  Since 
	$\pth_B(\fconflict)\in Q$, $Q$ lies below $\rect$ by \lemref{QbelowR}, 
	and $\robA$ lies below $\robB$ at $\plan(\fconflict)$, 
	there is a $xy$-monotone path from $\pth_B(\fconflict)$ to $p^+$ that
	does not conflict with $\pth_A(\fconflict)$. See \figref{feasible-case-IIIa}.
	Hence, by \lemref{case-III-shortcut-below}, there exists 
	a $(\plan(\nu),\plan(\mu^+))$-plan $\plan'$ such that 
	$\|\pth'_A\| \leq \|\pth_A[\nu,\mu^+]\|$ and 
	$\|\pth'_B\|=\|\pth_B(\fconflict)-\pth_B(\mu^+)\|_1$. On the other hand,
	$\|\pth_B[\fconflict,\mu^+]\| > \|\pth_B(\fconflict)-\pth_B(\mu^+)\|_1$ since 
	$\pth_B[\fconflict,\mu^+]$ leaves $\I(\rect)$, by \lemref{main-properties}~(iii), 
	contradicting the optimality of $\plan$.
\end{proof}
%-------------------------------------------------------------------
\begin{lemma}\label{lem:feasible-case-III-c}
	Consider  Case~III(c) with $\nu_0 < \lambda_1$. 
	Then the modified plan $\plan^*$ does not conflict for any 
	$\nu \in [\nu_0, \lambda_1]$.
	%The modification in Case~III(c) yields a feasible plan.
\end{lemma}
%-------------------------------------------------------------------
\begin{proof}
	Let $\rho_0,\rho_1$ be as defined in the surgery of Case~III. Only $\plan[\rho_0,\rho_1]$ 
	is modified in Case~III(c), so the lemma follows immediately if $\rho_0 \geq \lambda_1$.
	Now assume that $\rho_0< \lambda_1$ and that there is a conflict in 
	$\plan^*[\rho_0,\lambda_1]$. Let $\fconflict$ be the first instance of conflict, as 
	defined in~\eqref{first-conflict}. 
	Note that $\yco{\pth_B^*(\fconflict)}=y(\topx{Z^-})=y(\topx{\rect})-1$.
	There are two cases.
	\medskip

	\begin{figure}
		\vspace{1in}
        \makebox[\textwidth][c]{\includegraphics[width=1.2\textwidth]{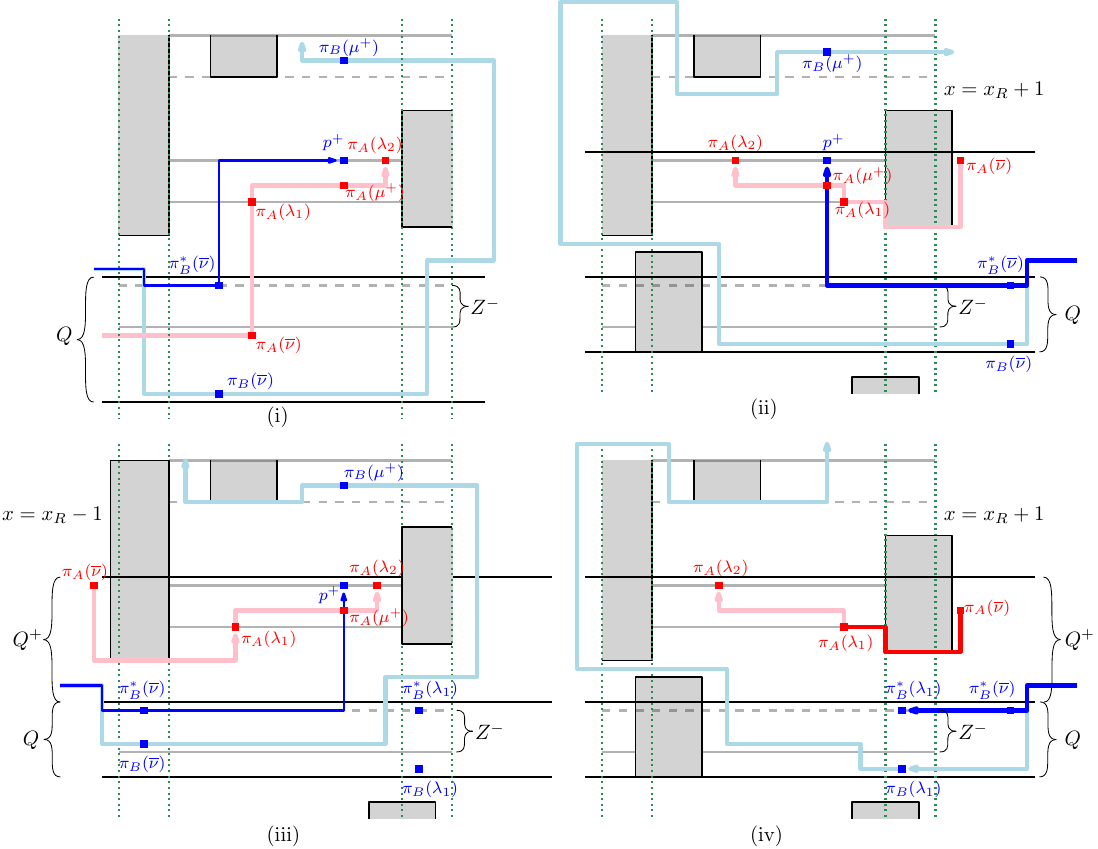}}
		\caption{Illustrations of the proof of \protect\lemref{feasible-case-III-c}. In all figures, the horizontal axis is compressed, and in (i), the height of $Q$ is not shown to scale (which must be at least three for $\plan(\fconflict),\plan^*(\fconflict)$ to be free configurations). (i) 
		$\pth_A^*(\fconflict)\in \I(\rect)\cup Q$;
        (ii) $\pth_A^*(\fconflict)\not\in \I(\rect)\cup Q$ and $\plan^*(\fconflict)$ is $x$-separated; 
		(iii) $\pth_A^*(\fconflict)\not\in \I(\rect)\cup Q$ and
		$\plan^*(\fconflict)$ is $y$-separated with $A$ left of $B$ and left of $\Rsq$;
        (iv) $\pth_A^*(\fconflict)\not\in \I(\rect)\cup Q$ and
		$\plan^*(\fconflict)$ is $y$-separated with $A$ left of $B$ and right of $\Rsq$. }
		\label{fig:feasibility-IIIc}
	\end{figure}
%---------------------------------------------------------------------------------
\noindent \textit{Case~1: $\pth_A^*(\fconflict)\in \I(\rect)\cup Q$.}
\smallskip \\
%---------------------------------------------------------------------------------
By \lemref{before-lambda1-x-sep-below},
$\plan^*(\fconflict)$ being $x$-separated implies that
$\plan(\nu')$ is $x$-separated for all $\nu'\in [\fconflict,\mu_0]$, which means 
	$\plan^*[\fconflict,\lambda_1]$ has no conflict, contradicting our assumption.
	So, $\plan^*(\fconflict)$ is not $x$-separated and thus $y$-separated. 
    Moreover, the two robots touch since they are about to be in conflict.

	If $\robA$ lies above $\robB$, then
	$\yco{\pth^*_A(\fconflict)} = \yco{\pth^*_B(\fconflict)}+1=y(\topx{\rect})$.
	Hence, $\pth^*_A(\fconflict)$ must lie on $\botin{\Rxy^+}$ because
	$\pth^*_A(\fconflict) \in \I(\rect) \cup Q$ and $Q$ lies below $\botR$ 
	by \lemref{QbelowR}. But this contradicts the assumption 
	that $\pth_A(\xi)$ is not in the closure of $\Rxy^+$ 
	(see \lemref{before-lambda0-not-other-side}) for all $\xi \in [\nu_0,\lambda_1]$. 
    On the other hand, if $A$ lies below $B$ in $\plan^*(\fconflict)$, 
    then the point $\pth_A(\fconflict)$ does not conflict 
	with any $xy$-monotone $\pth_B(\fconflict)p^+$-path (and there is such a path).
    See \figref{feasibility-IIIc}(i).
	So, by \lemref{apply-shortcut-to-pi*}, we obtain a contradiction to $\plan$ being optimal. 
\medskip

%---------------------------------------------------------------------------------
\noindent \textit{Case~2: $\pth_A^*(\fconflict)\notin \I(R)\cup Q$.}
\smallskip \\
%---------------------------------------------------------------------------------
Using \lemref{AaboveB} and the fact that $\plan$ conflicts immediately 
after $\fconflict$, we have $|\pth_A^*(\fconflict)_y - \pth_B^*(\fconflict)_y \leq 1$, and hence
	\[
	y(\topR)-1 = \yco{\pth_B^*(\fconflict)} \leq\yco{\pth_A^*(\fconflict)} 
	\leq \yco{\pth_B^*(\fconflict)}+ 1 = y(\topR).
	\]
We now have two subcases.
\begin{itemize}
\item 
First, suppose that $\plan^*(\fconflict)$ is $y$-separated. 
In this case, $\yco{\pth_A^*(\fconflict)} = y(\topR)$. 
But $\pth_A^*(\fconflict) \not\in \I(\rect)$, so we conclude that 
$\xco{\pth_A^*(\fconflict)} \not\in \Ixp$ and thus $\pth_A^*(\fconflict)$ 
is $x$-separated with every point in $\rect$.
% \mdb{Recall that $\mu^+ \in [\mu_2,\mu_3]$ is a time such that $\pth_A(\mu^+)_x = \pth_B(\mu^+)_x$.}
The above means that the following $xy$-monotone
$\pth_B^*(\fconflict)p^+$-path %\mdb{Replaced $\pth_B(\mu^+)$ by $p^+$}
does not conflict with $\pth_A(\fconflict)$: 
\begin{itemize}
\item 
    First, move $B$ horizontally within $Q$ from $\pth_B^*(\fconflict)$ to 
    the $x$-range of $\rect$, namely to vertical line containing the left (resp.\ right) 
    edge of $\rect$ if $\xco{\pth_B^*(\fconflict)} < x_\rect^-$ 
    (resp.\ $\xco{\pth_B^*(\fconflict)} > x_\rect^-$). Note that this move is empty
    if $\xco{\pth_B^*(\fconflict)}\in \Ixp$.
\item Then move $B$ vertically to $\botR$.
\item Finally, move $B$ to $p^+$ along an L-shaped path.
\end{itemize}
See \figref{feasibility-IIIc}(ii). The first move is without conflict because $\robB$ remains $y$-separated 
with $\pth_A^*(\fconflict)$, the second and third moves are without conflict 
because $\robB$ lies in $\rect$ during these moves and 
is thus $x$-separated from $\pth^*_A(\fconflict)$.
So, by \lemref{apply-shortcut-to-pi*}, $\plan$ is not optimal, a contradiction.
\item 
Next, suppose $\plan^*(\fconflict)$ is not $y$-separated. Hence, $\plan^*(\fconflict)$
is $x$-separated. Assume wlog that 
$\xco{\pth_A^*(\fconflict)} = \xco{\pth_B^*(\fconflict)}-1$, 
so that $\robA$ is to the left of 
$\robB$; the case of $\robA$ being to the right of $\robB$ is symmetric. 
Since $\yco{\pth^*_A(\fconflict)} \in \Iyp$ and  $\pth_A^*(\fconflict)\notin \I(R)\cup Q$,
we conclude that  $\xco{\pth_A^*(\fconflict)} \not \in \Ixp$.
% \mdb{I do not immediately see this.}
% \mst{If $\yco{\pth^*_A(\fconflict)} \in \Iyp$ and $\xco{\pth_A^*(\fconflict)} \in \Ixp$, then $\pth^*_A(\fconflict)\in \I(R)$?}
% \mdb{Right.}

If $\xco{\pth_A^*(\fconflict)} < x_R^- - 1$, then 
$\xco{\pth_B^*(\fconflict)} = \xco{\pth_A^*(\fconflict)} + 1 < x_R^-$.
Furthermore, $p_x^+\geq x_R^-$. This implies that, as in the previous case, there exists an $xy$-monotone 
$\pth^*_B(\fconflict)p^+$-path in $\fre$ that does not conflict with~$\pth^*_A(\fconflict)$:
move $B$ horizontally to the $x$-range of $R$,
then vertically to $\botR$, and then with an L-shaped path to $p^+$.
(Here the first move is without conflict because $\pth_A(\fconflict)_x +1 = \pth_B(\fconflict)_x < p^+_x$.)
See \figref{feasibility-IIIc}(iii). Therefore, by \lemref{apply-shortcut-to-pi*}, $\plan$ is not optimal.

Otherwise, $\xco{\pth_A^*(\fconflict)} > x_R^+ + 1$. 
Let $Q^+$ be the rectangle of height~$1$ that has $\topQ$ as its bottom edge.
Recall that we are in the situation where $\pth_A(\fconflict)\not\in Q$,
and $\pth_B(\fconflict)\in Q$,  and $\plan(\fconflict)$ is not $y$-separated.
See \figref{feasibility-IIIc}(iv).
Since $\pth_A(\fconflict)\not\in Q$ and
$\pth_A(\fconflict) \not\in \I(\R)$---the latter is true because $\xco{\pth_A(\fconflict)} =\xco{\pth_A^*(\fconflict)} > x_R^+ + 1$---we know that $\pi^*_A(\fconflict)$ lies above $\pi^*_B(\fconflict)$ by \lemref{AaboveB}.
We thus have 
\[
y(\topQ) < \yco{\pth_A(\fconflict)} \leq y(\topQ)+1.
\]
We also have
\[
\pth_B(\mu_0)_x \leq x^+_R + 1 < \pth_A(\fconflict)_x < \pth_B(\fconflict)_x.
\]
Since $[\pth_B(\mu_0)_x, \pth_B(\fconflict)_x]\times [y(\topQ), y(\topQ)+1]\subset Q^+$,
it follows that $\pth_A(\fconflict) \in Q^+$. 
So, by Lemma~\ref{lem:before-lambda1-x-sep-below}, the fact that $\plan^*(\fconflict)$ 
is $x$-separated implies that $\plan(\nu')$ is $x$-separated for 
all $\nu'\in [\fconflict,\mu_0]$. Hence, $\plan^*[\fconflict,\lambda_1]$ has no conflict, contradicting our assumption.
\end{itemize}
\end{proof}
%-------------------------------------------------------------------
Putting everything from \subsecref{feasibility} together, we obtain the following lemma.
%-------------------------------------------------------------------
\begin{lemma}
	\label{lem:feasible}
	The modified plan $\plan^*$ is feasible.
\end{lemma}

%-------------------------------------------------------------
\subsection{Optimality of $\plan^*$}
\label{subsec:optimality}
%-------------------------------------------------------------
We will now prove the optimality of $\plan^*$. Since the surgery does not 
change the $x$-coordinate for any $\lambda$, the cost of $\plan^*$ may only
increase because of the ghost segments added to remove discontinuities created 
by the push operations. Roughly speaking, we show that we are always in
one of the following three cases.
\begin{itemize}
\item $\push(X, I, y)$ does not create a discontinuity at either endpoint of $\plan^*_X[I]$
      and, hence, no ghost segment is added; or
\item $\push(X, I, y)$ does not increase the cost of $\plan^*[I]$,
      that is, $\|\plan^*[I]\| \leq \|\plan[I]\|$ where the lengths of the ghost segments are included; or
\item $\push(X, I, y)$ increases the cost of $\plan^*[I]$, but $\|\plan^*[I]\|-\|\plan[I]\|$
      (which is the cost increase) is compensated by a cost decrease resulting from another push operation. 
\end{itemize}
We need the following basic properties of the push operation.
%-------------------------------------------------------------
\begin{lemma}\label{lem:surgery-props}
Consider the operation $\push(X, [\nu_1, \nu_2], y^*)$. If there is a time 
	$\nu \in [\nu_1, \nu_2]$ such that $\yco{\pth_X(\nu)} = y^*$, then the total
length of $\pth_X^*[\nu_1, \nu_2]$, including those of the ghost segments 
$g(\nu_1) := p_X^-(\nu_1)p^+_X(\nu_1)$ and $g(\nu_2) := p^-_X(\nu_2)p^+_X(\nu_2)$,
is at most $\cost{\pth_X[\nu_1, \nu_2]}$.
\end{lemma}
%-------------------------------------------------------------
\begin{proof}
The horizontal distance traversed by $\pth_X[\nu_1, \nu_2]$ and by 
$\pth_X^*[\nu_1, \nu_2]$ is the same,
since a \push does not modify the $x$-coordinates of points on the path.
The vertical distance traversed by $\pth_X^*[\nu_1, \nu_2]$ 
equals the total length of the ghost segments, which is 
$|\yco{\pth_X(\nu_1)} - y^*|+|\yco{\pth_X(\nu_2)} - y^*|$. Since $\pth_X[\nu_1,\nu_2]$
visits a point with $y$-coordinate~$y^*$, the vertical distance traversed
by $\pth_X[\nu_1, \nu_2]$ must be at least this amount.
% $|y(\pth_X(\nu_1)) - y^*|+|y(\pth_X(\nu_2)) - y^*|$, thus proving the lemma.
\end{proof}
%-------------------------------------------------------------
\lemref{surgery-props} immediately implies that the primary push operations on $\robA$ 
on $\robB$---the latter happens only in Case~III---do not increase the cost of the plan. 
(In fact, as we will see below, the cost sometimes goes down, for example in Cases~I and~III(c), 
which can compensate for the cost increase in secondary pushes.) 
We thus focus on the secondary push operations, which are performed only on $\robB$ 
during an unsafe interval.\footnote{For simplicity, if a push operation is performed 
during an interval $I$, we use $\plan^*[I]$ to denote the subplan that includes 
the parametrization of  the motion along the 
ghost segments at the endpoints of $\plan[I]$.}
%-------------------------------------------------------------
\begin{lemma}
	\label{lem:non-x-sep-push-cost}
	Let $I=[\nu_1,\nu_2]$ be an unsafe interval that is not $x$-separated. If a secondary push
	$\push(B,I,y^*)$, for some $y^*$, is performed, then $\|\plan^*[I]\| \leq \|\plan[I]\|$.
\end{lemma}
%-------------------------------------------------------------
\begin{proof}
	As discussed above, $\|\pth_A^*[I]\| \leq \|\pth_A[I]\|$, so we need to argue that 
	$\|\pth^*_B[I]| \leq \|\pth_B[I]\|$.
	Since $I$ is not an $x$-separated unsafe interval, one of the endpoints of $\plan[I]$, 
	which we assume wlog to be $\plan(\nu_1)$, is not $x$-separated. Then $\plan(\nu_1)$ is $y$-separated and 
	either $\pth_B(\nu_1)\in\topRsq\cup\botRsq$ or $\nu_1=\lambda_1$. 

	If $\pth_B(\nu_1)\in\topRsq$ then $y^*=y(\topRsq)$,  which by \lemref{surgery-props}
	implies $\|\pth^*_B[I]| \leq \|\pth_B[I]\|$. The argument for $\pth_B(\nu_1)\in\botRsq$)
    is similar.

	  It remains to consider the case $\nu_1=\lambda_1$. 
	In case~I, we have $\pth_A(\lambda_1) \in \topR$.
	Since $\plan(\lambda_1)=\plan(\nu_1)$ is $y$-separated and a secondary push is performed during
	$I$, we have $\pth_B(\nu_1) \in \topRsq$ and the secondary push operation pushes 
	$\robB$ to $\topRsq$. In Case~III, \lemref{lambda2} implies that $\robA$ lies below $\robB$, 
	and since $\plan(\lambda_1)$ is $y$-separated, $\pth_B(\nu_1)\in\botRsq$ and the 
	secondary push operation pushes $\robB$ to $\botRsq$. Finally, in Case~II, $\robA$ is 
	pushed to $\botR$ during the interval $[\lambda_1,\bar\lambda]$, so a secondary 
	push on an unsafe interval starting at $\lambda_1$ implies that $\robB$ lies below $\robA$ in 
	$\plan(\lambda_1)$. Since $\pth_A(\lambda_1)\in\botR$ and $\plan(\lambda_1)$ 
	is $y$-separated, we have that $\pth_B(\nu_1) \in \botRsq$. Hence, again, 
	by \lemref{surgery-props}, $\|\pth^*_B[I]| \leq \|\pth_B[I]\|$ in all cases.
	This completes the proof of the lemma.
\end{proof}
%-------------------------------------------------------------
\lemref{non-x-sep-push-cost} implies that only a secondary push on an $x$-separated 
unsafe interval may increase the cost of the plan. By~\lemref{unsafe-is-swap}, any 
$x$-separated unsafe interval in $\plan$ is a swap interval.
Hence,it suffices to examine the secondary push operations 
performed on swap intervals, which we do case by case.
%-------------------------------------------------------------
\begin{lemma}\label{lem:snapping-case-I}
The modification in Case~I yields a plan~$\plan^*$ that is optimal.
\end{lemma}
\begin{proof}
	Let $I=[\nu_1,\nu_2]$ be a swap interval in $\plan[\lambda_1,\lambda_2]$ such that 
	a secondary push $\push(B,I, y^*)$ is performed. Since $\robA$ is pushed to 
	$\topR$ during $[\lambda_1,\lambda_2]$, $\robB$ lies above $A$ during $I$ and $y^*=y_R^++1$.
	Since there is at most one swap interval with $B$ above $A$ by
	Lemma~\ref{lem:limit-swap-intervals}, there
	is at most one swap interval $I$ for which a secondary push is performed.
	Pushing $B$ to $\topRsq$ during $I$ increases $\cost{\pth^*_B}$ by at
	most the total length of the ghost segments $g_B(\nu_1)$ and 
	$g_B(\nu_2)$, which is $2y_{\R}^+ + 2 - \yco{\pth_B(\nu_1)} - \yco{\pth_B(\nu_2)}$.
	Pushing $A$ to $\topR$ during $[\lambda_1,\lambda_2]$ decreases the length of its
	path by at least $2y_{\R}^+ - \yco{\pth_A(\nu_1)} - \yco{\pth_A(\nu_2)}$. 
	By definition, the configurations $\plan(\nu_i)$ are $y$-separated for $i = 1,2$,
	so $\yco{\pth_B(\nu_i)} - \yco{\pth_A(\nu_i)} \geq 1$.
	Hence,
    \begin{align*}
	    \cost{\plan} - \cost{\plan^*} &\geq \Big( 2y_{\R}^+ - \yco{\pth_A(\nu_1)} - 
		\yco{\pth_A(\nu_2)} \Big) - \Big( 2y_{\R}^+ + 2 - \yco{\pth_B(\nu_1)} - 
		\yco{\pth_B(\nu_2)} \Big) \\
	    &= \Big( \yco{\pth_B(\nu_1)} - \yco{\pth_A(\nu_1)} - 1 \Big) +  
		\Big( \yco{\pth_B(\nu_2)} - \yco{\pth_A(\nu_2)} - 1 \Big)\\
        &\geq 2\cdot (1 -1 )=0.
    \end{align*}
	We thus conclude that $\plan^*$ is optimal.
\end{proof}
%-------------------------------------------------------------------
\begin{lemma}\label{lem:snapping-case-II}
The modification in Case~II yields a plan~$\plan^*$ that is  optimal.
\end{lemma}
%-------------------------------------------------------------------
\begin{proof}
%	In Case~II, we only modify the paths during the interval~$[\lambda_1,\lambda_2]$.  By construction, $\cost{\pth_A^*[\lambda_1, \lambda_2]} \leq \cost{\pth_A[\lambda_1, \lambda_2]}$.  Hence, if there is no unsafe interval  then $\plan^*$ is still optimal, since $\pth^*_B[\lambda_1,\lambda_2]=\pth_B[\lambda_1,\lambda_2]$ in that case. Now assume there is an unsafe interval $[\nu_1, \nu_2]$. There are three cases to consider.
    %\medskip \\
    %-----------------------------------------------------------------
	%\emph{Case 1: $[\nu_1, \nu_2]$ is a swap interval.} 
    %-----------------------------------------------------------------
	Let $I=[\nu_1,\nu_2]$ be a swap interval.  We claim that $\pth_B[\nu_1, \nu_2]$ is not 
	modified. Indeed,
	if $[\nu_1, \nu_2]$ is a swap interval, then either Case~II(b) or Case~II(c) occurs. In Case~II(b), we have $[\nu_1, \nu_2] = [\mu_1, \mu_2]$
	and $\pth_B$ lies below $\pth_A$ during $[\nu_1, \nu_2]$. Since $\bar{\lambda} = \mu_1$ in this case, and $A$ is pushed to $\topR$ during
	the interval $[\bar{\lambda}, \lambda_2]$, $B$ does not conflict with $A$ during $[\nu_1, \nu_2]$ after the modification of
	$\pth_A$. In Case~II(c), there may be two swap intervals. If $[\nu_1, \nu_2]$ is the first swap interval, $[\mu_1,\mu_2]$, then $B$ lies above $A$
	during $[\nu_1, \nu_2]$. Since $\bar{\lambda} = \mu_2$ in this case, $\pth_A$ is pushed to $\botR$ during $[\lambda_1, \bar{\lambda}] \supseteq [\nu_1, \nu_2]$,
	and thus away from $B$, and $\pth_B$ is not modified during $[\nu_1, \nu_2]$. On the other hand, if $[\nu_1, \nu_2]$ is the second swap interval, then $\nu_1 > \mu_2 = \bar{\lambda}$
	and $B$ lies below $A$ by \lemref{unsafe-is-swap}. Since $\pth_A$ is pushed to $\topR$ during $[\bar{\lambda}, \lambda_2] \supseteq [\nu_1, \nu_2]$
	and thus away from $B$, the path~$\pth_B$ is not modified during $[\nu_1, \nu_2]$.

	We conclude that if $[\nu_1, \nu_2]$ is a swap interval then $\pth_B$ is not modified during $[\nu_1, \nu_2]$, as claimed. 
	Hence, $\cost{\pth^*_B[\nu_1,\nu_2]} = \cost{\pth_B[\nu_1,\nu_2]}$.  Combining this with 
	\lemref{non-x-sep-push-cost} (for Case~II(a)), we conclude that $\|\plan^*\|\leq \|\plan\|$ for Case~II.
\end{proof}
%-------------------------------------------------------------------
%-------------------------------------------------------------------
\begin{lemma}
	\label{lem:snapping-case-III}
The modifications in Case~III yields a plan~$\plan^*$ that is  optimal.
\end{lemma}
%-------------------------------------------------------------------
\begin{proof}
%The primary pushes on $\pth_B$ in Cases~II(a,b) are of the form $\push(B, I, y(\botQ))$ for some interval $I\supset [\mu_0,\mu_1]$. Hence, there are no conflicts during $[\mu_0,\mu_1]$ after these pushes, and so a secondary push can only occur in an unsafe interval $[\xi_0,\xi_1]$ unequal to $[\mu_0,\mu_1]$.

	First consiser Cases~III(a,b).
	Because $\pth_A[\lambda_1,\lambda_2]$ is pushed to $\botR$ in these cases,
a secondary push can only occur during a swap interval $I$ only if $\robB$ lies
below $\robA$ during $I$. Recall that 
	$[\mu_0,\mu_1]$ is the only swap interval with $B$ below $A$ in Case~III.
	The primary push on $\pth_B$ in Cases~II(a,b), however, are of the form 
	$\push(B, J, y(\botQ))$ for some interval $J\supset [\mu_0,\mu_1]$, so
	there are no conflicts in $\plan^*[\mu_0,\mu_1]$ after the primary push on $\robB$ 
	has been performed.
	Hence, there is no secondary push in Cases~III(a,b) during a swap interval, and 
	thus the modified plan $\plan^*$ is optimal.

	Next, consider Case~III(c). Because $\pth_A[\lambda_1,\lambda_2]$ is pushed to $\topR$ in 
	Case~III(c), a secondary push can only occur during a swap interval $I$ only 
	if $\robB$ lies above $\robA$ during $I$. Recall that 
	$[\mu_2,\mu_3]$ is the only swap interval with $\robB$ below $\robA$ in Case~III. So 
	$I=[\mu_2,\mu_3]$ and the surgery procedure performs $\push(\robB, [\mu_2,\mu_3],y_R^++1)$.
	We now argue that even though the length of the path of $\robB$ might increase because of the secondary push, 
	the cost of the overall plan does not increase.

	Since $[\mu_0,\mu_1]\subseteq [\rho_0,\rho_1]$, there is a time $\mu'\in [\rho_0,\rho_1]$ such that $\yco{\pth_A(\mu')} - \yco{\pth_B(\mu')} \geq 1$.
	Moreover, there is a time $\mu''\in [\mu_2,\mu_3]$ such that 
	$\yco{\pth_B(\mu'')} - \yco{\pth_A(\mu'')} \geq 1$.
    (This is, in fact, true for all $\mu''\in [\mu_2,\mu_3]$, since the robots are $y$-separated
    during the interval.)
	Since $\pth^*_A[\lambda_1,\lambda_2]$ is $xy$-monotone and $\mu''>\mu'$, we have
	\[
    \|\pth_A[\lambda_1,\lambda_2]\| - \|\pth^*_A[\lambda_1,\lambda_2]\| \geq 
		\max \Big\{ 0, \; 2\cdot \big( \yco{\pth_A(\mu')} -\yco{\pth_A(\mu'')} \big) \Big\}.
    \]
	Now consider the difference in length for $B$. The \emph{primary} push (on $\robB$) 
	during $[\rho_1,\rho_2]$ makes
	the path of~$\robB$ shorter by at least $2\cdot (y(\topR)- 1 - \yco{\pth_B(\mu')})$, 
	while the \emph{secondary} push during $[\mu_2,\mu_3]$ increases the length of the path 
	of~$B$ 
	by at most $2\cdot (y(\topR)+ 1 - \yco{\pth_B(\mu'')})$;
	the latter is true because we have argued there is only one secondary push that increases length.
	Putting it all together, we obtain	
	\begin{align*}
		& \|\plan\| - \|\plan^*\| \\
		& \;\; \geq  2\cdot \Big( \big( (\yco{\pth_A(\mu')} - \yco{\pth_A(\mu'')}) \big) 
		    + \big( y(\topR) - 1 -\yco{\pth_B(\mu')} \big)
		    - \big( (y(\topR) + 1 - \yco{\pth_B(\mu'')}) \Big)\\
		& \;\; = 2\cdot \Big( \yco{\pth_A(\mu')} -\yco{\pth_B(\mu')} - 1 + 
			\yco{\pth_B(\mu'')} - \yco{\pth_A(\mu'')} - 1) \Big) \\
		& \;\; \geq 2\cdot \Big (1-1 + 1-1 \Big)  \; =\; 0.
	\end{align*}
	Hence, $\plan^*$ is still optimal.
\end{proof}
%--------------------------------------------------------------------------
Combining Lemmas~\ref{lem:snapping-case-I}--\ref{lem:snapping-case-III}, we obtain the following result.
\begin{lemma}
	\label{lem:optimal}
	The modified plan $\plan^*$ is optimal.
\end{lemma}
%--------------------------------------------------------------------------
Combining Lemmas~\ref{lem:feasible} and~\ref{lem:optimal} with the discussion 
in the beginning of the section, we obtain the following result.
%----------------------------------------------------------------------------
\begin{proposition}
	\label{prop:progress}
Let $\plan$ be a rectilinear, decoupled, alternating $(\bs,\bt)$-plan that 
contains a bad horizontal segment. Then there exists a plan $\plan^*$ that satisfies (P1)--(P5).
\end{proposition}
%----------------------------------------------------------------------------

%----------------------------------------------------------------------------
\section{Hardness of Min-Makespan}\label{sec:hardness}
%---------------------------------------------------------------------------
In this section we prove \thmref{two-robots-hardness},
which states that when the objective is to minimize the makespan
instead of the sum of the path lengths, the motion-planning
problem for two square robots becomes hard.
%---------------------------------------------------------------------------
\hardness*
%---------------------------------------------------------------------------
Our reduction is from \partition---given a set $X=\{x_1,\ldots,x_m\}$ of $m$
integers, decide if there is a partition of $X$ into disjoint subsets $X_A$
and $X_B$ such that $\sum_{x_i\in X_A} x_i = \sum_{x_i\in X_B} x_i$---and
it is similar to the one of Kobayashi and Sommer~\cite[Theorem 14]{kobayashi2010} 
for the shortest edge-disjoint paths problem.

It will be convenient to scale the given instance such that
the input elements sum to~1; this scaling is not necessary, but it will simplify
the presentation. Thus, we define $Y=\{y_1,\ldots,y_m\}$
where $y_i := x_i/\left(\sum_{i=1}^m x_i\right)$, and we ask:
is there a partition of $Y$ into disjoint subsets $Y_A$
and $Y_B$ such that $\sum_{y_i\in Y_A} y_i = \sum_{y_i\in Y_B} y_i = \tfrac{1}{2}$?
We call such a partition \emph{valid}.
\medskip

The idea of the reduction is to build a workspace~$\envir$ that consists
of gadgets $\envir_1,\ldots,\envir_m$, each corresponding to an element of~$Y$ and to choose a parameter $T_{\max} \geq 0$, 
so that both robots must pass through every gadget and there is
a plan with makespan at most $T_{\max}$ if and only if there is a valid partition of~$Y$.
\figref{gadget} shows the gadget\footnote{For simplicity 
we have used zero-width obstacles---that is, line segments---and
passages of exactly width~1 in the construction, 
but such degeneracies could also be avoided.}
$\envir_i$ corresponding to an element~$y_i\in Y$.
%---------------------------------------------------------------------------
\begin{figure}
\centering
\includegraphics{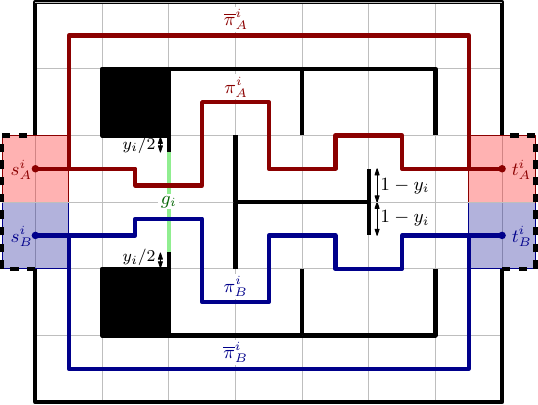}
\caption{The gadget $\envir_i$ for an element~$y_i\in Y$. The thick black segments and
the black square are obstacles. The thin grey lines forming a unit grid 
are not obstacles---they are drawn to show the sizes of the various parts
of the construction. If $i = 1$ then the left side 
of the gadget will be closed off, and if $i = m$ then the right side 
will be closed off, as indicated by the dashed segments. The gate~$g_i$
of $\envir_i$ is shown in light green.}
\label{fig:gadget}
\end{figure}
%---------------------------------------------------------------------------
The points $s_A^i$ and $s_B^i$ are the points where $\robA$ and $\robB$ will
enter the gadget, respectively, and $t_A^i$ and $t_B^i$ are the points where 
$\robA$ and $\robB$ will leave it. Note that the gadget has a vertical 
obstacle segment of length~$\tfrac{1}{2}y_i$ attached to the bottom-right and top-right corner,
respectively, of the two square obstacles. The segment connecting these
two obstacle segments is called the \emph{gate} of the gadget---the gate
is \emph{not} an obstacle---and it is denoted by~$g_i$.
The length of the gate is $2-y_i$, which prevents the two robots from passing
through the gate at the same time. The following observation is easy to verify.
% \ben{Moreover, in any min-makespan plan within a gadget, only one robot travels through any $g_i$. Reaching $(t^i_A, t^i_B)$ from $(s^i_A, s^i_B)$ by sending both robots through $g_i$ requires time 
% at least $11 - y_i + y_i + 1$, whereas sending only one robot through $g_i$ yields a plan requiring not more than time $11$.}
%---------------------------------------------------------------------------

% \ben{I might prefer listing the following as properties, satisfied within each $\envir_i$ and depicted in the figure}
\begin{observation} \label{obs:gadget} $\mbox{}$
\begin{enumerate}[(i)]
\item The path $\pi_A^i$ is the shortest path from $s_A^i$ to $t_A^i$,
      the path $\pi_B^i$ is the shortest path from $s_B^i$ to $t_B^i$, 
      and these paths have length $11-y_i$.
\item The alternative paths $\overline{\pi}_A^i$ from $s_A^i$ to $t_A^i$ 
      and $\overline{\pi}_B^i$ from $s_B^i$ to $t_B^i$ have length $11$.
\item No point on $\pi_A^i$ conflicts with any point on $\overline{\pi}_B^i$, and
      no point on $\overline{\pi}_A^i$ conflicts with any point on $\pi_B^i$.
\end{enumerate}
\end{observation}
%---------------------------------------------------------------------------
The entire workspace~$\envir$ is obtained by concatenating the 
gadgets $\envir_1,\ldots,\envir_m$ so that $s_A^{i+1}$ and $s_B^{i+1}$
coincide with $t_A^{i}$ and $t_B^{i}$ for all $1\leq i<m$, as shown in \figref{reduction}. 
The instance of \minmakespan is completed by setting 
$\bs := (s_A^1,s_B^1)$ and $\bt := (t_A^m,t_B^m)$.
%---------------------------------------------------------------------------
\begin{figure}
\centering
\includegraphics{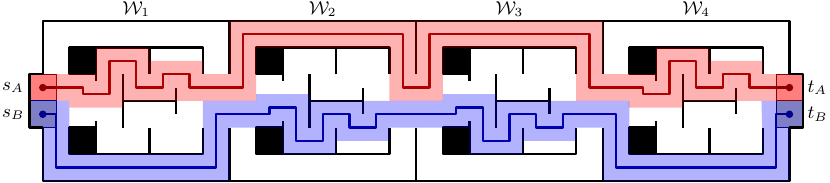}
\caption{The workspace created by our construction for a set $Y = \{y_1, ..., y_4\}$. The paths shown for
the two robots correspond to the partition $Y_A=\{y_1,y_4\}$ and $Y_B=\{y_2,y_3\}$. The black areas are obstacles, and the
red/blue regions are the areas swept by the two robots.}
\label{fig:reduction}
\end{figure}
%---------------------------------------------------------------------------

The construction of $\envir$ can  clearly be carried out in linear time.
The next lemma establishes the correctness of the reduction, thus finishing
the proof of \thmref{two-robots-hardness}.
%---------------------------------------------------------------------------
\begin{lemma}
\label{lem:makespan-iff}
Let $\plan^*$ be a plan of minimum makespan for the created instance~$\envir$.
Then $\pc(\plan^*) \leq 11m-\tfrac{1}{2}$ if and only if the corresponding instance~$Y$
of \partition is valid.
% \ben{Changed to equality; a partition is valid iff $A$ and $B$ reach $t_A, t_B$ simultaneously which is slightly stronger}
% \mdb{But we do not prove equality: strictly speaking, we only show that that path of length $11m-\tfrac{1}{2}$ 
% exists, not that there cannot be a shorter path.}
\end{lemma}
%---------------------------------------------------------------------------
\begin{proof}
\begin{description}
\item[$\Longleftarrow$:] 
    Let $Y_A,Y_B$ be a valid partition of~$Y$. Let $\pi^*_A$ be the
    path from $s^1_A$ to $t^m_A$ that uses $\pi_A^i$ if
    $y_i\in Y_A$ and $\overline{\pi}_A^i$ otherwise. Similarly, let
    $\pi^*_B$ be the path from $s_B^1$ to $t_B^m$ that uses $\pi_B^i$ if
    $y_i\in Y_B$ and $\overline{\pi}_B^i$ otherwise; see \figref{reduction}
    for an example. 
    By Observation~\ref{obs:gadget}(i) and~(ii), and because $Y_A,Y_B$ is a valid
    partition, we have
    \[
    \| \pi^*_A\| = \sum_{y_i\in Y_A} (11-y_i) + \sum_{y_i\in Y_B} 11 
                 =  11m-\tfrac{1}{2}
                 = \sum_{y_i\in Y_A} 11 + \sum_{y_i\in Y_B} (11-y_i)
                 = \| \pi^*_B\|.
    \]
    Moreover, $\robA$ and $\robB$ can traverse their respective paths 
    $\pi^*_A$ and $\pi^*_B$ at speed~1 without colliding with each other,
    by Observation~\ref{obs:gadget}(iii).
    Hence, $\pc(\plan^*) \leq 11m-\tfrac{1}{2}$.
\item[$\Longrightarrow$:] 
    Suppose that $Y$ does not have a valid partition. We must show that this
    implies that $\pc(\plan)>11m-\tfrac{1}{2}$ for any feasible plan~$\plan=(\pi_A,\pi_B)$.
    Define $G_A$ to be the set of gates crossed by~$\pi_A$. More precisely,
    $G_A$ is the set of gates~$g_i$ such that at some moment in time the left 
    edge of $\robA$ is contained in~$g_i$. Define $G_B$ similarly.
    Note that if $g_i\not\in G_A$ then
    the shortest path for $\robA$ through $\envir_i$ is via~$\overline{\pi}_A^i$;
    the analogous statement holds for~$\robB$. 
    We now distinguish two cases.
    \begin{itemize}
    \item \emph{Case I: $G_A\cap G_B=\emptyset$.} \\
           Assume without loss of generality that 
           $\sum_{g_i\in G_A} y_i \leq \sum_{g_i\in G_B} y_i$. Then
           \[
           \sum_{g_i\in G_A} y_i \leq \tfrac{1}{2} \cdot \left( \sum_{g_i\in G_A} y_i + \sum_{g_i\in G_B} y_i \right) \leq \tfrac{1}{2} \cdot  \sum_{i=1}^m y_i = \tfrac{1}{2}.
           \] 
           % \mst{$\sum_{g_i\in G_A} y_i$ is a small abuse of notation. I would prefer $\sum_{i: g_i\in G_A}y_i$. }
           % \ben{I'm good with either.}
           Moreover, we cannot have 
           $\sum_{g_i\in G_A} y_i = \tfrac{1}{2}$, otherwise $Y_A := \{y_i : g_i\in G_A\}$
           and $Y_B := Y\setminus Y_A$ would be a valid partition.
           Hence, $\sum_{g_i\in G_A} y_i < \tfrac{1}{2}$. Given that $\robA$ only
           crosses the gates in~$G_A$, the fastest way for
           $\robA$ to reach $t_A^m$ is by using $\pi_A^i$ if $g_i\in G_A$ and 
           $\overline{\pi}_A^i$ otherwise. Hence,
           \[
           \pc(\plan) \geq \sum_{g_i\in G_A} (11-y_i) + \sum_{g_i\not\in G_A} 11
                        = 11m - \sum_{g_i\in G_A} y_i > 11m - \tfrac{1}{2}.
           \]  
    \item \emph{Case II: $G_A\cap G_B \neq \emptyset$.} \\
           Let $g_i$ be a gate such that $g_i\in G_A\cap G_B$, and assume 
           without loss of generality that $\robA$ crosses $g_i$ before $\robB$ does.
           Let $t^*$ be the first time at which the left edge of $\robA$ is contained
           in~$g_i$. The fastest way for $\robA$ to reach such a position is by taking
           the paths $\pi_A^j$ for all $j<i$ and then going straight to the right from
           $s_A^i$ until it crosses~$g_i$. Hence,
           \[
           t^* \geq \sum_{j=1}^{i-1} (11-y_j)+2\tfrac{1}{2} + \tfrac{1}{2}y_i.
           \]
           The earliest time at which $\robB$ can have fully crossed $g_i$ is
           at time~$t^*+1$. ($\robB$ is able to achieve this if its right edge
           would be contained in $g_i$ at time $t^*$.) After $\robB$ has crossed 
           $g_i$ for the first time, it still has to reach ~$t_B^m$. The fastest way to
           do so, is by first following the part of $\pi_B^i$
           that remains after crossing~$g_i$---this assumes $\robB$ is located on $\pi_B^i$ 
           when it has fully crossed~$g_i$, which is the shortest possible path for $B$---and 
           then taking the paths $\pi_B^j$ for all $j> i$.
           The total length of this path is $8\tfrac{1}{2} - \tfrac{3}{2}y_i + \sum_{j=i+1}^{m} (11-y_j)$.
           Hence,
           \[
           \pc(\robB) \geq t^* + 1 + 8\tfrac{1}{2} - y_i + \sum_{j=i+1}^{m} (11-y_j)
                      \geq 11m + 1 - \sum_{j1}^{m} y_j
                      = 11m.
           \]
    \end{itemize}
    We conclude that $\pc(\plan)>11m-\tfrac{1}{2}$ in both cases, which finishes the proof.
\end{description}
\vspace*{-8mm}
\end{proof}

\section{Conclusion}
%----------------------------------------------------------------------------

We presented an $O(n^4\log{n})$-time algorithm to compute a plan of minimum total length for two unit
squares in a rectilinear environment; this is the first polynomial
algorithm for 2-robot optimal motion planning in a polygonal environment. In contrast,
we showed that minimizing the makespan in the same setting is weakly \nphard. We conclude
with some open problems:
\begin{itemize}
    \item Can the runtime of our algorithm for the \minsum problem be improved to
    $O(n^2\log{n})$? Can our algorithm be extended to $k$ unit squares with $n^{O(k)}$ runtime?
    \item Is the \minsum problem for two unit squares (or disks) in a general polygonal environment 
in $\p$?
    \item Is there a pseudo-polynomial-time algorithm for the \minmakespan problem for two
    unit squares translating in a rectilinear environment? What about an approximation algorithm?
\end{itemize}

\bibliography{references,soda}

%----------------------------------------------------------------------------
\end{document}